\documentclass{article}
\usepackage{spconf,amsmath,graphicx}
\usepackage{algorithm}
\usepackage{algorithmic}
\usepackage{amsthm}

\usepackage{subfigure}
\usepackage{amssymb}
\usepackage{xcolor}
\usepackage{booktabs}
\usepackage{diagbox}
\usepackage{tikz}
\usetikzlibrary{shadings}
\usepackage{amsmath}
\newtheorem{thm}{\bf Theorem}[section]

\title{Knowledge-Free Black-Box Watermark and Ownership Proof for Image Classification Neural Networks}
\name{Fangqi Li, Shilin Wang}
\address{\texttt{ \{solour\_lfq\}@sjtu.edu.cn}\\
School of Electronic Information and Electrical Engineering, Shanghai Jiao Tong University.}
\begin{document}
%
\maketitle
\begin{abstract}
Watermarking has become a plausible candidate for ownership verification and intellectual property protection of deep neural networks.
Regarding image classification neural networks, current watermarking schemes uniformly resort to backdoor triggers.
However, injecting a backdoor into a neural network requires knowledge of the training dataset, which is usually unavailable in the real-world commercialization.
Meanwhile, established watermarking schemes oversight the potential damage of exposed evidence during ownership verification and the watermarking algorithms themselves.
Those concerns decline current watermarking schemes from industrial applications.
To confront these challenges, we propose a knowledge-free black-box watermarking scheme for image classification neural networks.
The image generator obtained from a data-free distillation process is leveraged to stabilize the network's performance during the backdoor injection.
A delicate encoding and verification protocol is designed to ensure the scheme's security against knowledgable adversaries.
We also give a pioneering analysis of the capacity of the watermarking scheme.
Experiment results proved the functionality-preserving capability and security of the proposed watermarking scheme.
\end{abstract}
\section{Introduction}
\label{sec:1}
Since artificial intelligence models, especially Deep Neural Networks (DNN) can be readily deployed as commercial services, it is necessary to regulate them as Intellectual Properties (IP), after which their accountability can be addressed.

To achieve Ownership Verification (OV) and protect DNNs as IPs, various watermarking schemes have been proposed for distinct network architectures, such as U-Nets~\cite{zhang2021deep}, image segmentation networks~\cite{ours}, natural language processing networks~\cite{9647115}, generative adversarial networks~\cite{ong2021protecting}, etc.
There are diversified watermarking schemes for image classification DNNs~\cite{li2019persistent,zhang2018protecting,zhu2020secure}, which have been playing a central role in industrial applications including identity recognition, autonomous vehicle, human-machine interaction, etc.

So far, the most critical issue that prevents watermarking schemes from the application is the absence of knowledge regulation.
It is uniformly assumed that the agent who runs the watermarking scheme has unlimited access to the training dataset, which is contradictory to real-world model purchasing and industrial pipeline.
Moreover, the security of most watermarking schemes depends on the secrecy of algorithms. 
An adversary knowing the watermarking algorithm can trivially compromise the ownership proof, making such watermarks of little pragmatic value.
Finally, current watermarking schemes usually focus on identifying a suspicious DNN.
Once such an ownership proof has to be presented to any third party, especially when it has to be repeated multiple times, the exposed knowledge can be utilized to escape IP regulation.

To address these concerns regarding knowledge management during DNN watermarking, we design a knowledge-free watermarking scheme for image classification DNNs.
We adopt the data-free distillation to generate substitute samples for ordinary training samples during watermark injection.
An delicate encoding process generates triggers that carry the ownership information while remain persistent against adversarial tuning with the knowledge of the watermark.
Furthermore, one-way hash functions are incorporated into the trigger generation and the OV procedure, so the owner's identity information would not be erased even if the ownership proof is eavesdropped on.
The contributions of this paper are three folded:
\begin{itemize}
\item We introduce the scenario of knowledge-free watermarking for DNN and raise corresponding security requirements.
\item By combining the data-free distillation and cryptological primitives, a knowledge-free watermarking scheme is proposed for image classification DNN.
We also analyzed the corresponding watermarking capacity of this scheme, the first analytical bound of such a metric.
\item Extensive experiments on numerous datasets and network architectures verified that our method can provide reliable protection of DNNs without accessing the training dataset while resists knowledgeable and eavesdropping adversaries.
\end{itemize}

\section{Preliminaries and Challenges}
\label{sec:2}
\subsection{Watermarks for Image Classification DNN}
To establish the intellectual property regulation of DNNs, the owner of a DNN embeds its identity information, denoted as \texttt{key}, into its product to produce a watermarked DNN $M_{\text{WM}}$.
Such embedding can be done simultaneously with the normal training process or after it~\cite{zhang2021deep}.
After distributing its product, the owner can prove its ownership over a suspicious DNN model $M$ by retrieving \texttt{key} from it using an ownership examination module \texttt{verify} defined by the watermarking algorithm.
To complete the proof, the owner generates an evidence information \texttt{e} from \texttt{key} and presents it to the arbiter, who can then independently compute $\texttt{verify}(M,\texttt{e})$ and take the outcome as the ownership justification. 

The basic accuracy and unambiguity of a watermarking scheme requires that for any legal \texttt{key} and the corresponding $\texttt{e}\leftarrow\texttt{key}$,
\begin{equation}
\label{equation:accuracy}
\texttt{verify}(M_{\text{WM}},\texttt{e})=\texttt{Pass},
\end{equation}
and
\begin{equation}
\label{equation:unambiguity}
\texttt{verify}(M_{\text{WM}},\texttt{e}')=\texttt{Fail},
\end{equation}
where $\texttt{e}'$ is an incorrect piece of evidence. 
Conditions Eq.~\eqref{equation:accuracy} and Eq.~\eqref{equation:unambiguity} have to hold with asymptotical probability one. 

Considering the black-box OV scenario for image classification DNNs, where the owner or the arbiter can only interact with the suspicious DNN as an API, watermarking schemes uniformly utilize the backdoor, i.e., a set of triggers with artificially assigned labels.
Predictions for these triggers differentiate the watermarked DNN from an unprotected model and justify the possession. 
To reveal the ownership, the evidence \texttt{e} can be reduced to the mapping between triggers and their labels, which is missing in an unwatermarked DNN.
Triggers can be specific stamps~\cite{adi2018turning,zhang2018protecting}, white noises~\cite{zhu2020secure}, images with oversized pixels~\cite{li2019persistent}, etc. 
By adversarially tuning the triggers, they can acquire penetrative ability against filtering~\cite{oursicip} or be adopted in federated learning~\cite{waffle}.

Compared with white-box watermarking schemes, black-box watermarking schemes have several advantages. 
(i) In real-world applications, the white-box access to the pirated DNN is usually unavailable, making black-box schemes more practical.
(ii) The black-box watermark remains valid against functionality-preserving attacks~\cite{oursicassp}, which can trivially spoil white-box watermarks within parameters.

However, the backdoor injection has to be conducted along with the normal training process, otherwise, the watermarked model's performance on normal inputs might be sacrificed.
Moreover, these watermarking schemes have to submit the trigger set for ownership proof, resulting in extra security risks.

\subsection{Knowledge Concerns in Backdoor Injection}
\label{section:2.2}
The agent who demands the watermark service might not know the training dataset.
As the instance shown in Fig.~\ref{figure:threat1}, the agent inside an industrial pipeline who is in charge of the IP regulation might be completely oblivious of the training process, let alone the training data.
For another example, after purchasing a DNN model, the purchaser wants to write its identity into the model. 
It is neither feasible to leave this task to the seller, nor proper for the seller to provide the training dataset. 
Unfortunately, most backdoor-based watermarking schemes require the training data, since tuning a DNN on the trigger set could bring irreversible damage to the model's performance on normal inputs.
In DNN commercialization, it is necessary that: \textbf{the watermarking scheme functions independently from the training dataset. (C1)}

\begin{figure}[!t]
\centering
\includegraphics[width=8cm]{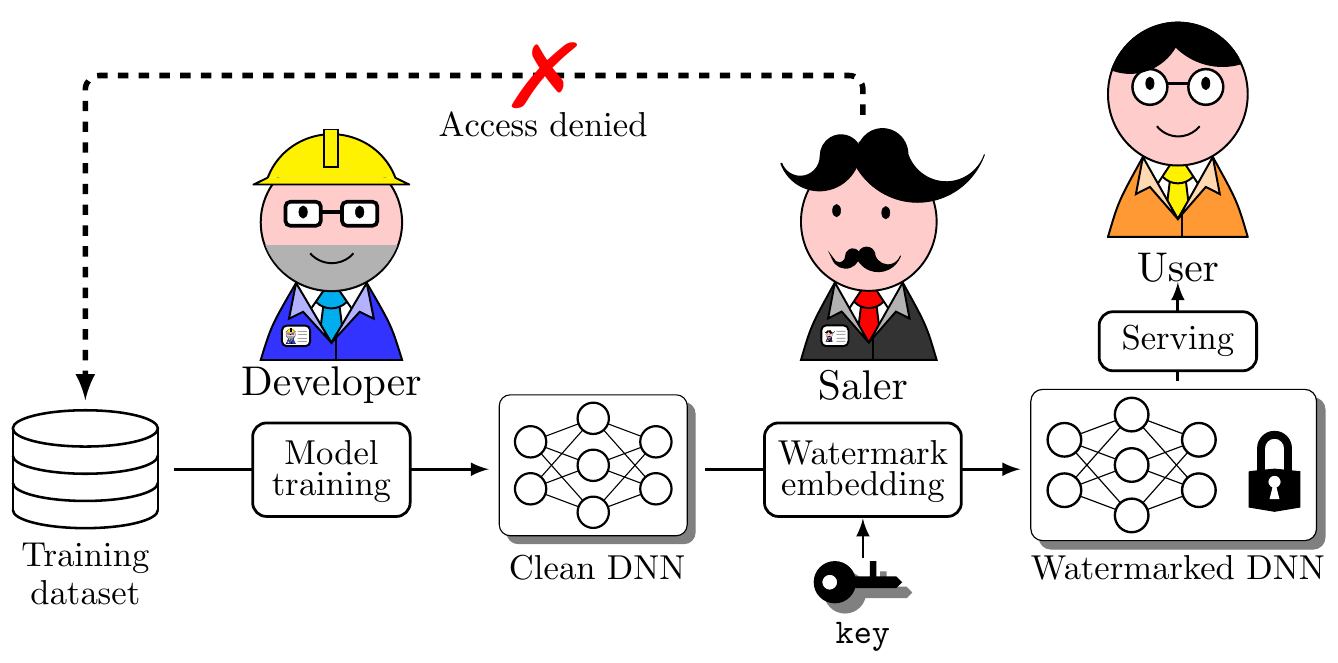}
\caption{DNN model packaging during a commercial pipeline.}
\label{figure:threat1}
\end{figure}


\begin{figure*}[htbp]
\centering
\subfigure[The first ownership verification.]{
\includegraphics[width=5cm]{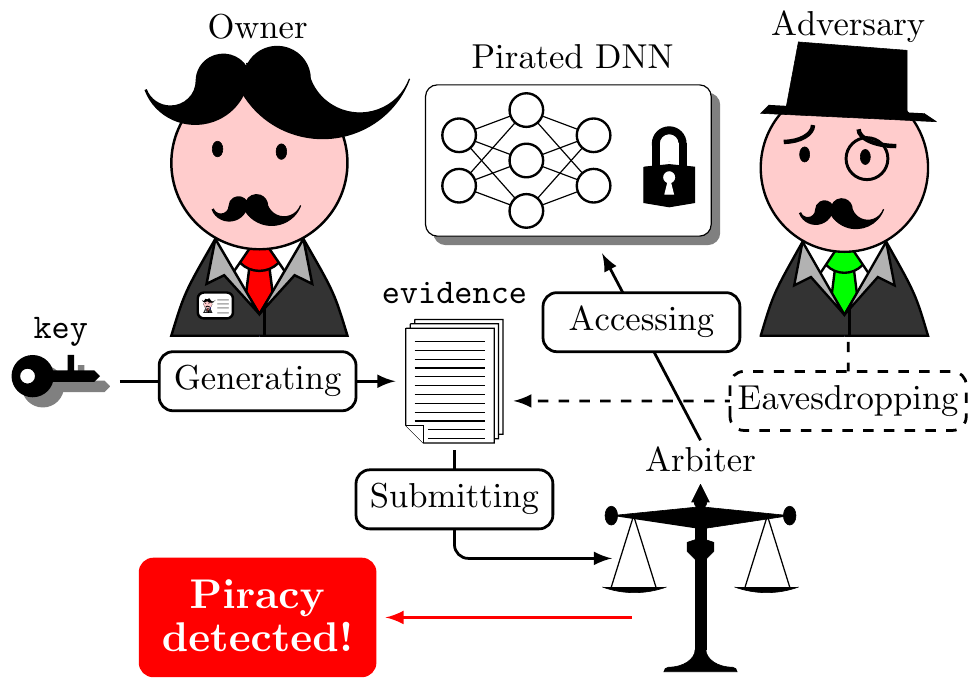}
}\subfigure[The adverarial tuning using exposed evidence.]{
\includegraphics[width=4.5cm]{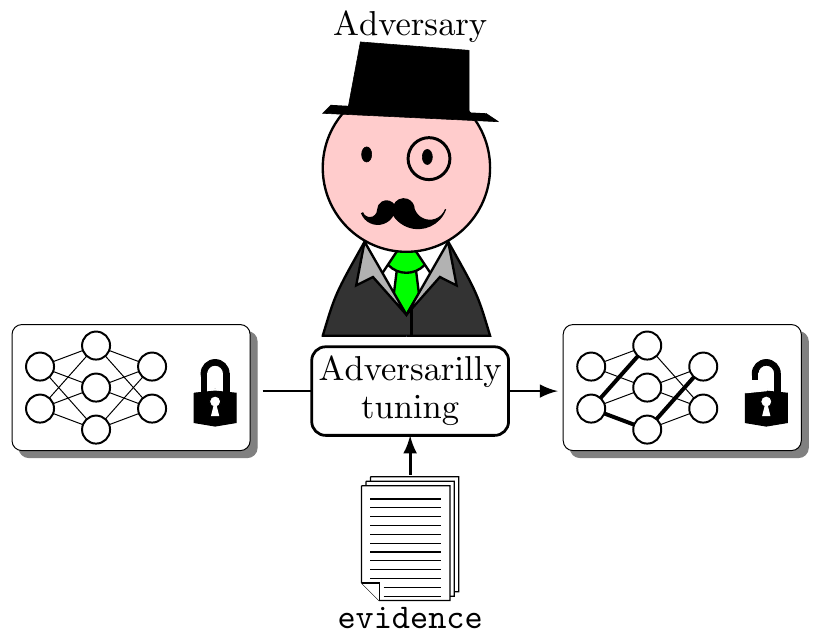}
}\subfigure[The second ownership verification.]{
\includegraphics[width=5cm]{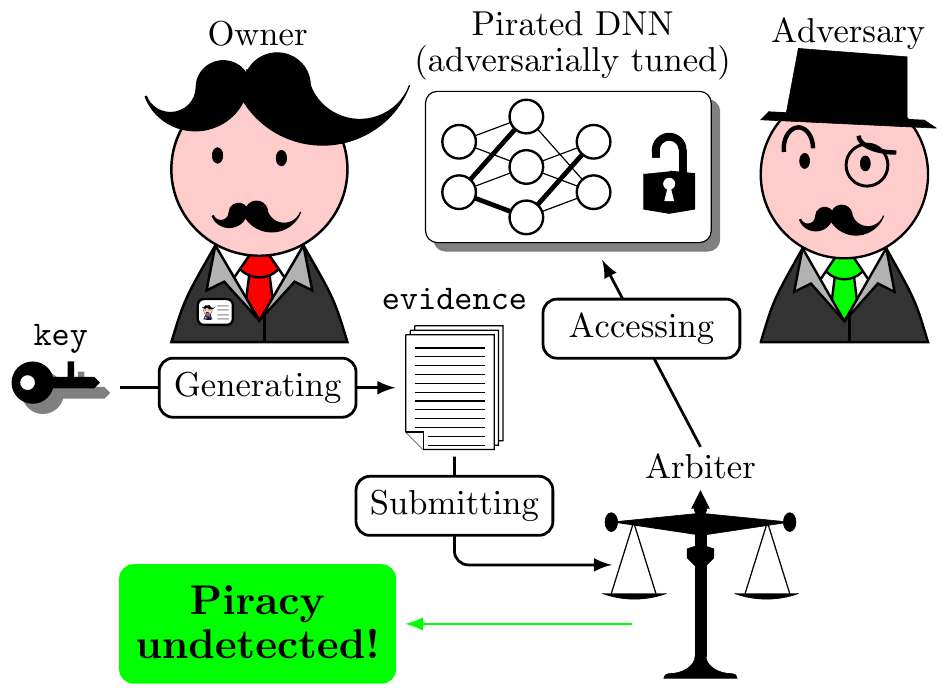}
}
\caption{The exposed evidence during ownership verification might be utilized to unvalidate the second proof.}
\label{figure:threat2}
\end{figure*}

\subsection{Knowledge Concerns in Ownership Verification}
\label{section:2.3}
Regulating DNN IP does not only involve evaluating whether two DNNs are identical~\cite{zjsp}, which can hardly convince any third party on the authentic ownership.
In practice, the owner has to present its evidence to an arbiter, e.g., the court or the host of competition, to legalize its liability.
During this ownership proof process, some evidence has to be exposed.
Consequently, it is necessary that:
\begin{itemize}
\item The evidence provided by the owner is mathematically unforgeable, this property has been formulated as the unambiguity in Eq.~\eqref{equation:unambiguity}.
\item \textbf{The knowledge of the watermarking scheme cannot be utilized to invalidate the watermark. (C2)}
\item \textbf{The exposed evidence cannot be utilized to erase the watermark from the DNN and escape the IP regulation. (C3)}
\end{itemize}
The requirement \textbf{(C2)} builds the security of IP upon the secrecy of \texttt{key} rather than that of the algorithm, following the principles of modern computer security.
Consequently, all zero-bit watermarks~\cite{zhang2018protecting} and those whose mechanisms can be trivially reversed~\cite{guan2020reversible} are impractical.
The motivation behind requirement \textbf{(C3)} is that the ownership proof shall be presented to a semi-honest third party and there remains a risk of eavesdropping as shown in Fig.~\ref{figure:threat2}.
Given both the watermarking algorithm and the identity evidence, an adversary's damage can be much more powerful.
The white-box watermark that modifies specific parameters can be trivially spoiled given their locations.
The black-box watermark can also be invalidated by blocking trigger pattern inputs.

These two additional challenges along with corresponding encoding methods and security analyses have never been extensively examined by current watermarking schemes.

\subsection{The Data-Free Distillation}
The unavailability of the training dataset haunts not only the IP manager of a DNN but also model extractors.
To unconditionally distill knowledge from a given DNN, Data-Free Distillation (DFD) methods have been proposed~\cite{fang2019data,micaelli2019zero}.
The motivation is that the decision boundaries of DNN models share similar topological properties due to their architectures.
\begin{figure}[!t]
\centering
\includegraphics[width=6.5cm]{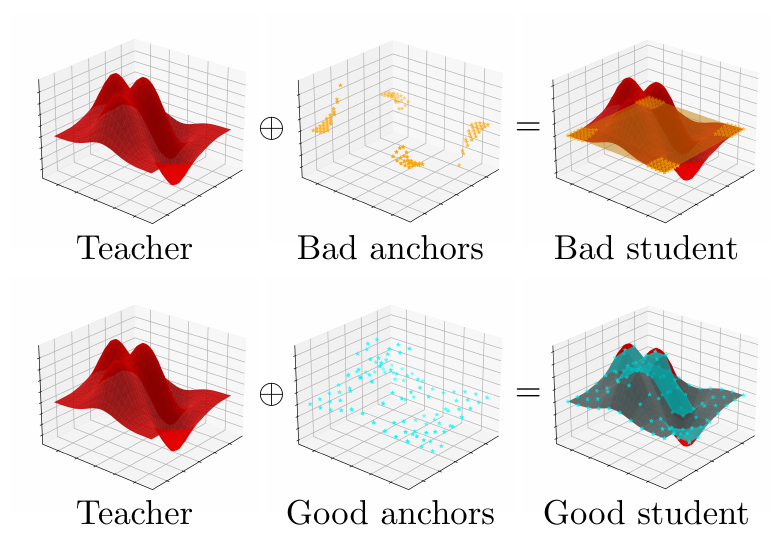}
\caption{Data-free distillation, an illustration.}
\label{figure:dfd1}
\end{figure}
Therefore, the student model only has to approximate the teacher model on a collection of \textbf{anchors}, where their outputs are mostly differentiated as shown by Fig.~\ref{figure:dfd1}.
Anchors are produced by a generator $G$. 
The training of the student model and $G$ iterates adversarially as demonstrated in Fig.~\ref{figure:dfd2}.
\begin{figure}[!t]
\centering
\includegraphics[width=6.5cm]{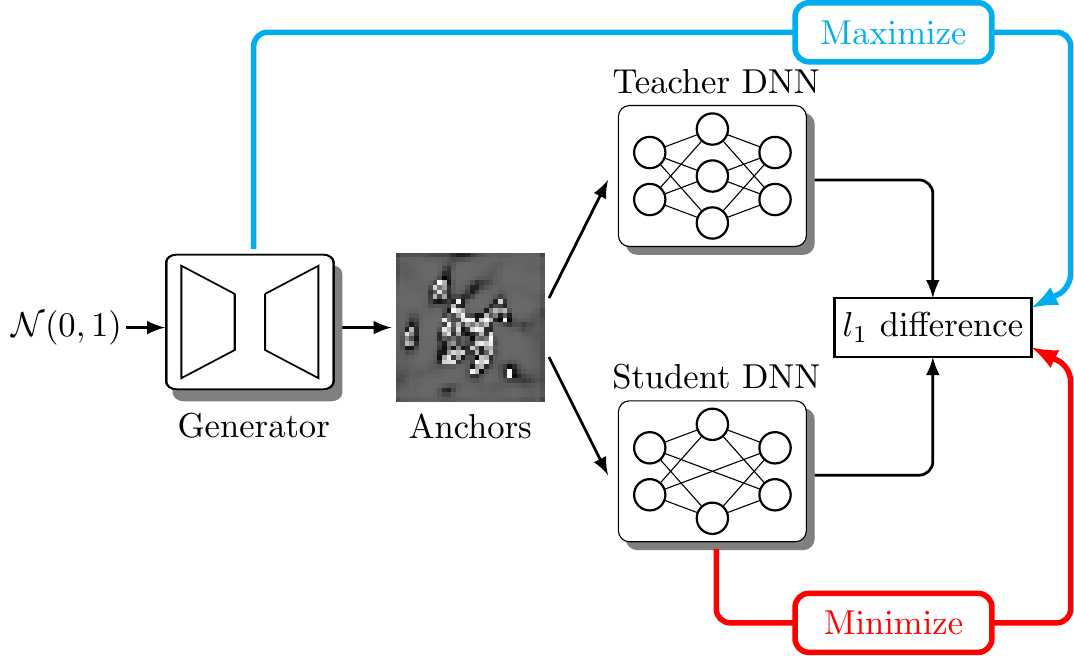}
\caption{The workflow of data-free distillation.}
\label{figure:dfd2}
\end{figure}
The generated anchors can be taken as a substitute for the training dataset for performance preservation.

\section{The Proposed Method}
\label{sec:3}
\begin{figure*}[htbp]
\centering
\includegraphics[width=14cm]{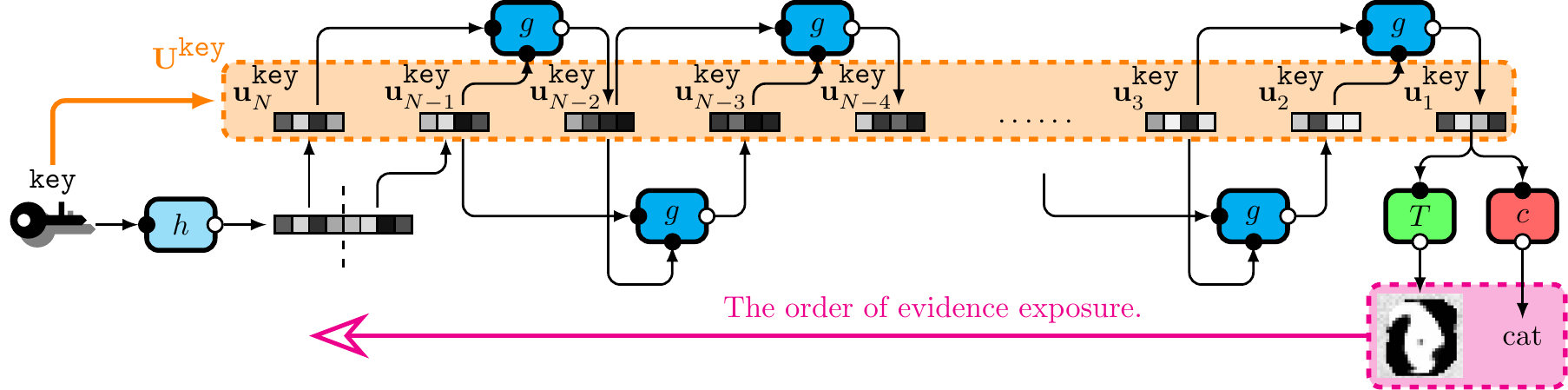}
\caption{The encoding process of the identity information and the ownership verification order.}
\label{figure:encode}
\end{figure*}
\subsection{The Motivation}
It is possible to distill a DNN from the anchors, so they can be taken as proper substitutes for the normal training dataset.
The IP manager without knowledge of the original training data can use these anchors during trigger injection to stabilize the DNN product's performance.
To ensure that the exposure of a subset of triggers would not leak information about unrevealed triggers, we leverage the one-wayness of cryptological secure hash functions to generate privacy-preserving triggers.
Finally, we reuse the generated anchors to fuzzify the triggers' pattern to increase their robustness against adversarial filtering and model tuning.

\subsection{Identity Information Encoding}
An arbitrary digital information \texttt{key} is firstly mapped into a sequence of $N$ codes
\begin{equation}
\nonumber
\textbf{U}^{\texttt{key}}=(\textbf{u}_{1}^{\texttt{key}},\textbf{u}_{2}^{\texttt{key}},\cdots,\textbf{u}_{N}^{\texttt{key}}),
\end{equation}
where $\textbf{u}_{n}^{\texttt{key}}\in\mathcal{U}$ is a binary string, and $\mathcal{U}$ is the space of all legal codes.
Each code can be uniquely mapped into a trigger image $T(\textbf{u}_{n}^{\texttt{key}})$ and a label $c(\textbf{u}_{n}^{\texttt{key}})$, these two mappings are publicly available.
In particular, $c(\cdot)$ is a pseudorandom function whose range is identical to that of the classification task.
The collection of codes $\textbf{U}^{\texttt{key}}$ is built as a one-way sequence reversely using two one-way hash functions: $h(\cdot)$ with arbitrary input length and $g(\cdot)$ with reduce factor 2.
Initially we have
\begin{equation}
\nonumber
(\textbf{u}_{N-1}^{\texttt{key}},\textbf{u}_{N}^{\texttt{key}})=h(\texttt{key}),
\end{equation}
then $\forall n \in\left\{1,2,\cdots,N-2\right\}$
\begin{equation}
\nonumber
\textbf{u}_{n}^{\texttt{key}}=g(\textbf{u}_{n+1}^{\texttt{key}},\textbf{u}_{n+2}^{\texttt{key}}).
\end{equation}
The labelled trigger set is
\begin{equation}
\nonumber
\mathcal{B}^{\texttt{key}}=\left\{(T(\textbf{u}_{n}^{\texttt{key}}),c(\textbf{u}_{n}^{\texttt{key}})) \right\}_{n=1}^{N}.
\end{equation}

Having trained the DNN model on the labeled trigger set, OV is done by exposing the codes sequentially.
The owner submits the following $K\times 3$ array as the evidence:
\begin{equation}
\label{equation:evidence}
\texttt{e}=
\begin{pmatrix}
\textbf{u}_{K'+1}^{\texttt{key}} & T\left(\textbf{u}_{K'+1}^{\texttt{key}}\right) & c\left(\textbf{u}_{K'+1}^{\texttt{key}}\right) \\
\textbf{u}_{K'+2}^{\texttt{key}} & T\left(\textbf{u}_{K'+2}^{\texttt{key}}\right) & c\left(\textbf{u}_{K'+2}^{\texttt{key}}\right) \\
\cdots & \cdots & \cdots \\
\textbf{u}_{K'+K}^{\texttt{key}} & T\left(\textbf{u}_{K'+K}^{\texttt{key}}\right) & c\left(\textbf{u}_{K'+K}^{\texttt{key}}\right)
\end{pmatrix},
\end{equation}
with $K\leq K+K'\leq N$, for the $t$-th time of OV, $K'=t-1$.
Having received the evidence, the arbiter examines whether they satisfy the relationships defined by $T(\cdot)$, $c(\cdot)$, $h(\cdot)$, $g(\cdot)$, and the suspicious model $M$. 
The verifier proceeds as Algo.~\ref{algorithm:OV}, where $\Phi(\cdot)$ is the Gaussian cumulative distribution function.
The entire process is visualized in Fig.~\ref{figure:encode}.

\begin{algorithm}[!t]
\caption{The OV module \texttt{verify}.
$T(\cdot)$, $c(\cdot)$, $g(\cdot)$, and the number of class $C$ are given.}
\label{algorithm:OV}
\begin{algorithmic}[1]
\REQUIRE The suspicious model $M$, evidence \texttt{e}, and the sensitivity threshold $\tau$.
\ENSURE The verification result.
\STATE $\text{acc}=0$;
\FOR {$k=1$ to $K$}
\IF {
$\left\{
\begin{aligned}
T(\texttt{e}_{k,1})=\texttt{e}_{k,2} &\textbf{ and } \\
c(\texttt{e}_{k,1})=\texttt{e}_{k,3} &\textbf{ and } \\
M(\texttt{e}_{k,2})=\texttt{e}_{k,3} &\textbf{ and }\\
m\geq K-2 &\textbf{ or }\texttt{e}_{k,1}=g(\texttt{e}_{k+1,1},\texttt{e}_{k+2,1})
\end{aligned}\right.
$
}
\STATE $++\texttt{acc}$;
\ENDIF
\ENDFOR
\STATE $\hat{\mu}=\frac{\text{acc}}{K}$;
\STATE $\hat{\sigma}=\sqrt{\frac{\hat{\mu}(1-\hat{\mu})}{K}}$;
\IF {$\Phi\left(\frac{\frac{1}{C}-\hat{\mu}}{\hat{\sigma}} \right)\leq \tau$}
\STATE Return \texttt{Pass};
\ELSE
\STATE Return \texttt{Fail};
\ENDIF
\end{algorithmic}
\label{algorithm:OV}
\end{algorithm}

The security of the proposed scheme is analyzed in the following theorems.
\begin{thm}
\textbf{(Accuracy)}
The legal owner whose trigger's classification accuracy is statistically higher than random guessing can pass the ownership examination with confidence level $(1-\tau)$, satisfying the accuracy defined by Eq.~\eqref{equation:accuracy}.
\end{thm}
\begin{proof}
This is evident from the seventh to the ninth line within Algo.~\ref{algorithm:OV}, which forms a statistical hypothesis with confidence level $(1-\tau)$ on the event: the accuracy of the submitted evidence is higher than the random guess.
Consequently, a larger $K$ with a smaller $\tau$ results in a more convincing ownership proof.
\end{proof}

\begin{thm}
\textbf{(Unambiguity)}
An adversary who has never tuned the watermarked DNN fails the OV with a probability no less than $(1-\tau)$, satisfying the accuracy defined by Eq.~\eqref{equation:unambiguity}.
\end{thm}
\begin{proof}
It is easy to construct the code sequence by using the public hash functions.
However, the third condition within the third line in Algo.~\ref{algorithm:OV} holds for approximately $\frac{1}{C}$ since $c(\cdot)$ is a pseudorandom function.
Therefore, an adversary cannot pass the test without tuning the DNN with its triggers and risking its performance.
\end{proof}

\begin{thm}
\label{thm:3}
\textbf{(Identity-Preserving)}
An adversary obtaining \texttt{e} cannot infer any remained evidence, i.e., $\textbf{u}_{K'+K+1}^{\texttt{key}},\cdots,\textbf{u}_{N}^{\texttt{key}}$ and spoil further ownership proof.
\end{thm}
\begin{proof}
To deduce $\textbf{u}^{\texttt{key}}_{K'+K+1}$ from $\texttt{e}$, the adversary has to inverse the hash function $g$
\begin{equation}
\nonumber
\textbf{u}_{K'+K-1}^{\texttt{key}}=g(\textbf{u}_{K'+K}^{\texttt{key}},\textbf{u}_{K'+K+1}^{\texttt{key}})
\end{equation}
with a fixed output and partial input, which is contradictory to its one-wayness, so it cannot shuffle its label and adaptively compromise further ownership tracing.
\end{proof}

The master identity information \texttt{key} need not to be exposed during the entire ownership proof.
In cases where \texttt{e} is eavesdropped on and is claimed to be the adversary's identity evidence, the owner only has to increase $K'$ and present further evidence, i.e., $\textbf{u}^{\texttt{key}}_{K'+K+1},\textbf{u}^{\texttt{key}}_{K'+K+2}\cdots$, to claim the ownership according to Theorem~\ref{thm:3}.
This setting can resist no less than $(N-K)$ eavesdropping attacks, by adopting $N\gg K$, the challenge \textbf{(C2)} raised in Sec.~\ref{section:2.3} is solved. 
Notice that renaming labels has no influence to the integrity of OV since both the owner and the arbiter can easily reconstruct the original mapping by comparing triggers with the same predicted label. 

\subsection{Post-Trigger Generation and Injection}
To enhance the persistency of the watermark, it is necessary to choose a proper image mapping function $T(\cdot)$.
Otherwise, an adversary can trivially filter them.
To reduce this risk, it is desirable that the triggers slightly deviate from the original distribution defined by $T(\cdot)$ so adversarial filtering is either impossible or has to sacrifice the DNN's performance, and tuning or pruning with insufficient knowledge cannot invalidate them.
As the normal training dataset is unavailable, we resort to the DNN itself and the DFD image generator to produce adaptive perturbation.
The triggers $T(\textbf{u}_{n}^{\texttt{key}})$ are tuned into \textbf{post-triggers} that:
(i) Pixelwisely similar to triggers.
(ii) Have labels invariable against fine-tuning with the normal training dataset (anchors) and adversarial tuning.
This intuition is demonstrated in Fig.~\ref{figure:advencode}.
\begin{figure}[!t]
\centering
\includegraphics[width=7cm]{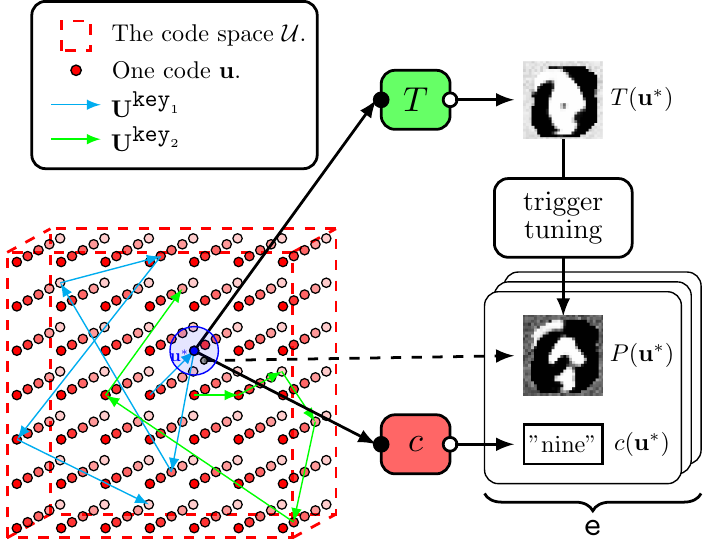}
\caption{The evidence generation process}
\label{figure:advencode}
\end{figure}

Formally, the post-trigger $P(\textbf{u}_{n}^{\texttt{key}})$ corresponding to the code $\textbf{u}_{n}^{\texttt{key}}$ minimizes the following loss function:
\begin{equation}
\label{equation:advloss}
\begin{aligned}
\mathcal{L}(P(\textbf{u}_{n}^{\texttt{key}}))=&\text{CE}(c(\textbf{u}_{n}^{\texttt{key}}),\hat{M}(P(\textbf{u}_{n}^{\texttt{key}})))\\
&+\lambda_{1}\cdot \|P(\textbf{u}_{n}^{\texttt{key}})-T(\textbf{u}_{n}^{\texttt{key}}) \|_{2}^{2},
\end{aligned}
\end{equation}
where $\text{CE}(\cdot,\cdot)$ is the cross-entropy loss, $\hat{M}$ is the DNN fine-tuned with anchors (as a substitute of the training dataset) and current post-triggers.
The target of Eq.~\eqref{equation:advloss} is to find a post-trigger that is numerically tractable from code $\textbf{u}_{n}^{\texttt{key}}$ while preserves the robustness against data-dependent fine-tuning.
By searching for patterns whose labels remain invariable against fine-tuning, we end up with post-triggers that are indistinguishable from normal images from the DNN's perspective.
Since the tuned DNN $\hat{M}$ on the r.h.s. of Eq.~\eqref{equation:advloss} implicitly depends on $P(\textbf{u}_{n}^{\texttt{key}})$, the optimization has to be done in a two-step manner, where $\hat{M}$ and $P(\textbf{u}_{n}^{\texttt{key}})$ are updated alternately.
Details are given in Algo.~\ref{algorithm:triggergen}.
\begin{algorithm}[!t]
\caption{The post-trigger generation process.}
\label{algorithm:OV}
\begin{algorithmic}[1]
\REQUIRE The clean model $M$, the trigger $T(\textbf{u}_{n}^{\texttt{key}})$ and its label $c(\textbf{u}_{n}^{\texttt{key}})$, hyperparamter $\lambda_{1}$, rounds of iteration $Q$, the number of anchors $R$, and the DFD algorithm.
\ENSURE The post-trigger $P(\textbf{u}_{n}^{\texttt{key}})$.
\STATE Run the DFD algorithm for $M$ to obtain the anchor generator $G$.
\STATE $P_{1}=T(\textbf{u}_{n}^{\texttt{key}})$;
\FOR {$q=1$ to $Q$}
\STATE \textbf{\# Step 1: optimizing the r.h.s. of Eq.~\eqref{equation:advloss}.}
\STATE Fine-tune $M$ on $(P_{q},c(\textbf{u}_{n}^{\texttt{key}}))$ to obtain $\tilde{M}_{q}$;
\STATE Generate anchors $\mathcal{A}_{q}=\left\{G(z_{q,r}),M(G(z_{q,r})) \right\}_{r=1}^{R}$, where each $z_{q,r}\leftarrow\mathcal{N}(0,1)$.
\STATE Fine-tune $\tilde{M}_{q}$ on $\mathcal{A}_{q}$ to obtain $\hat{M}_{q}$;
\STATE \textbf{\# Step 2: optimizing the l.h.s. of Eq.~\eqref{equation:advloss}.}
\STATE Obtain $P_{q+1}$ w.r.t. $\hat{M}_{q}$ by minimizing Eq.~\eqref{equation:advloss}.
\ENDFOR
\STATE Return $P_{Q+1}$;
\end{algorithmic}
\label{algorithm:triggergen}
\end{algorithm}

Post-triggers are injected simultaneously with anchors (as substitutes for the normal training dataset) to preserve the watermarked DNN's performance as a solution to challenge \textbf{(C1)} in Sec.~\ref{section:2.2}.
Concretely, the loss function for post-trigger injection is
\begin{equation}
\label{equation:lossembed}
\begin{aligned}
\mathcal{L}(M_{\text{WM}})=&\sum_{n=1}^{N}\text{CE}(M_{\text{WM}}(T(\textbf{u}_{n}^{\texttt{key}})),c(\textbf{u}_{n}^{\texttt{key}}))\\
+\lambda_{2}\cdot&\sum_{\begin{aligned}s&=1,\\z_{s}\leftarrow&\mathcal{N}(0,1)\end{aligned}}^{S}\|M_{\text{WM}}(G(z_{s}))-M_{\text{clean}}(G(z_{s}))\|_{1},
\end{aligned}
\end{equation}
in which the parameters of $M_{\text{WM}}$ are initialized as those of the clean model $M_{\text{clean}}$'s and $S$ is the number of anchors generated as the substitute of the training dataset.
The loss function for the post-triggers' labels is the normal cross-entropy for classification, while the second term in Eq.~\eqref{equation:lossembed} is the $l_{1}$ loss of the output logits as suggested by~\cite{fang2019data} for functionality-preserving.

The current evidence $\texttt{e}$ deviates from that defined in Eq.~\eqref{equation:evidence} by replacing all triggers with corresponding post-triggers.
As the current post-triggers are no longer the exact output of the image encoder $T(\cdot)$, the first examination condition within the third line in Algo.~\ref{algorithm:OV} is modified into the fuzzy version:
\begin{equation}
\label{equation:fuzzyOV}
\|T(\texttt{e}_{k,1})-\texttt{e}_{k,2} \|_{1}\leq\epsilon,
\end{equation}
in which $\epsilon$ is the sensitivity hyperparameter.

Another advantage of using post-triggers is that the adversary acquiring $T(\cdot)$ cannot trivially distinguish ownership probing from normal inputs.
Because post-triggers can be seen as the hybrid of randomized triggers and adversarial images, their distribution lies on a more complex manifold intractable for the adversary.
This property is especially important when $T(\cdot)$ is reversible.
Post-triggers constitute a particial solution to the challenge \textbf{(C3)} raised in Sec.~\ref{section:2.3}.

\subsection{The Watermarking Capacity}
The capacity w.r.t. a DNN and a watermarking scheme measures the amount of identity information that can be embedded into the watermarked model~\cite{oursijcai}, or the upper bound of legal IP managers in federated learning~\cite{waffle}. 
Such quantity can hardly be calculated for established watermarking schemes, yet an analytic bound is tractable in our setting. 
Recall that \texttt{key} is only involved in generating the tail of the triggers' code sequence.
Therefore, we are interested in the number of \texttt{key}s (where each \texttt{key} is correlated with $N$ post-triggers) that can be correctly injected and retrieved.

Denote the maximal number of post-triggers that can be injected into the DNN (without reducing its performance lower than an unacceptable threshold $\gamma$) as $\hat{N}(\gamma)$, then $\frac{\hat{N}(\gamma)}{N}$ is an upper bound of the watermarking capacity.

Since OV takes a fuzzy comparison of the post-triggers, it is expected that the post-trigger of a specific code $\textbf{u}$ is indistinguishable from that generated from its neighbor, e.g., $\tilde{\textbf{u}}\in\mathcal{U}$.
This is implied by Eq.~\eqref{equation:fuzzyOV} and the continuity of $T(\cdot)$.
Consequently, the difference in $c(\textbf{u})$ and $c(\tilde{\textbf{u}})$ results in a confusion during the ownership proof, such event is defined as a \emph{collision}.
Denote the size of such confusion neighbor within $\mathcal{U}$ as $S(\epsilon)$, i.e., for any $\textbf{u}\in\mathcal{U}$, there exist at most $S(\epsilon)$ codes (they constitute the light blue sphere in Fig.~\ref{figure:advencode}) whose corresponding post-triggers are with a pixelwise sphere of radius $\epsilon$ centered at $T(\textbf{u})$.
The probability that two independent codes collide with each other is therefore $\frac{S(\epsilon)}{|\mathcal{U}|}\times \frac{C-1}{C}$.
Such collision would invalidate at least one of the codes, and the accumulation of such collisions would finally invalidate the ownership evidence.

The event that the $J$-th \texttt{key} cannot be correctly embedded is tantamount to: no less than $\frac{NC}{C-1}$ codes derived from this key are found to collide with established triggers.
The number of collisions follows a binomial distribution, which can be reduced to a Gaussian $\mathcal{N}(\mu,\sigma^{2})$ with
\begin{equation}
\label{equation:mus}
\begin{aligned}
\mu(J)&=\frac{J N^{2} (C-1)\cdot S(\epsilon)}{|\mathcal{U}| C},\\
\sigma^{2}(J)&=\frac{J N^{2} (C-1)\cdot S(\epsilon)}{|\mathcal{U}| C}\left(1-\frac{J N (C-1)\cdot S(\epsilon)}{|\mathcal{U}| C}\right),
\end{aligned}
\end{equation}
when $N$ is large.
So the probability of the invalidation event can be approximated as
\begin{equation}
\nonumber
P_{\text{Fail}}(J)=\Phi\left(\frac{\mu(J)-\frac{N C}{C-1}}{\sigma(J)}\right),
\end{equation}
where the dependence on $J$ is through Eq.~\eqref{equation:mus}.
Consequently, the probability that the first $J$ \texttt{key}s can be correctly embedded into the DNN is:
\begin{equation}
\nonumber
\begin{aligned}
P_{\text{Success}}(J)&=\prod_{j=1}^{J}(1-P_{\text{Fail}}(j))\\
&=\prod_{j=1}^{J}\left(1-\Phi\left(\frac{\mu(j)-\frac{N C}{C-1}}{\sigma(j)}\right)\right).
\end{aligned}
\end{equation}
Compactly, to ensure that all \texttt{key}s can be embedded into the DNN product with probability no less than $\zeta$ and its performance remains above $\gamma$, the watermarking capacity is upper bounded by:
\begin{equation}
\label{equation:bound}
\min\left\{\frac{\hat{N}(\gamma)}{N},\arg\max_{J}\left\{\prod_{j=1}^{J}\left(1-\Phi\left(\frac{\mu(j)-\frac{N C}{C-1}}{\sigma(j)}\right)\right) \geq\zeta \right\} \right\}.
\end{equation}

\section{Experiments and Discussions}
\label{sec:4}
\subsection{Settings}
Four image classification tasks were adopted for experiments: MNIST~\cite{deng2012mnist}, CIFAR10, CIFAR100~\cite{krizhevsky2009learning}, and Caltech101~\cite{fei2006one}.
Four candidate trigger encoders were considered: the Gaussian noise~\cite{zhu2020secure}, a random image generator network~\cite{chen2021learning}, the DFD generator itself~\cite{fang2019data,micaelli2019zero}, and Wonder Filter (WF)~\cite{li2019persistent}.
The corresponding DNN architectures are detailed in Table~\ref{table:settings}.
To examine the utility of the proposed knowledge-free watermarking scheme, we evaluated the performance of the watermark along with the watermarked DNN for three trigger injection schemes: backdoor injection with no data ($\mathcal{B}$), injection with the training dataset ($\mathcal{D}$), injection with anchors ($\mathcal{A}$), and two types of backdoors: triggers ($T$) and post-triggers ($P$).
We used the DFD module from~\cite{fang2019data}, adopted \texttt{PyTorch} for implementation, and the codes have been made available~\footnote{some github lib}.

\begin{table*}[!t]
\centering
\caption{The basic settings and configurations of experiments.}
\scalebox{0.78}{
\begin{tabular}{c|c|c|c|c|c}
\toprule
\multicolumn{4}{c||}{\textbf{Basic settings}} & \multicolumn{2}{c}{\textbf{Configuration ($3\times2$)}}\\
\toprule
\textbf{Index} & \textbf{Dataset} & \textbf{DNN architecture} & \textbf{Trigger encoder $T(\cdot)$} & \textbf{Injection scheme} & \textbf{Backdoor category} \\
\midrule
\textbf{E1},\textbf{E2},\textbf{E3},\textbf{E4} & MNIST ($C=10$) & LeNet~\cite{wei2019development} & Gaussian, Random, DFD,WF & (Post-)Triggers solely ($\mathcal{B}$). & Triggers ($T$). \\
\cline{5-5}
\textbf{E5},\textbf{E6},\textbf{E7},\textbf{E8} & CIFAR10 ($C=10$) & ResNet-34~\cite{he2016deep} & Gaussian, Random, DFD,WF & Injection with the & \\
\cline{6-6}
\textbf{E9},\textbf{E10},\textbf{E11},\textbf{E12} & CIFAR100 ($C=100$) & ResNet-34~\cite{he2016deep} & Gaussian, Random, DFD,WF & training dataset ($\mathcal{D}$). & Post-triggers ($P$). \\
\cline{5-5}
\textbf{E13},\textbf{E14},\textbf{E15},\textbf{E16} & Caltech101 ($C=101$) & ResNet-34~\cite{he2016deep} & Gaussian, Random, DFD,WF & Injection with anchors ($\mathcal{A}$). & \\
\bottomrule
\end{tabular}}
\label{table:settings}
\end{table*}

\begin{figure*}[htbp]
\subfigure[\textbf{E1.}]{
\begin{minipage}[htbp]{0.25\linewidth}
\includegraphics[width=4cm]{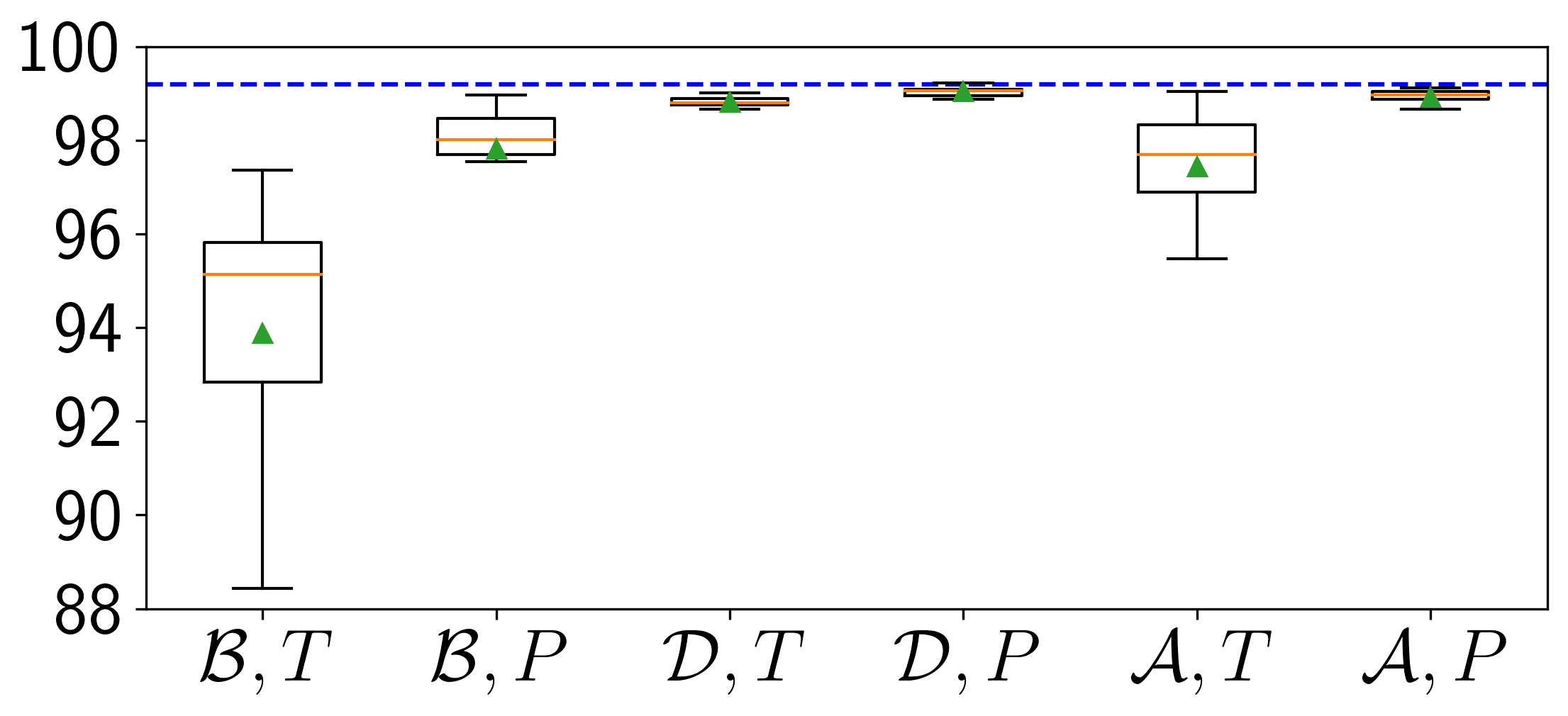}
\end{minipage}%
}%
\subfigure[\textbf{E2.}]{
\begin{minipage}[htbp]{0.25\linewidth}
\includegraphics[width=4cm]{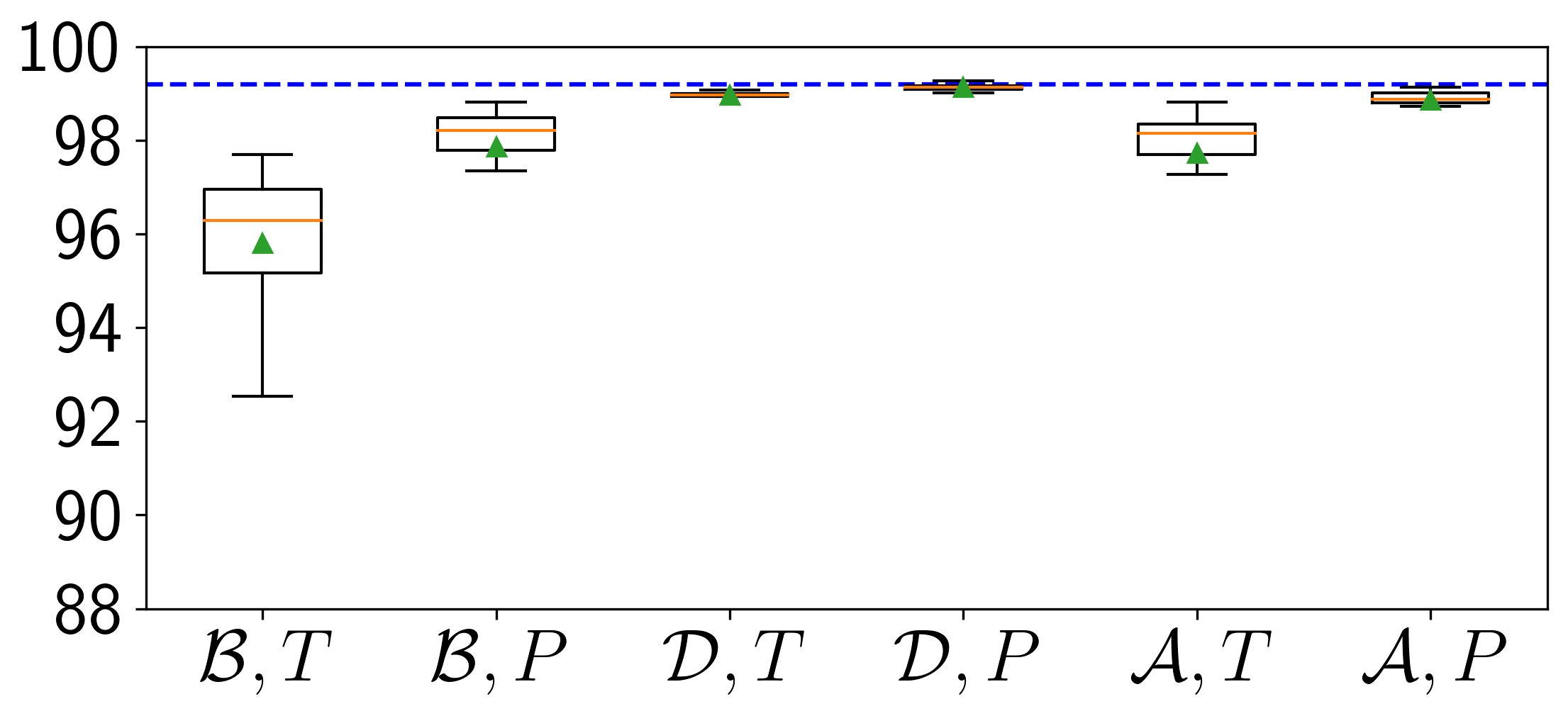}
\end{minipage}%
}%
\subfigure[\textbf{E3.}]{
\begin{minipage}[htbp]{0.25\linewidth}
\includegraphics[width=4cm]{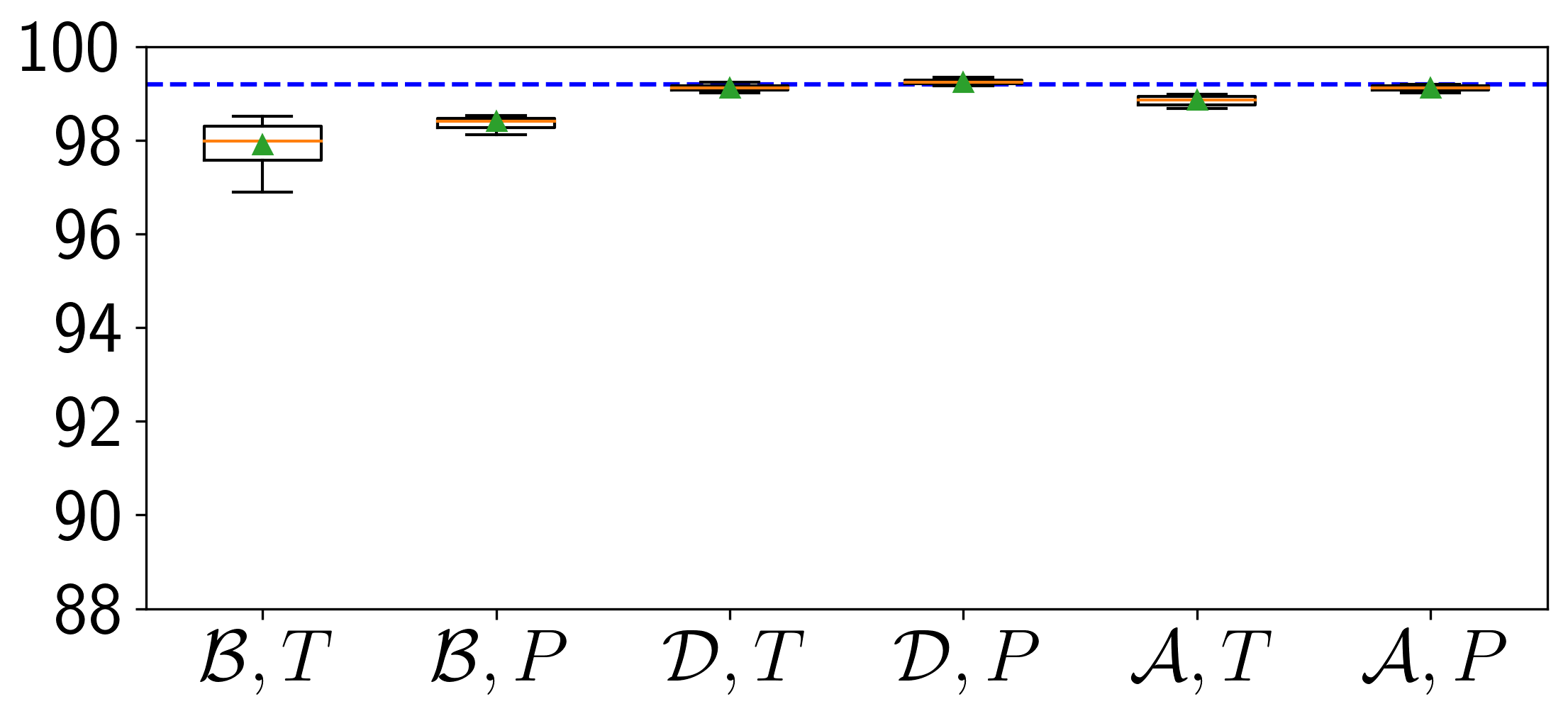}
\end{minipage}%
}%
\subfigure[\textbf{E4.}]{
\begin{minipage}[htbp]{0.25\linewidth}
\includegraphics[width=4cm]{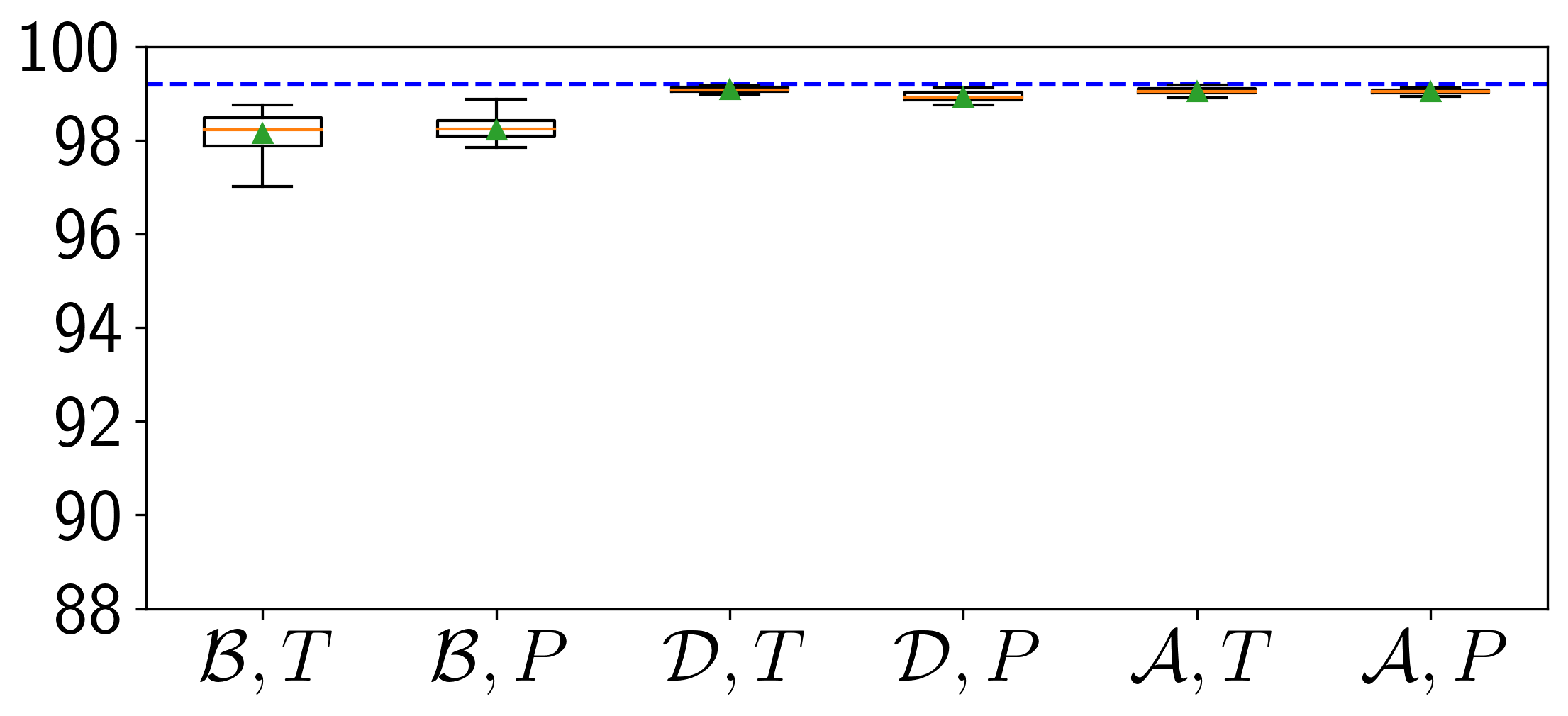}
\end{minipage}%
}
\subfigure[\textbf{E5.}]{
\begin{minipage}[htbp]{0.25\linewidth}
\includegraphics[width=4cm]{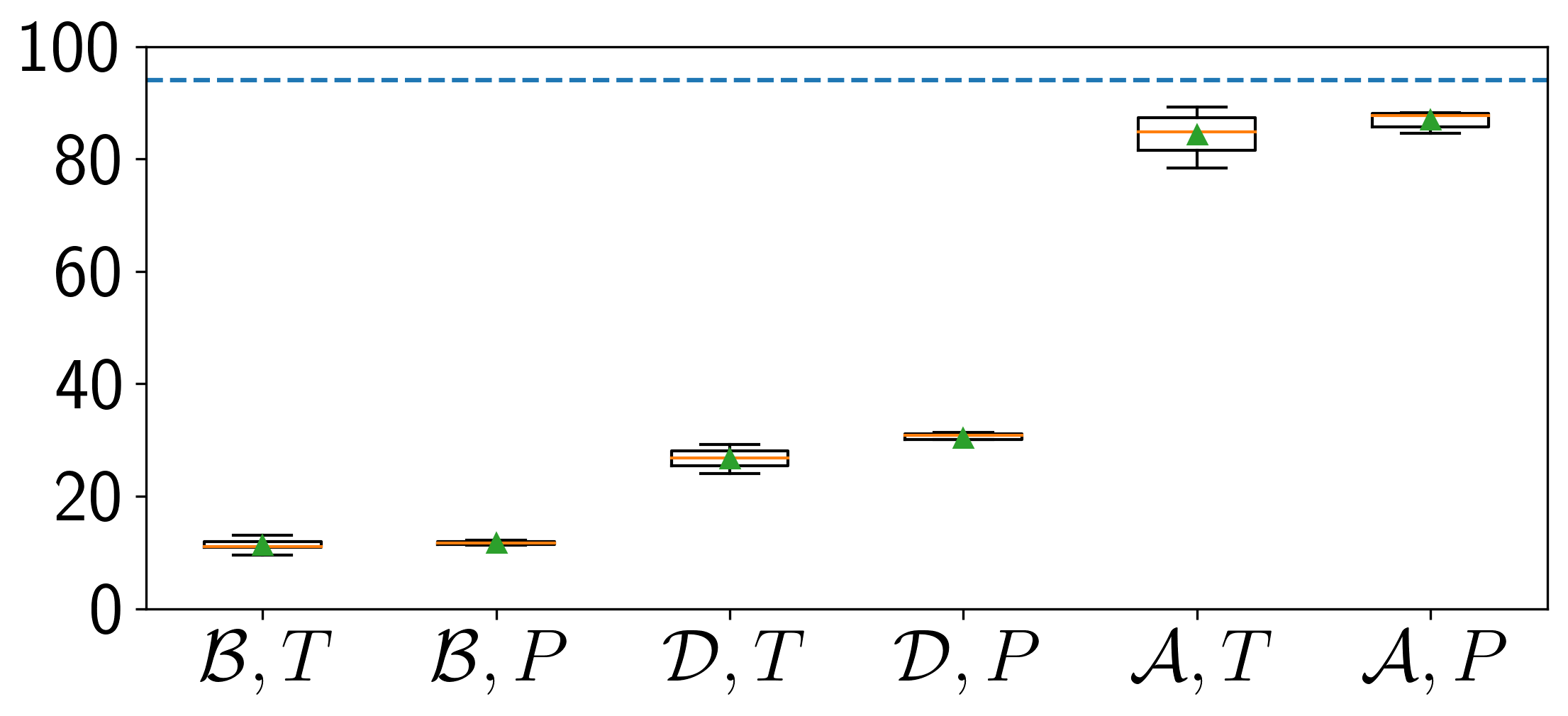}
\end{minipage}%
}%
\subfigure[\textbf{E6.}]{
\begin{minipage}[htbp]{0.25\linewidth}
\includegraphics[width=4cm]{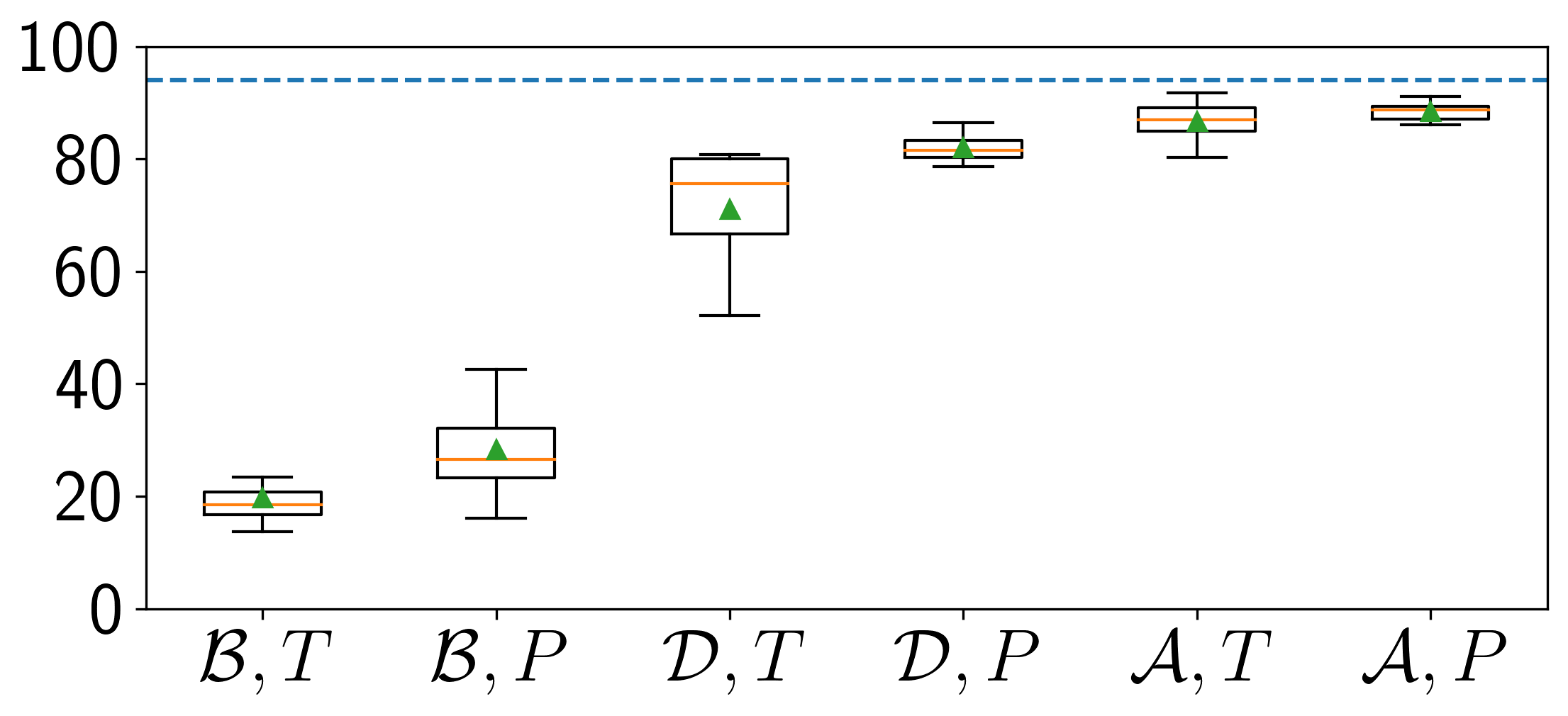}
\end{minipage}%
}%
\subfigure[\textbf{E7.}]{
\begin{minipage}[htbp]{0.25\linewidth}
\includegraphics[width=4cm]{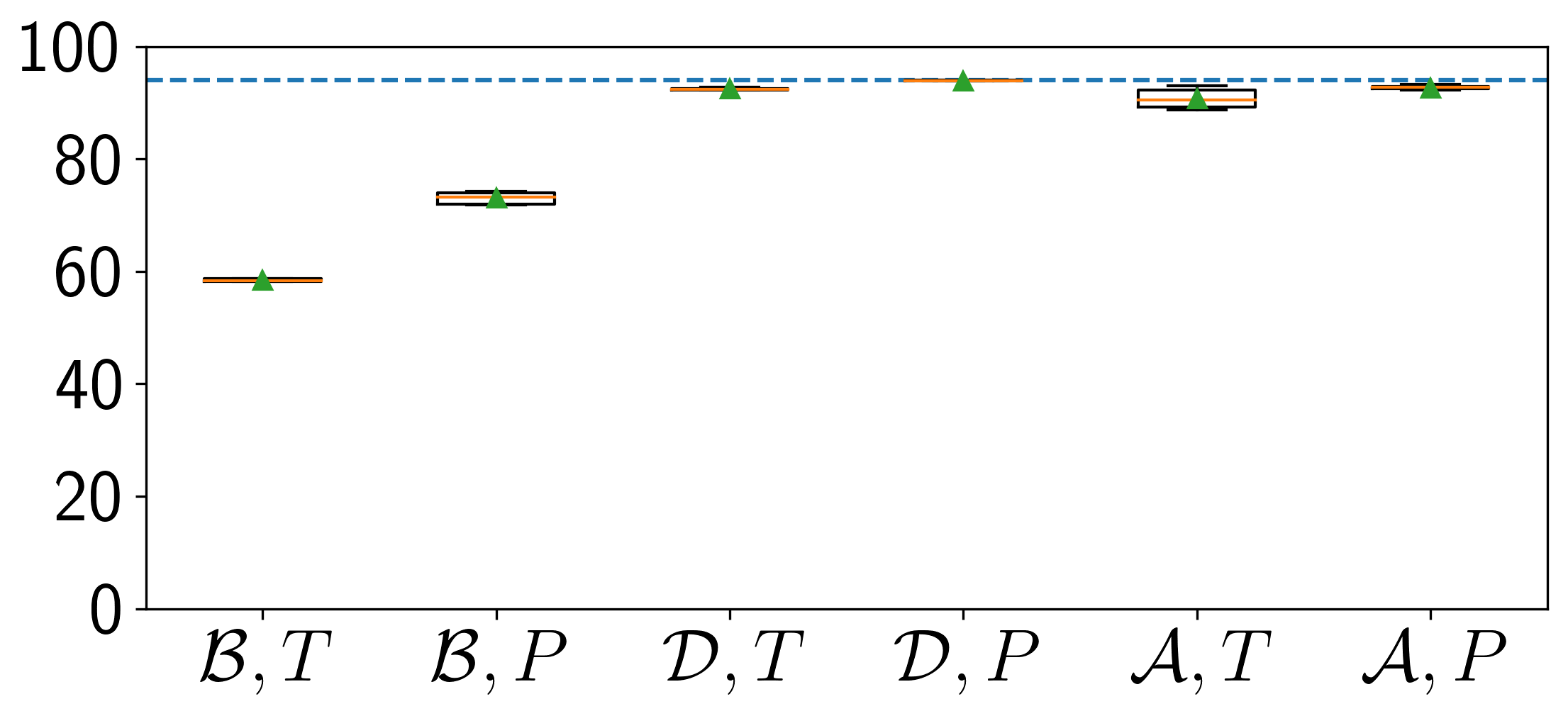}
\end{minipage}%
}%
\subfigure[\textbf{E8.}]{
\begin{minipage}[htbp]{0.25\linewidth}
\includegraphics[width=4cm]{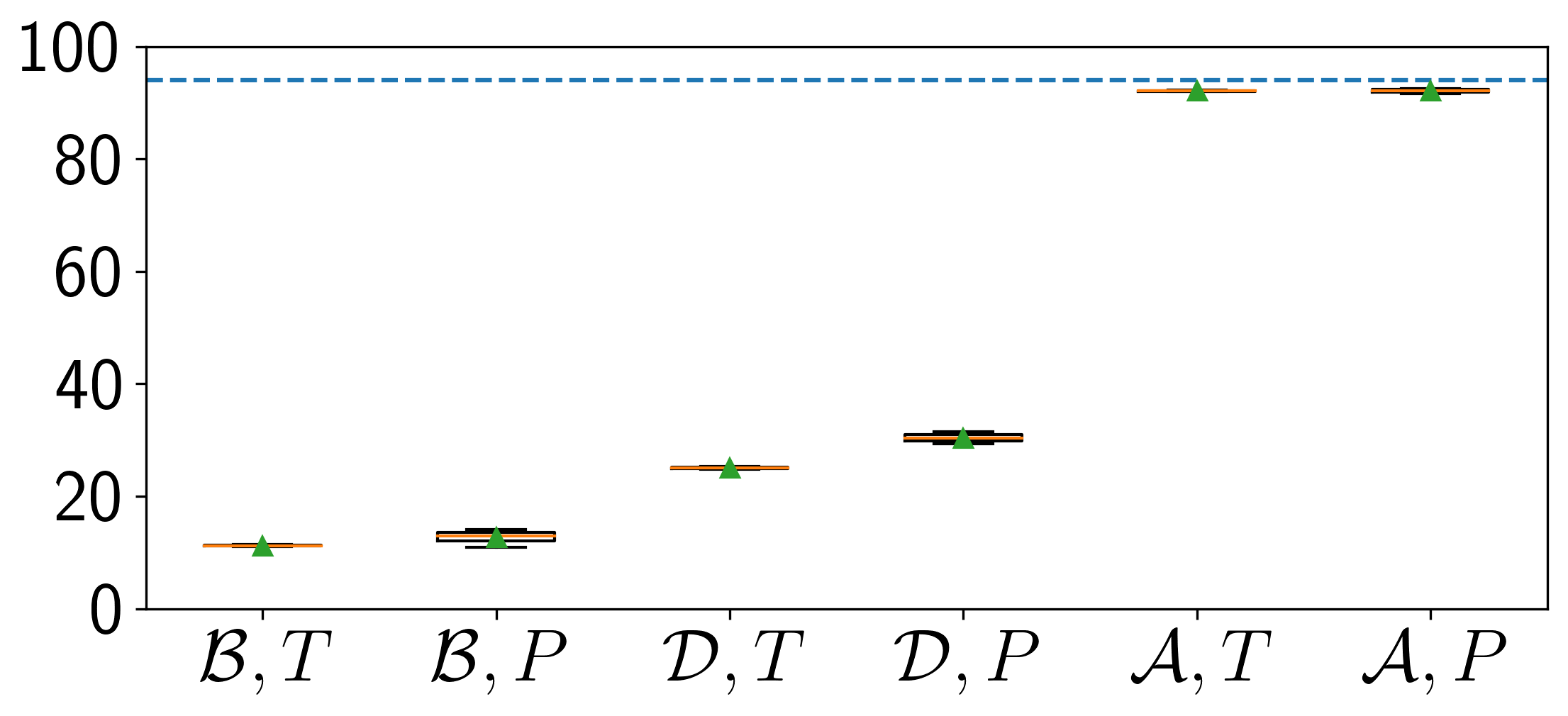}
\end{minipage}%
}
\subfigure[\textbf{E9.}]{
\begin{minipage}[htbp]{0.25\linewidth}
\includegraphics[width=4cm]{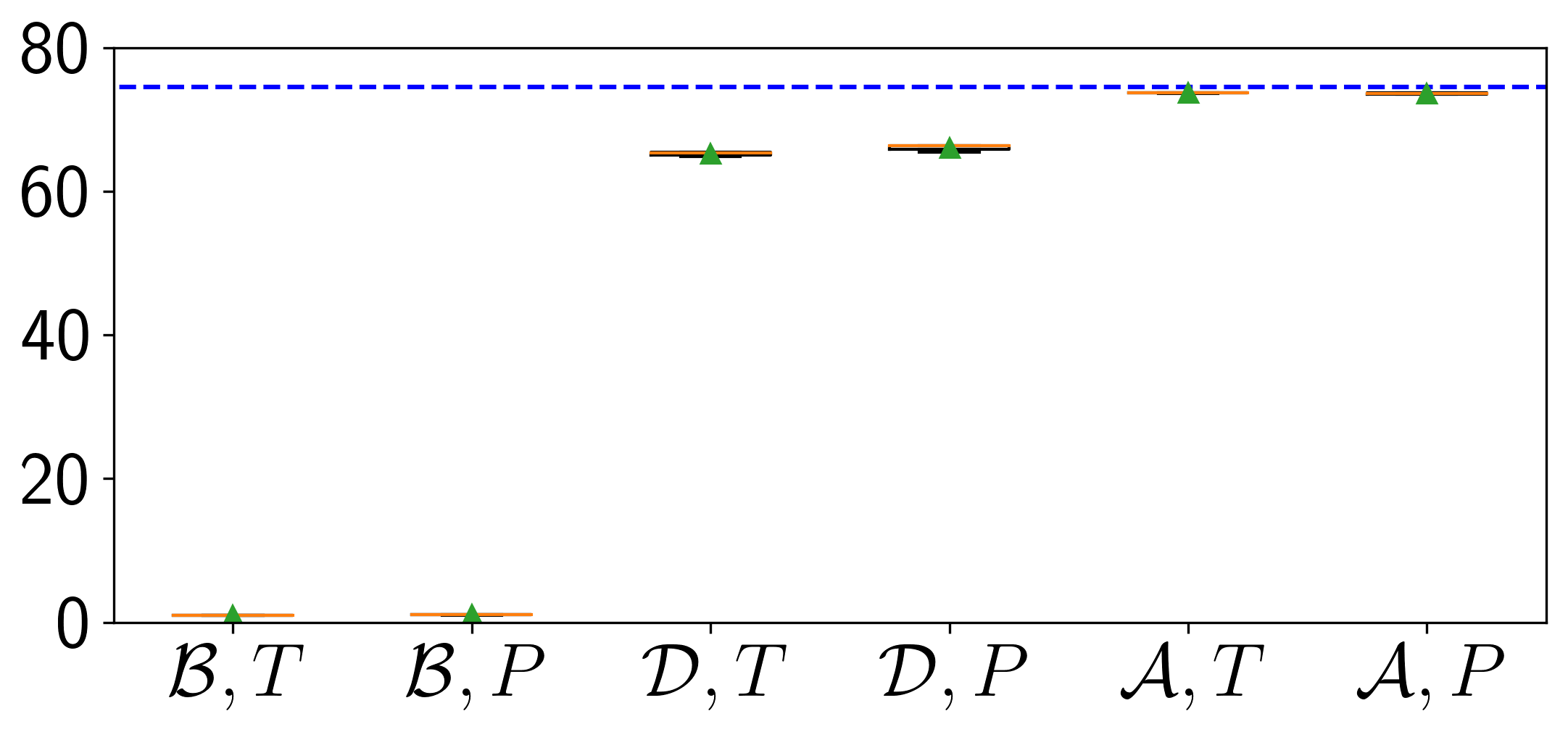}
\end{minipage}%
}%
\subfigure[\textbf{E10.}]{
\begin{minipage}[htbp]{0.25\linewidth}
\includegraphics[width=4cm]{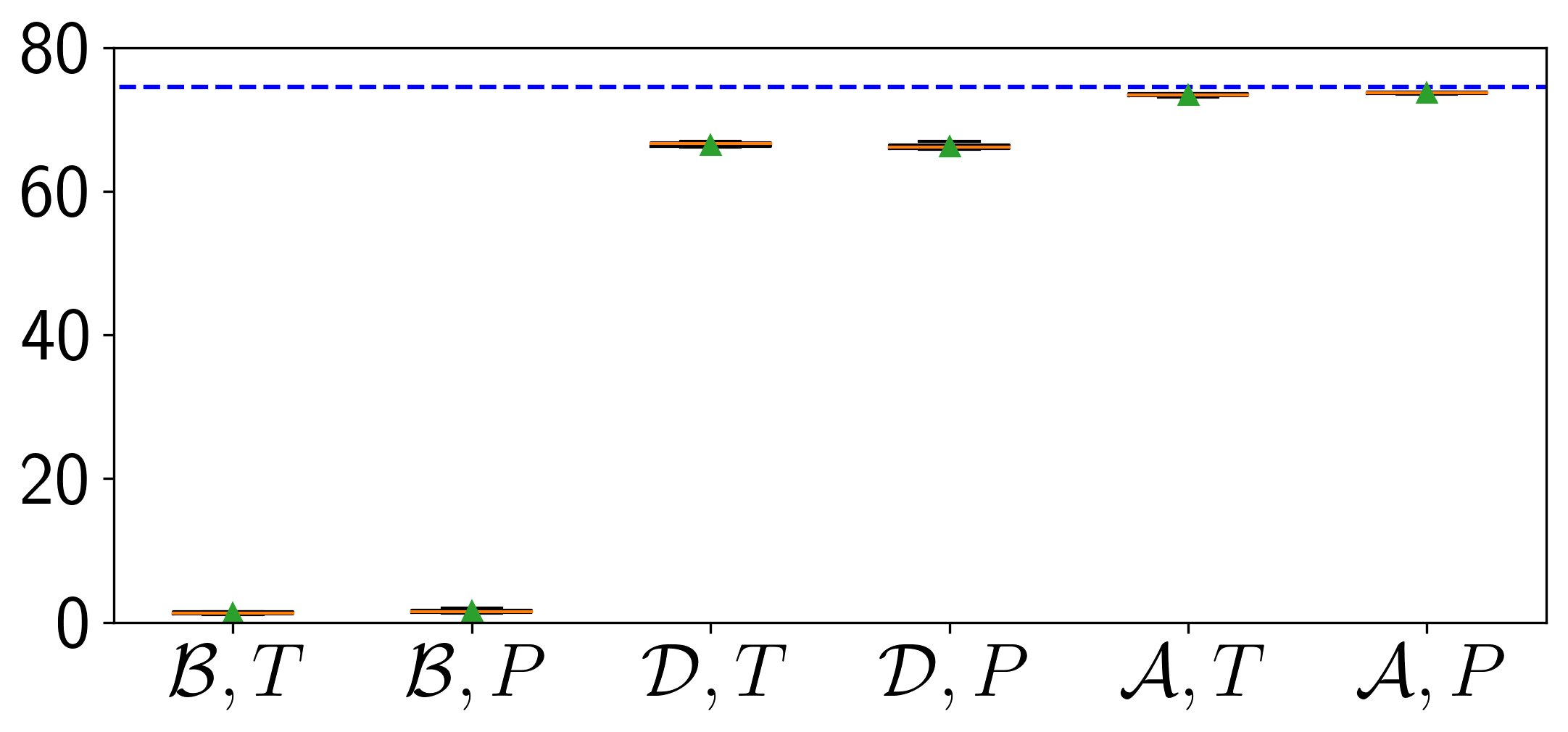}
\end{minipage}%
}%
\subfigure[\textbf{E11.}]{
\begin{minipage}[htbp]{0.25\linewidth}
\includegraphics[width=4cm]{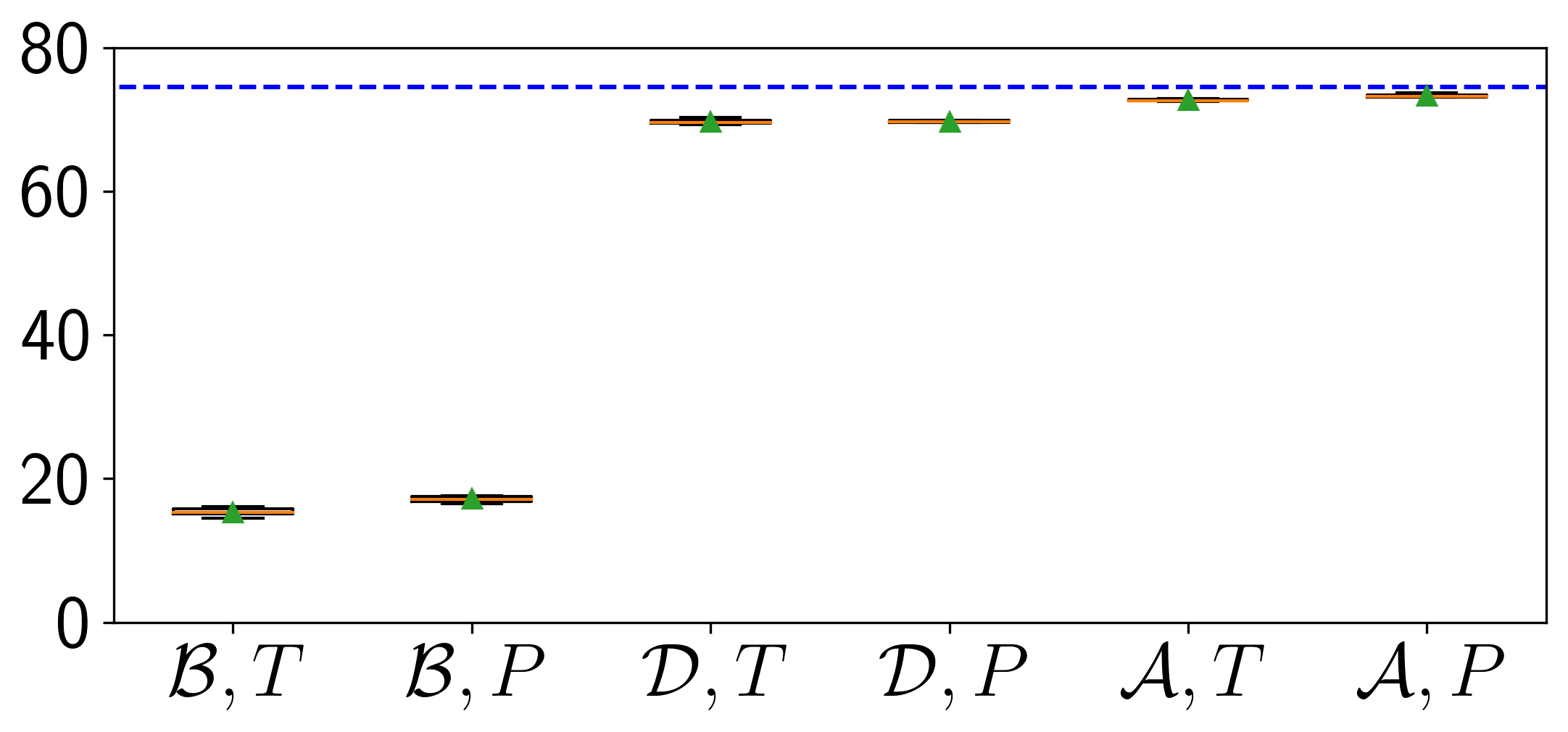}
\end{minipage}%
}%
\subfigure[\textbf{E12.}]{
\begin{minipage}[htbp]{0.25\linewidth}
\includegraphics[width=4cm]{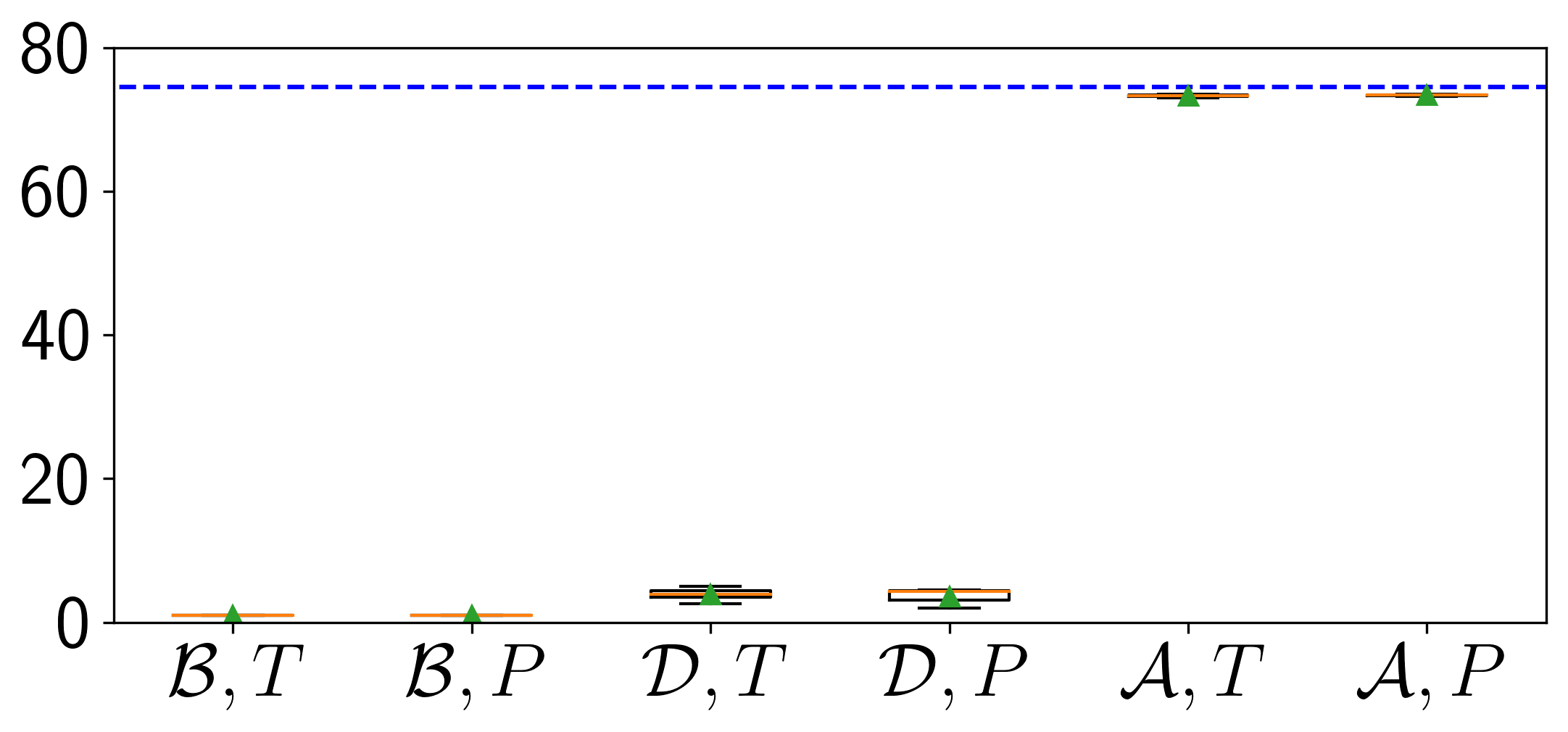}
\end{minipage}%
}
\subfigure[\textbf{E13.}]{
\begin{minipage}[htbp]{0.25\linewidth}
\includegraphics[width=4cm]{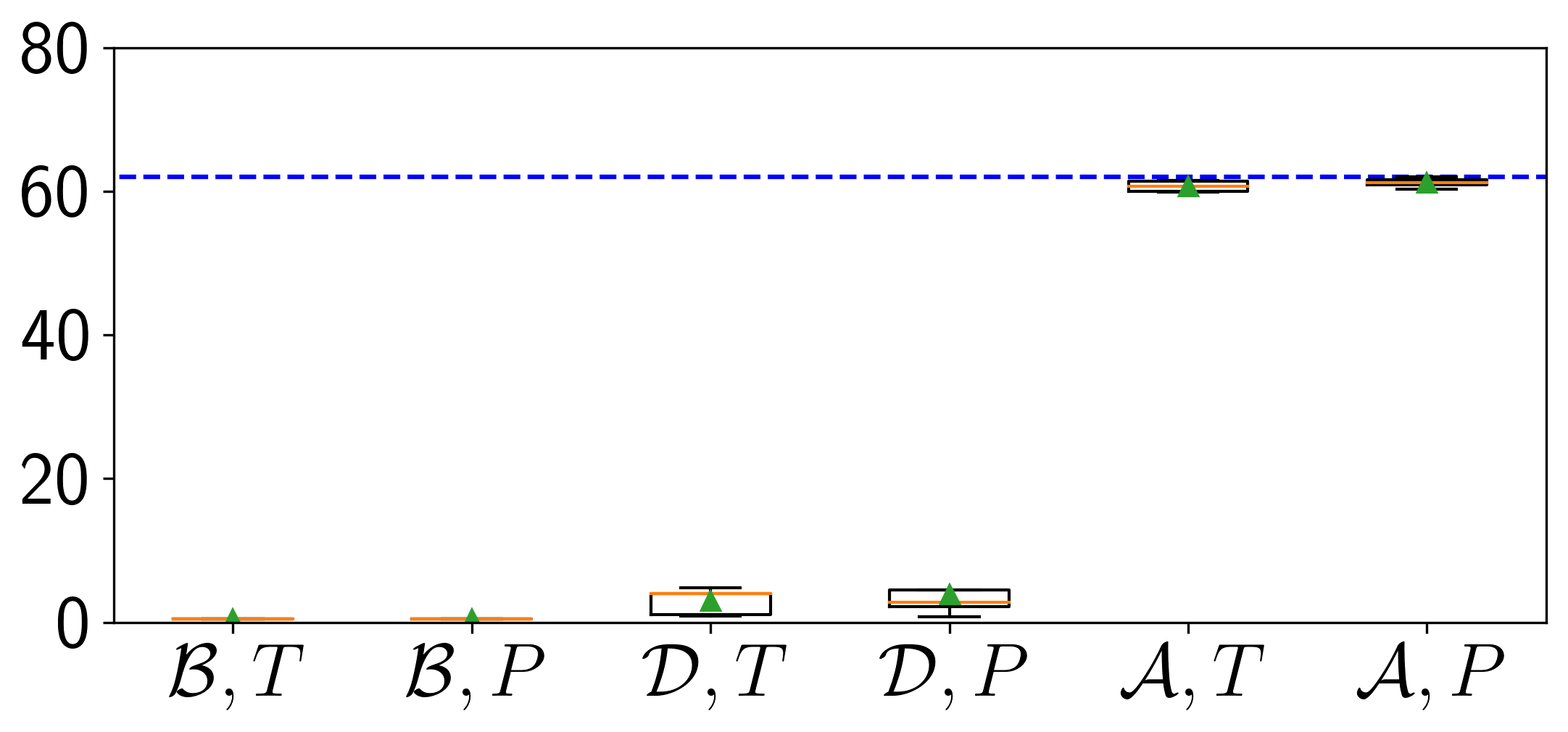}
\end{minipage}%
}%
\subfigure[\textbf{E14.}]{
\begin{minipage}[htbp]{0.25\linewidth}
\includegraphics[width=4cm]{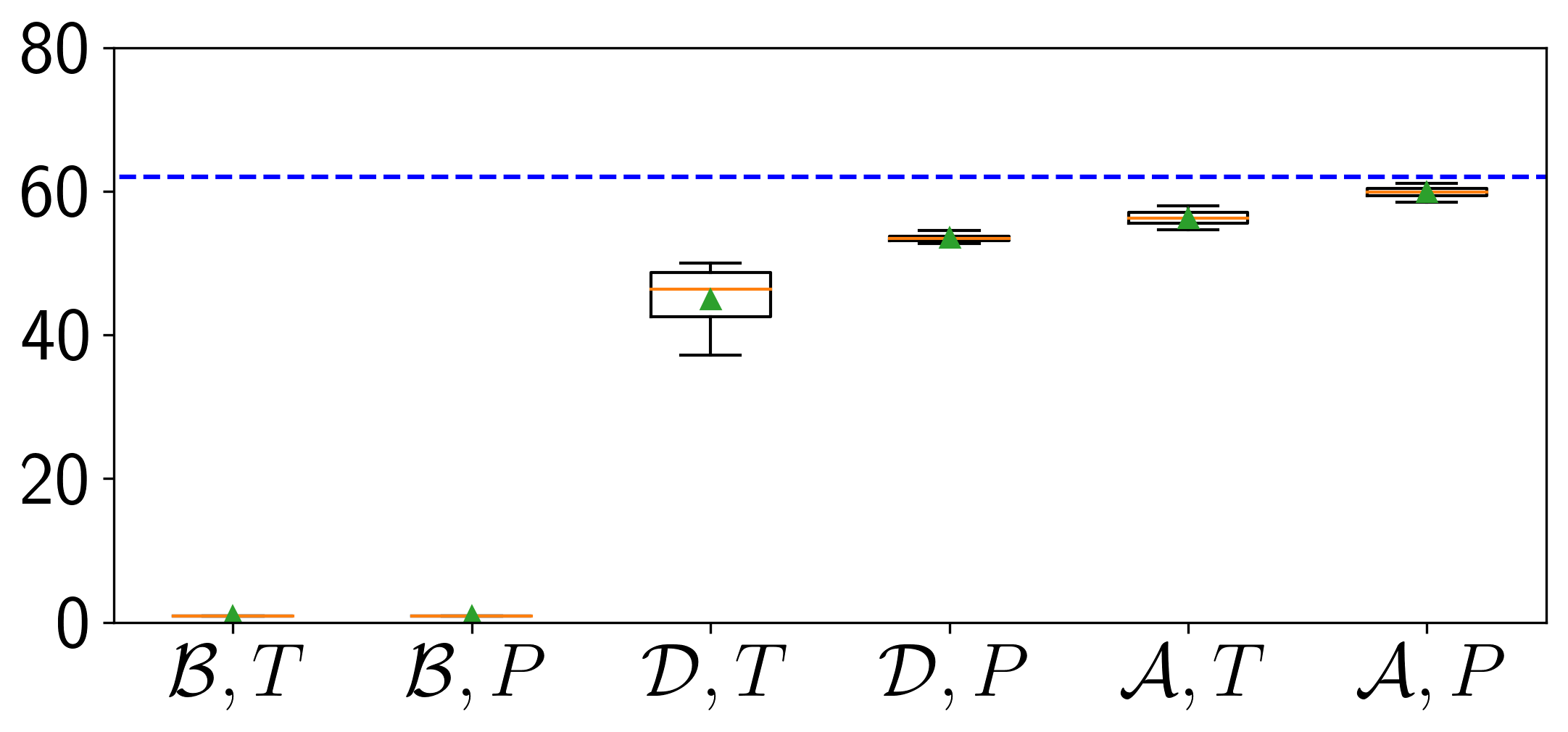}
\end{minipage}%
}%
\subfigure[\textbf{E15.}]{
\begin{minipage}[htbp]{0.25\linewidth}
\includegraphics[width=4cm]{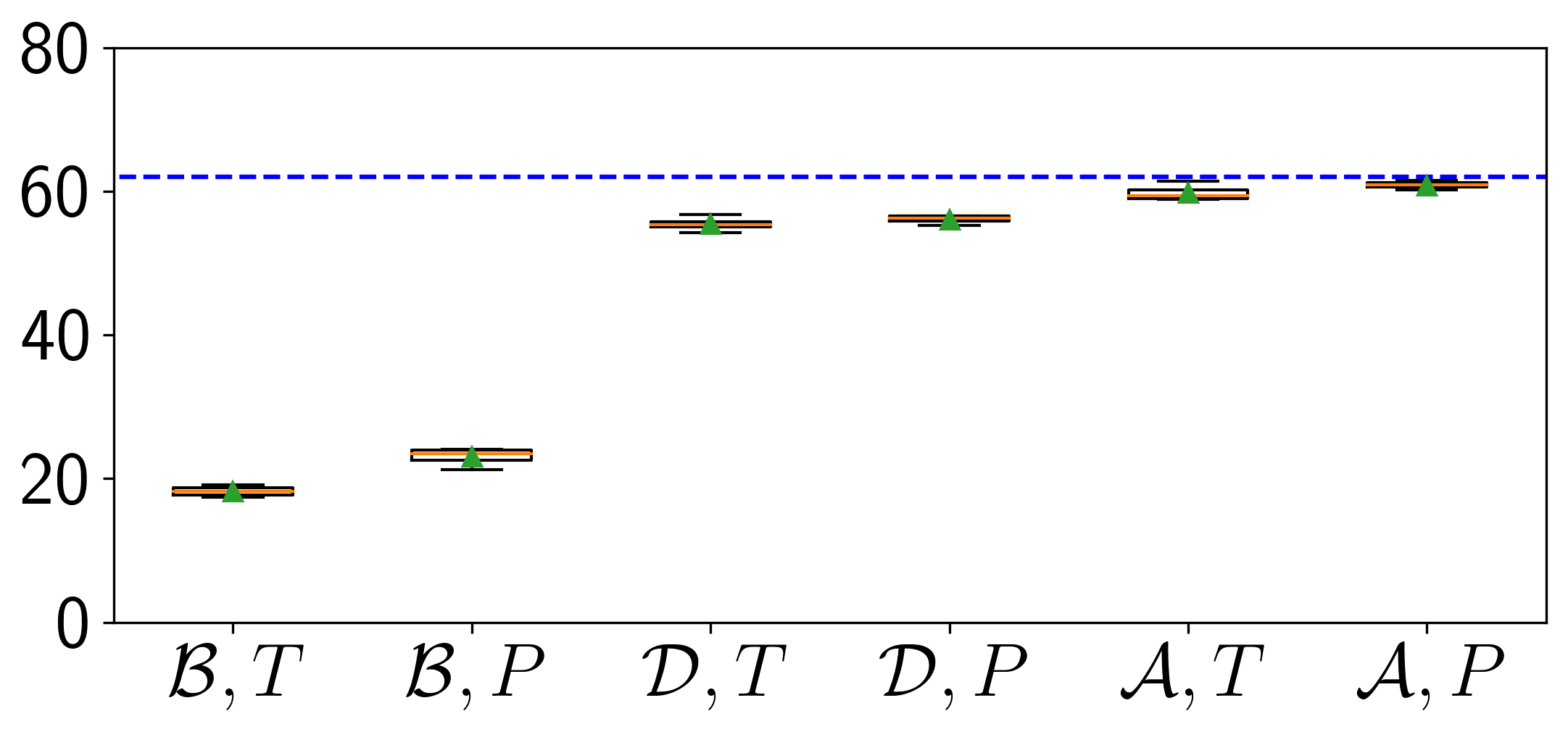}
\end{minipage}%
}%
\subfigure[\textbf{E16.}]{
\begin{minipage}[htbp]{0.25\linewidth}
\includegraphics[width=4cm]{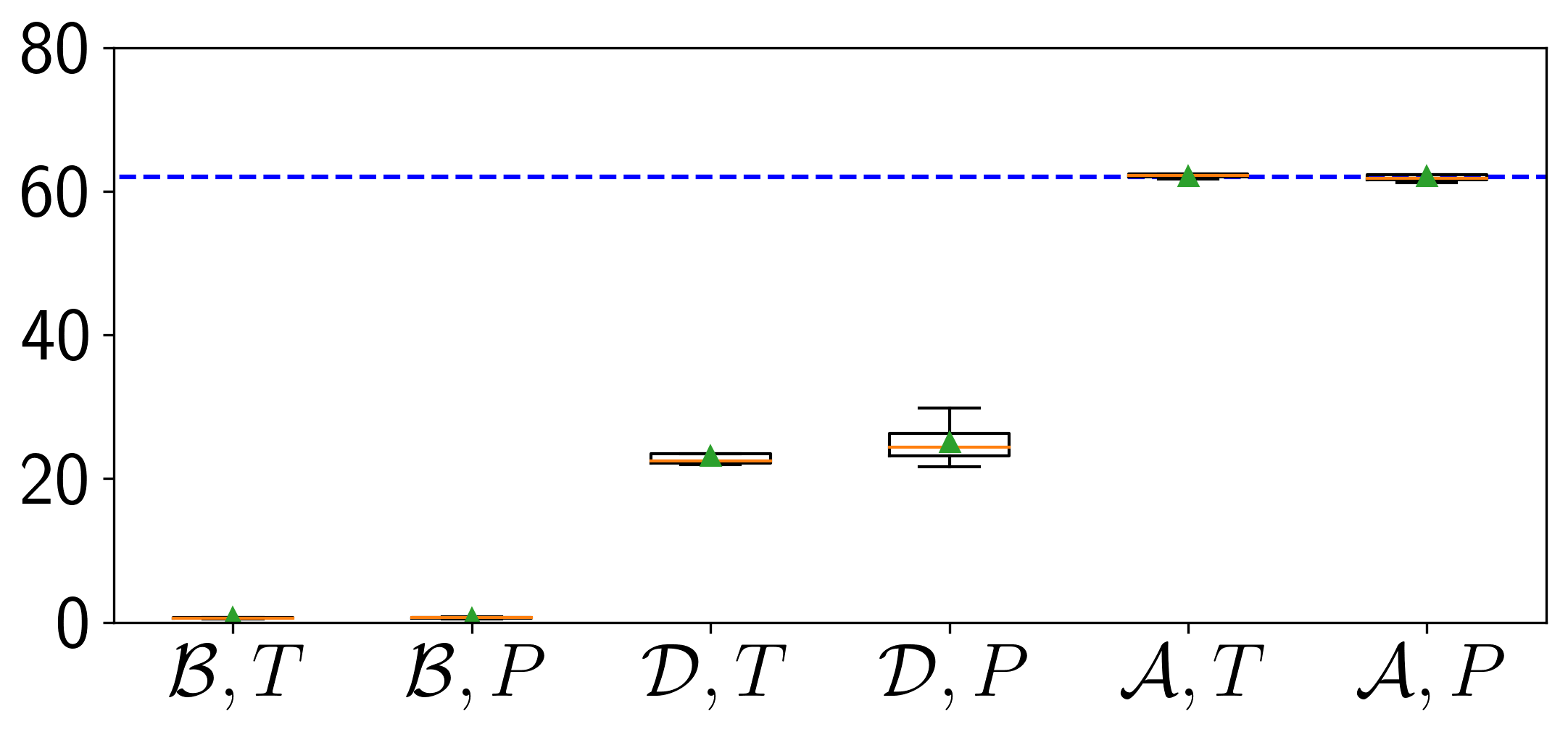}
\end{minipage}%
}%
\caption{The watermarked DNN's classification accuracy (\%) on the test dataset under different configurations. The performance of the clean DNN is marked in blue.}
\label{fig:exp1}
\end{figure*}

\begin{figure*}[!t]
\centering
\includegraphics[width=15.5cm]{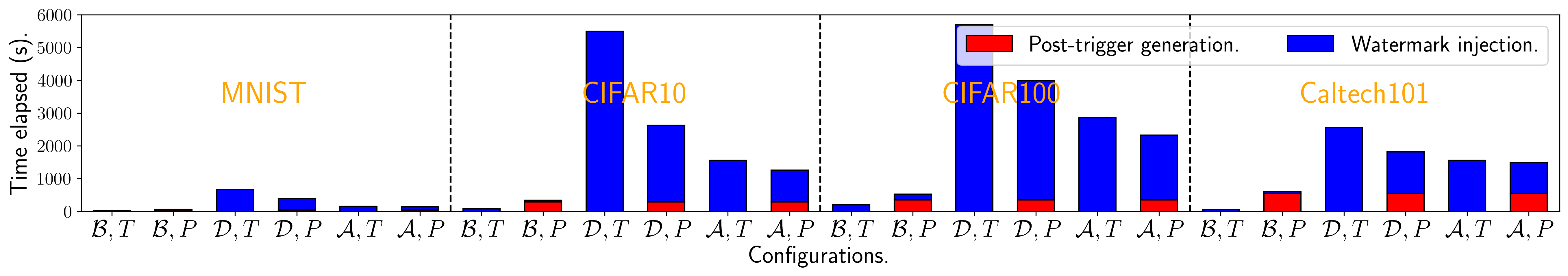}
\caption{The time consumption of watermark injection.}
\label{figure:time}
\end{figure*}

\begin{figure*}[!t]
\subfigure[\textbf{E1}-\textbf{E4}.]{
\begin{minipage}[htbp]{0.25\linewidth}
\includegraphics[width=4.3cm]{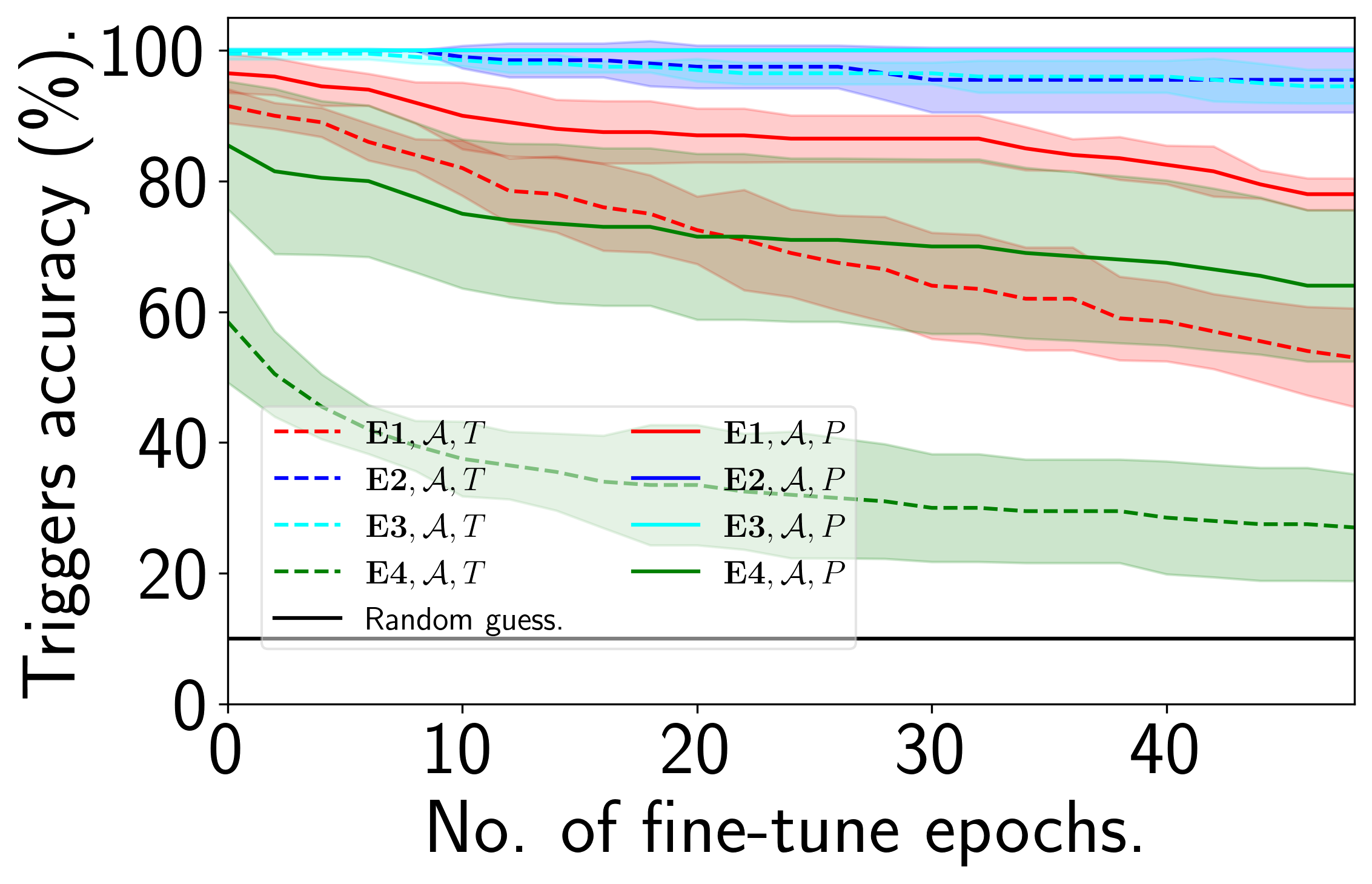}
\end{minipage}%
}%
\subfigure[\textbf{E5}-\textbf{E8}.]{
\begin{minipage}[htbp]{0.25\linewidth}
\includegraphics[width=4.3cm]{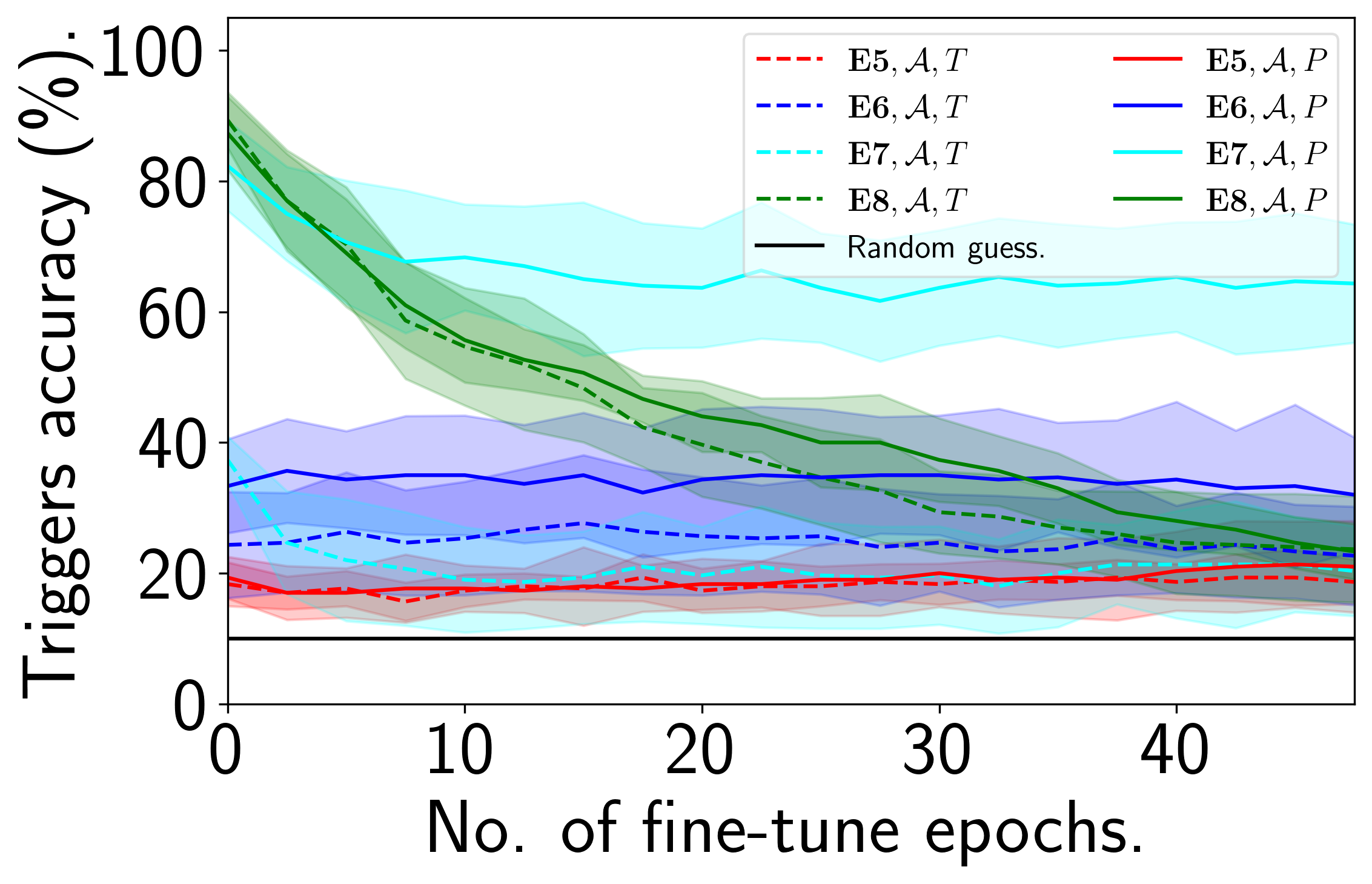}
\end{minipage}%
}%
\subfigure[\textbf{E9}-\textbf{E12}.]{
\begin{minipage}[htbp]{0.25\linewidth}
\includegraphics[width=4.3cm]{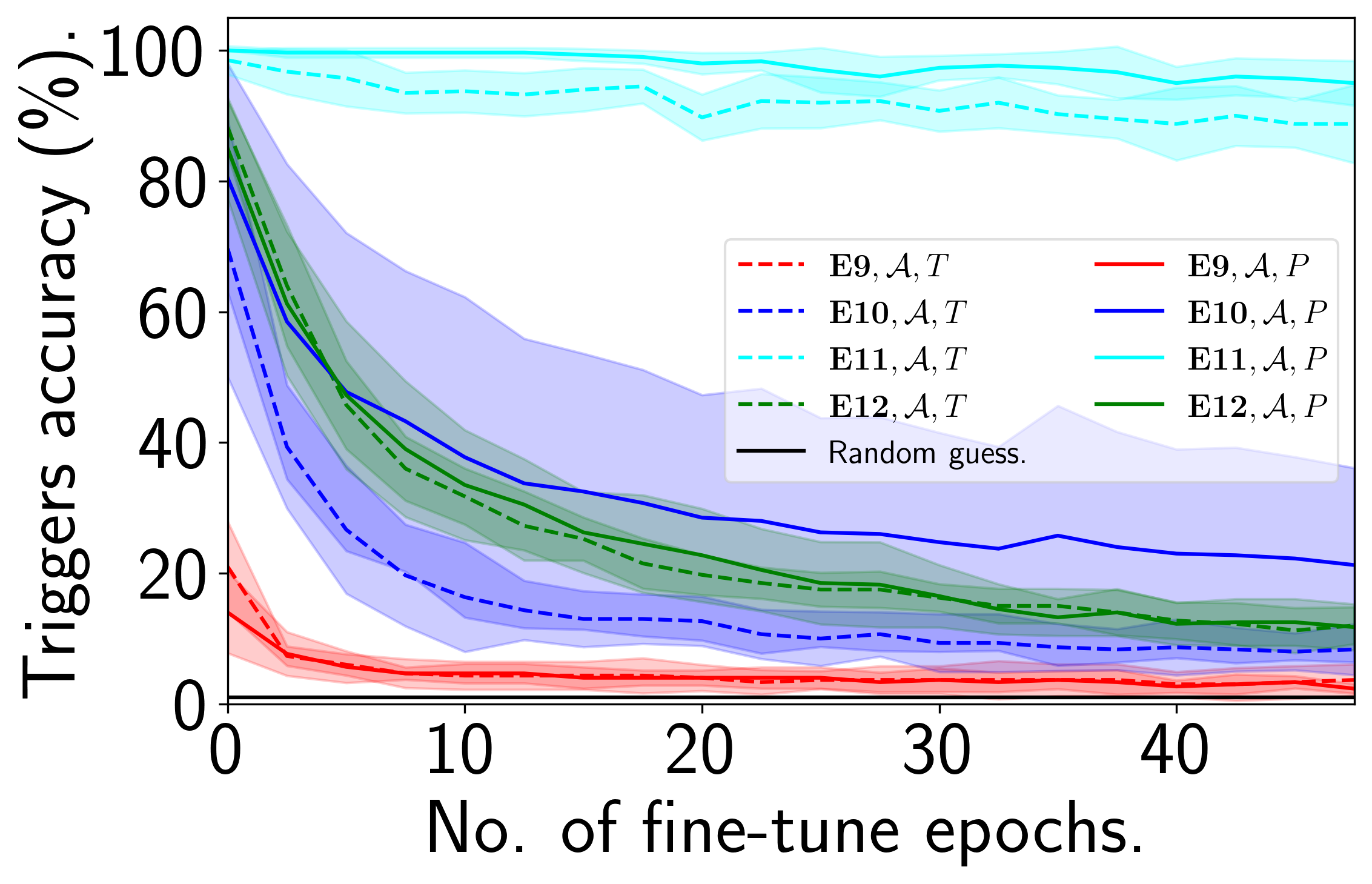}
\end{minipage}%
}%
\subfigure[\textbf{E13}-\textbf{E16}.]{
\begin{minipage}[htbp]{0.25\linewidth}
\includegraphics[width=4.3cm]{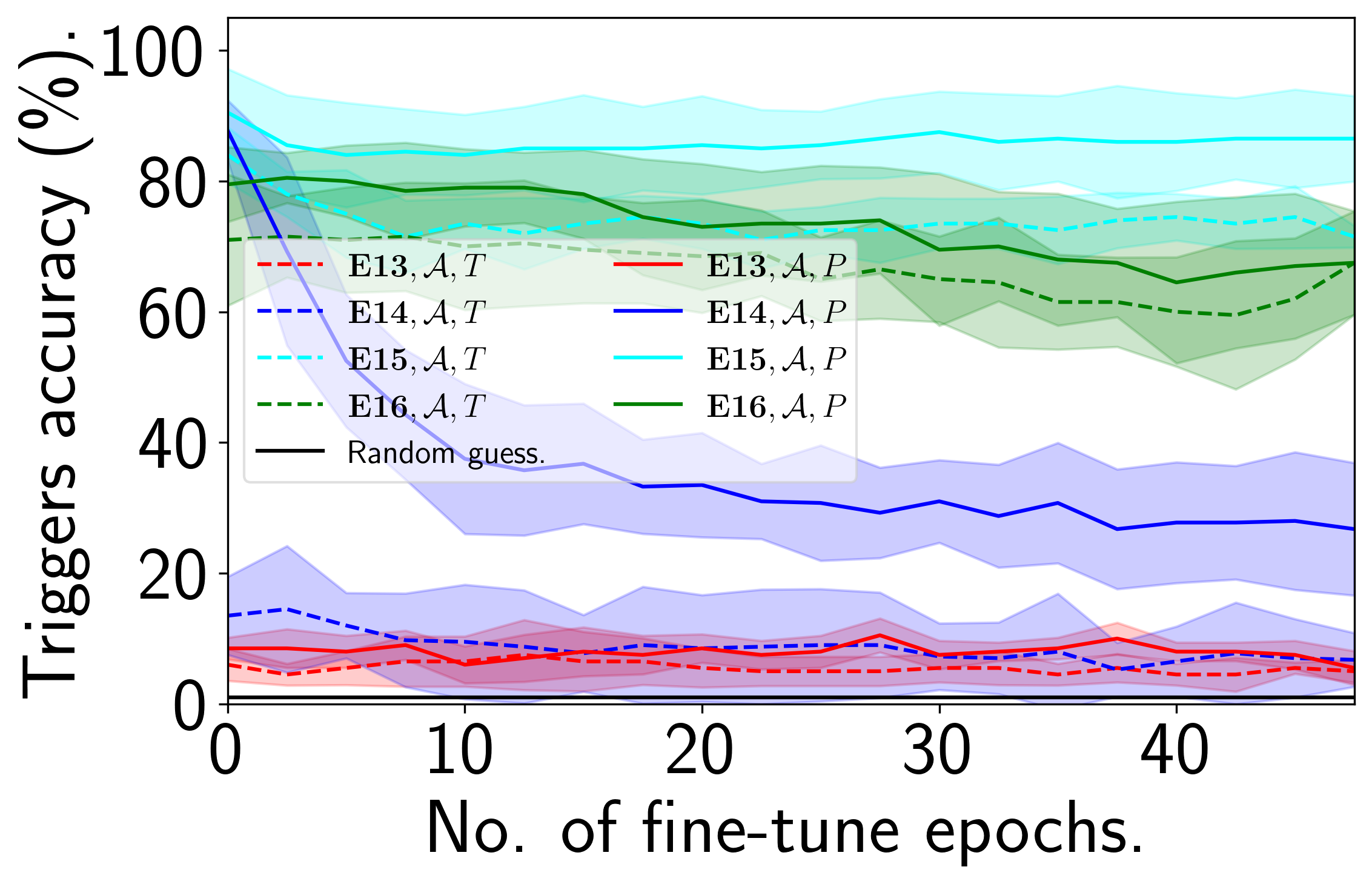}
\end{minipage}%
}%
\caption{The classification of (post-)triggers after fine-tuning.}
\label{fig:ft}
\end{figure*}

\begin{figure*}[!t]
\subfigure[\textbf{E1}-\textbf{E4}.]{
\begin{minipage}[htbp]{0.25\linewidth}
\includegraphics[width=4.3cm]{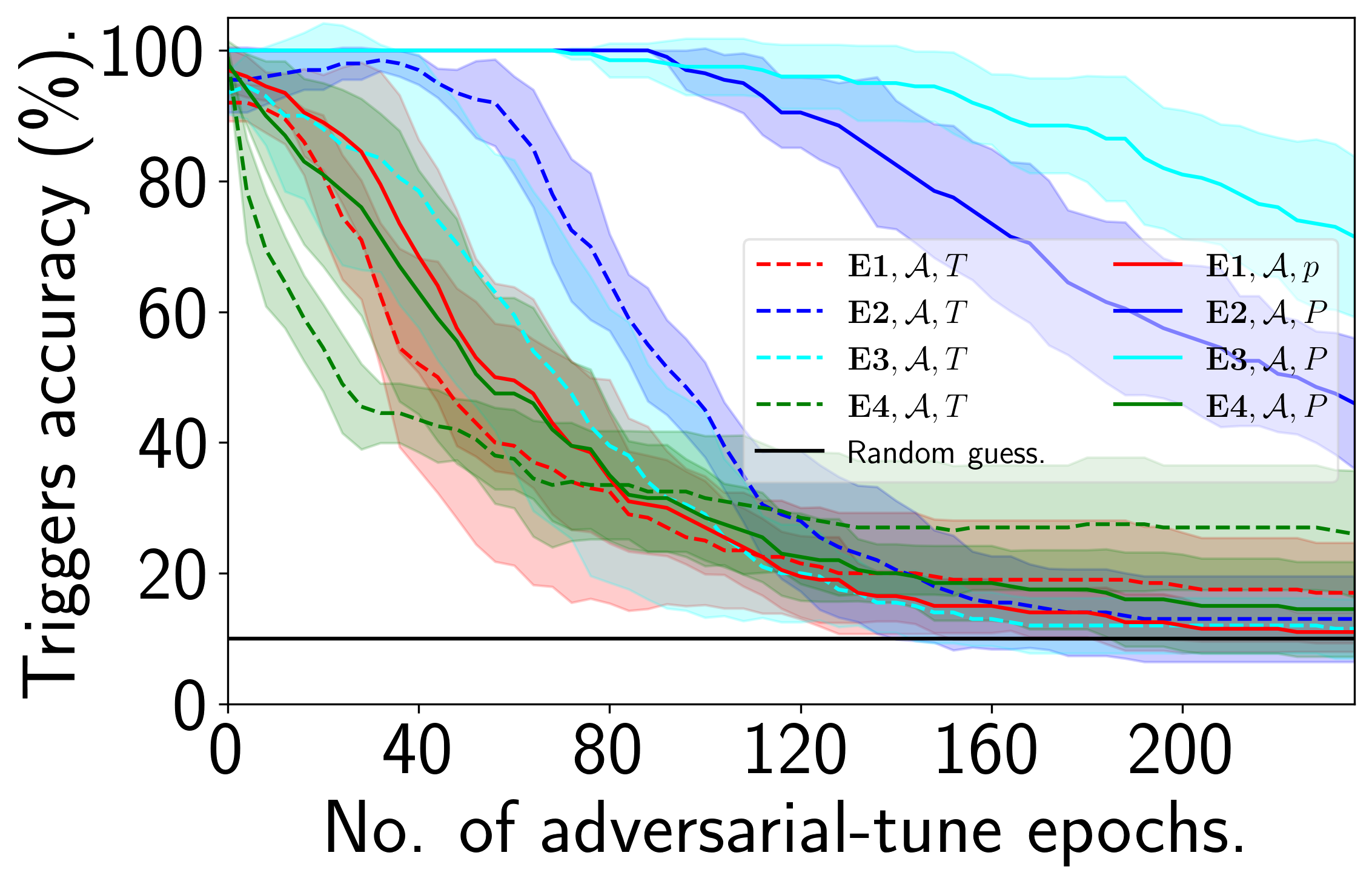}
\end{minipage}%
}%
\subfigure[\textbf{E5}-\textbf{E8}.]{
\begin{minipage}[htbp]{0.25\linewidth}
\includegraphics[width=4.3cm]{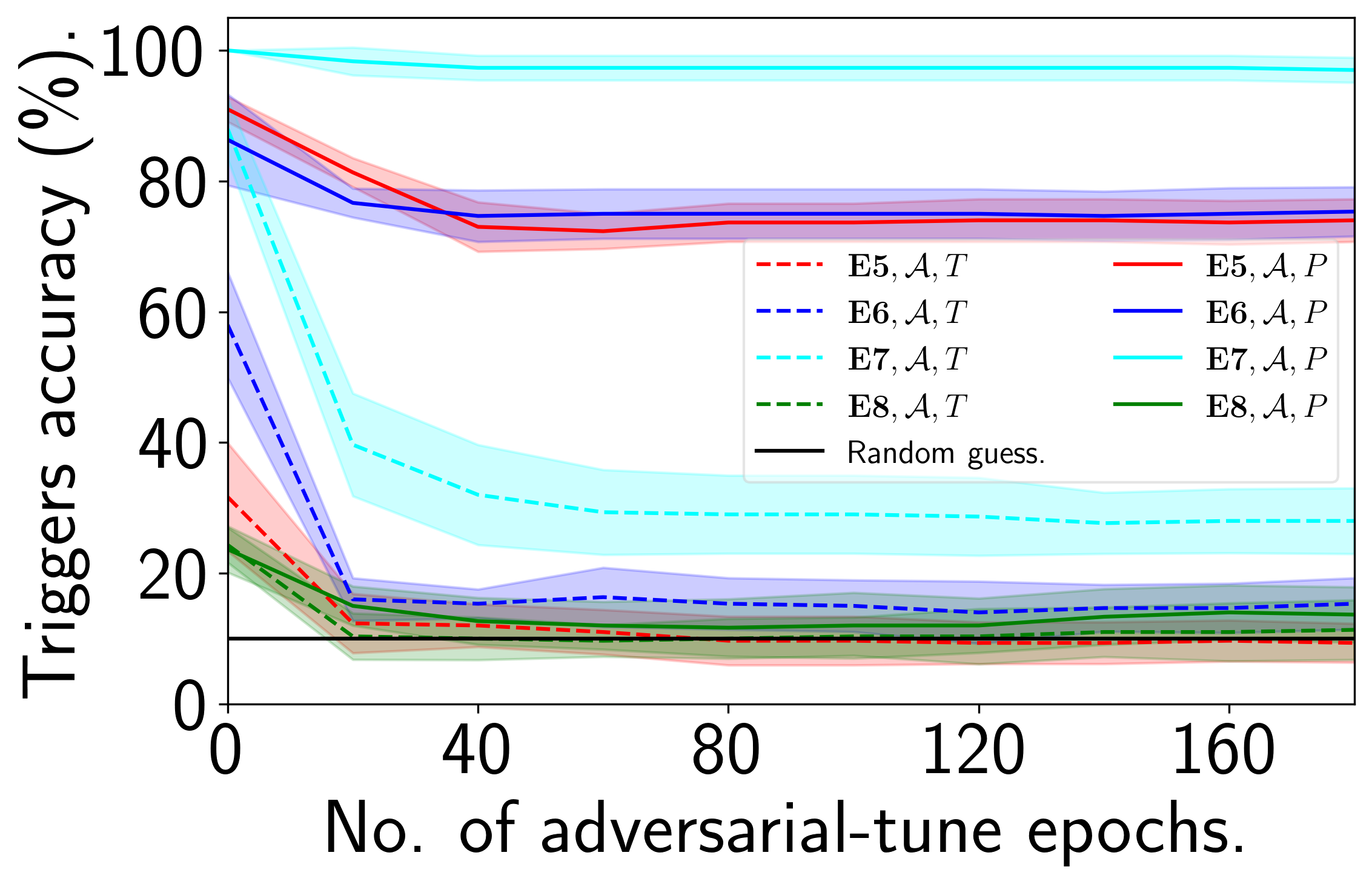}
\end{minipage}%
}%
\subfigure[\textbf{E9}-\textbf{E12}.]{
\begin{minipage}[htbp]{0.25\linewidth}
\includegraphics[width=4.3cm]{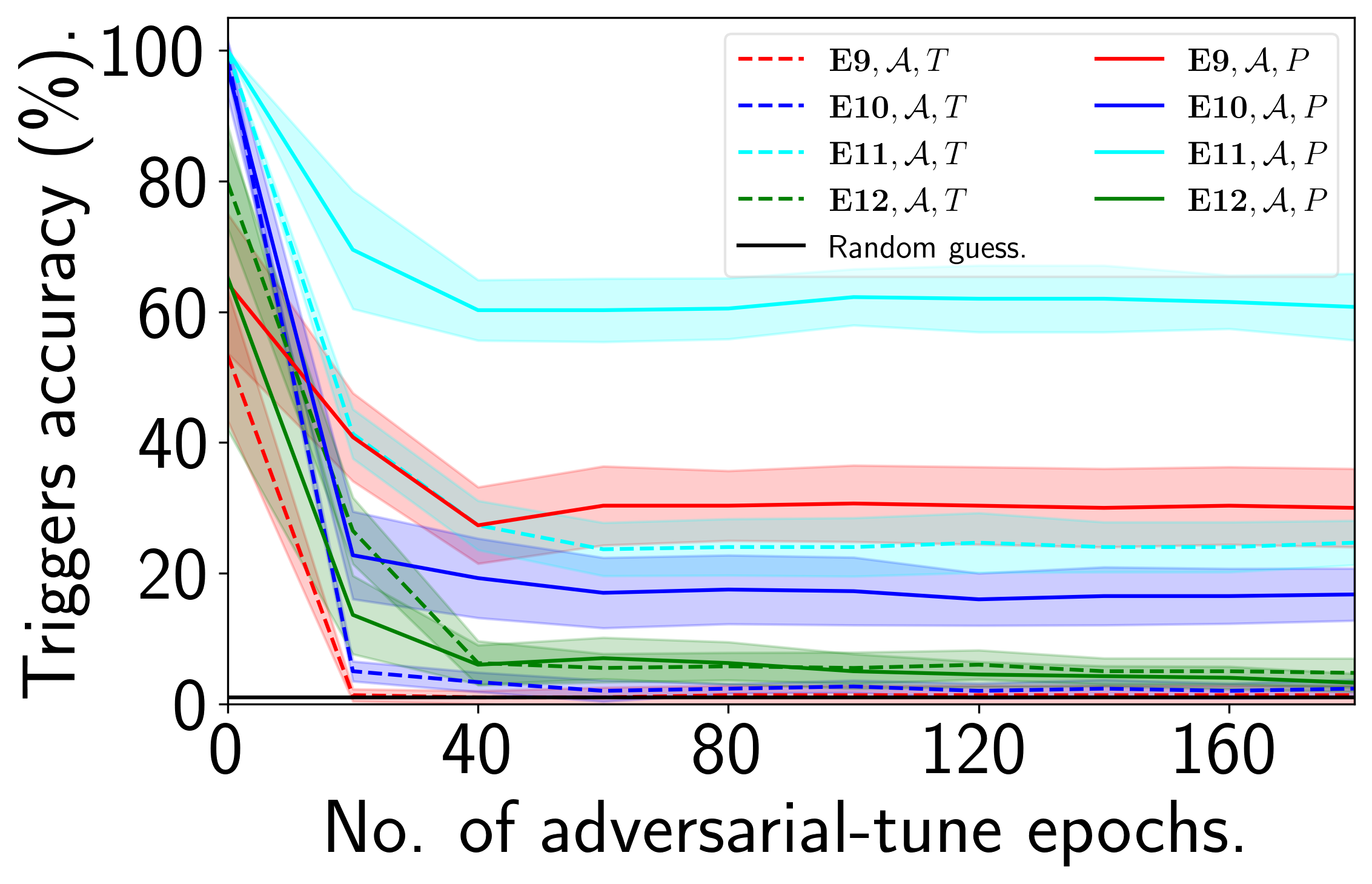}
\end{minipage}%
}%
\subfigure[\textbf{E13}-\textbf{E16}.]{
\begin{minipage}[htbp]{0.25\linewidth}
\includegraphics[width=4.3cm]{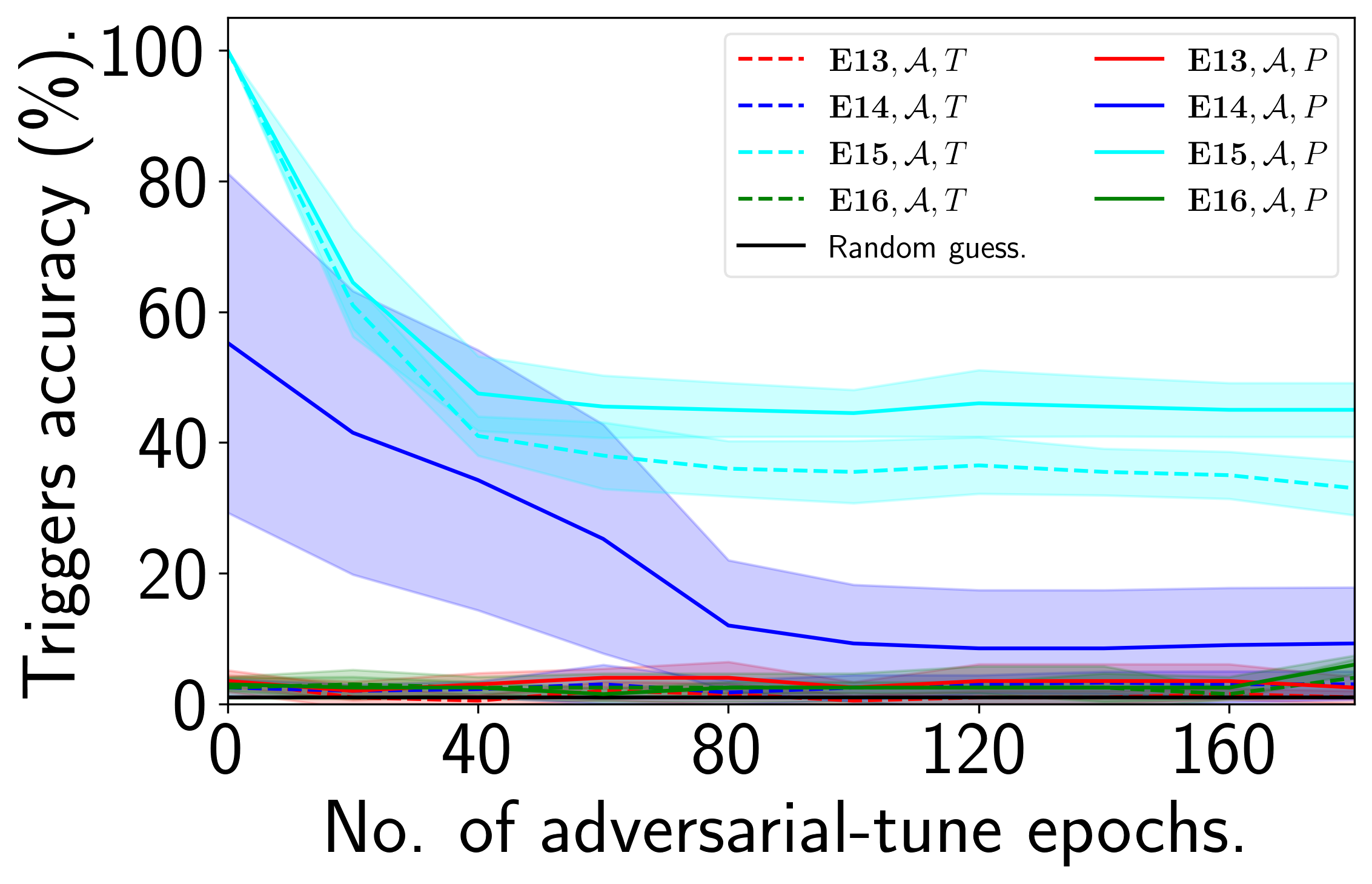}
\end{minipage}%
}%
\caption{The classification of (post-)triggers after adversarial tuning.}
\label{fig:adv}
\end{figure*}

\subsection{Watermark Injection}
The metric of the most interest is the impact of watermark injection on the DNN's classification accuracy on normal inputs.
During watermark injection defined by Eq.~\eqref{equation:lossembed}, we fixed $N$=$50$, $\lambda_{2}$=$5$, and $S$=$10\cdot N$ anchors were injected along with triggers.
The same $S$ random training samples were injected along with the trigger using the ordinary cross-entropy loss function for comparison.
In generating post-triggers, we set $Q$=$30$ and $R$=$100$.
The watermarked DNN's classification accuracy on the test dataset when the accuracy on the triggers reached 90\% was recorded, results are shown as Fig.~\ref{fig:exp1}, each entry contains the statistics of twenty folds of experiments.
Watermark injection along with anchors exerted uniformly the least impact on the DNN's performance as has been predicted, while injecting backdoor solely always devastated the model. 
Injection with training samples is not always the optimal choice, which fact indicates that tuning the model with poisoned datasets is a risky behavior unless we adopt delicate triggers or obtain extra knowledge on the configuration of the optimizer.
As a result, even equipped with knowledge of the training dataset, the IP manager is encouraged to inject triggers along with anchors and $l_{1}$ logit loss, which better preserves the commercial functionality of the DNN product.

We noticed that impacts of watermarking were different across different encoders, datasets, and network architectures, yet the ultimate performance with anchors injection remained the optimal level. 
The post-triggers generally preserved better performance than their counterpart triggers, since they have been tuned by Eq.~\eqref{equation:advloss} so their injection was less afflictive. 
This advantage is also reflected by the time consumption of watermark injection (averaged among four encoders) as shown in Fig.~\ref{figure:time}.
Injection with the original dataset took more time to stabilize the network, while the time cost in generating post-triggers could be compensated since they have been tuned to fit their respective labels.

\subsection{The Persistency of the Watermark}
Under the knowledge-independent assumption \textbf{(C1)}, we fixed the watermarking scheme's configuration among $(\mathcal{A},T)$, $(\mathcal{A},P)$ and examined the persistency of the injected watermark.
We firstly applied fine-tuning to the watermarked model and the classification accuracy of the triggers under a twenty-folded experiment is shown in Fig.~\ref{fig:ft}, from which we concluded that the impact of fine-tuning for triggers is uniformly larger than that for post-triggers.
This is because the post-triggers have been tuned to resist fine-tuning with anchors by Eq.~\eqref{equation:advloss}, such immunity is partially transferred to that against fine-tuning with the actual training dataset.
Even if an adversary has acquired the entire training dataset, it can hardly compromise the ownership proof by fine-tuning since the accuracy was uniformly above the random guess baseline.
We also found that (post-)triggers generated by the DFD generator preserved the optimal robustness against fine-tuning across different datasets and network architectures. 

\begin{figure}[!t]
\subfigure[1 epoch.]{
\begin{minipage}[htbp]{0.25\linewidth}
\includegraphics[width=2.2cm]{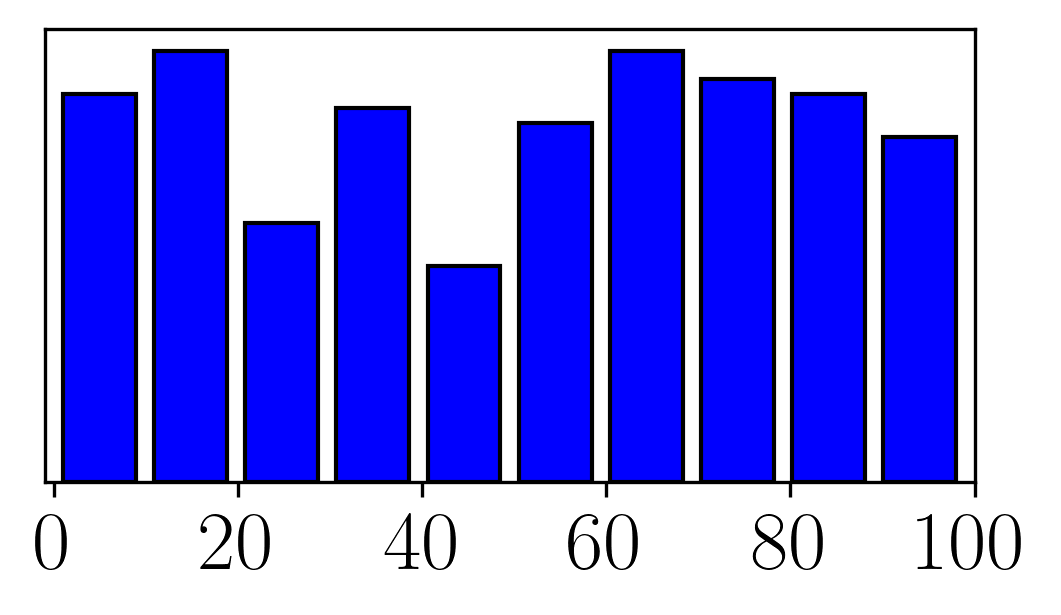}
\end{minipage}%
}%
\subfigure[10 epoches.]{
\begin{minipage}[htbp]{0.25\linewidth}
\includegraphics[width=2.2cm]{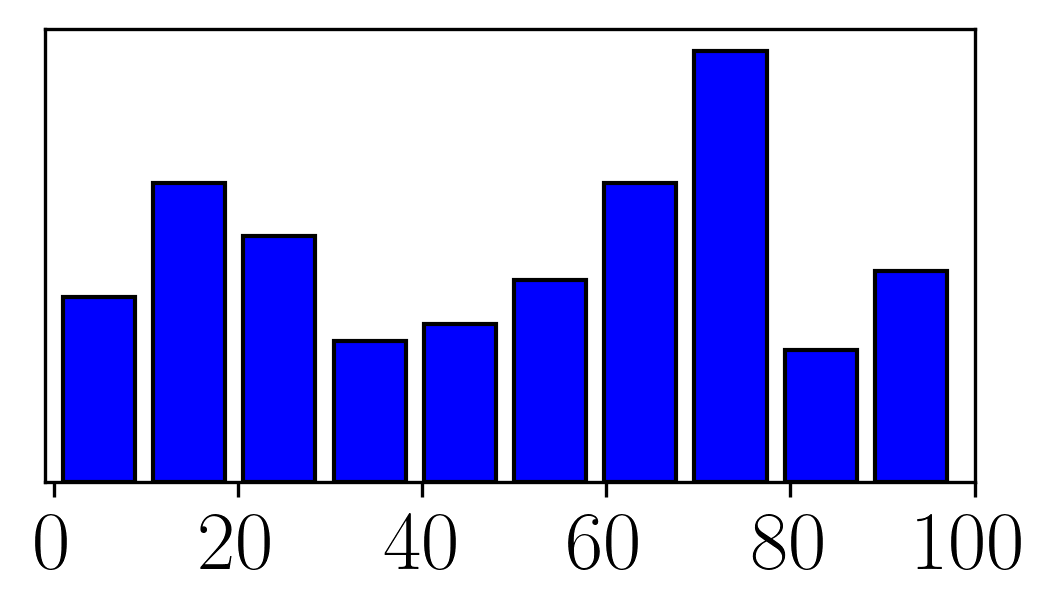}
\end{minipage}%
}%
\subfigure[20 epoches.]{
\begin{minipage}[htbp]{0.25\linewidth}
\includegraphics[width=2.2cm]{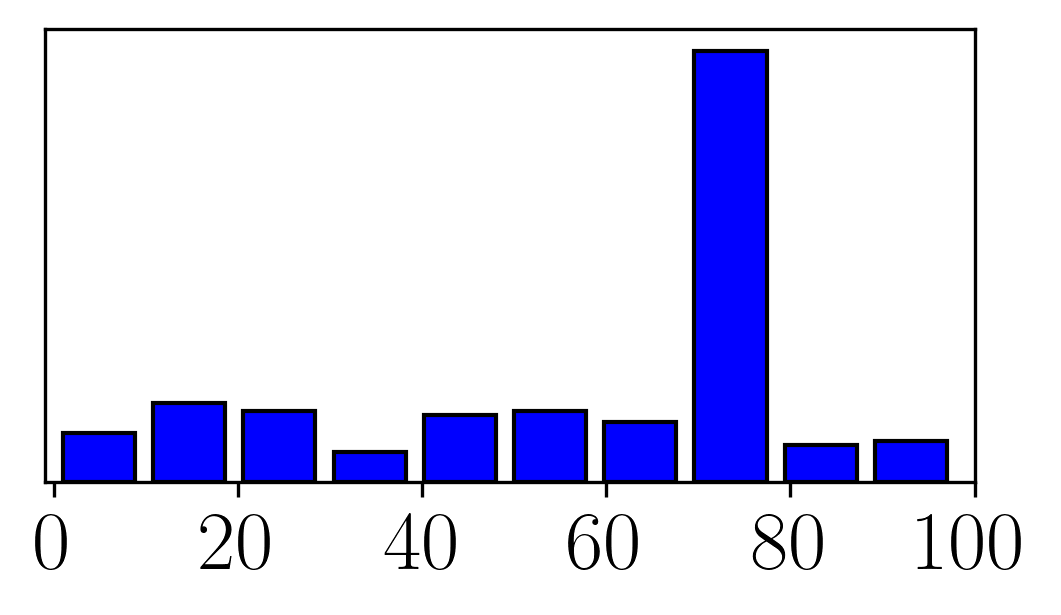}
\end{minipage}%
}%
\subfigure[30 epoches.]{
\begin{minipage}[htbp]{0.25\linewidth}
\includegraphics[width=2.2cm]{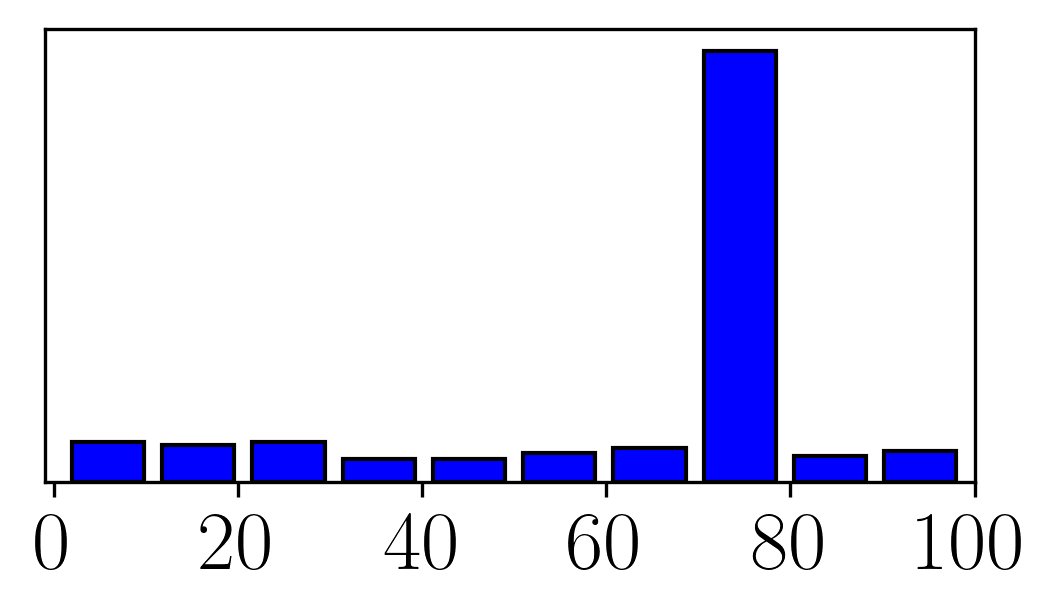}
\end{minipage}%
}
\caption{The distribution of predicted labels for a unseen set of triggers after $\left\{0,10,20,30 \right\}$ epochs of adversarial tuning.
The configuration is \textbf{E12}, $\mathcal{D}$, $T$, $c^{\text{adv}}$=$70$.}
\label{figure:advtune}
\end{figure}

\begin{figure}[!t]
\subfigure[MNIST.]{
\begin{minipage}[htbp]{0.5\linewidth}
\centering
\includegraphics[width=4.2cm]{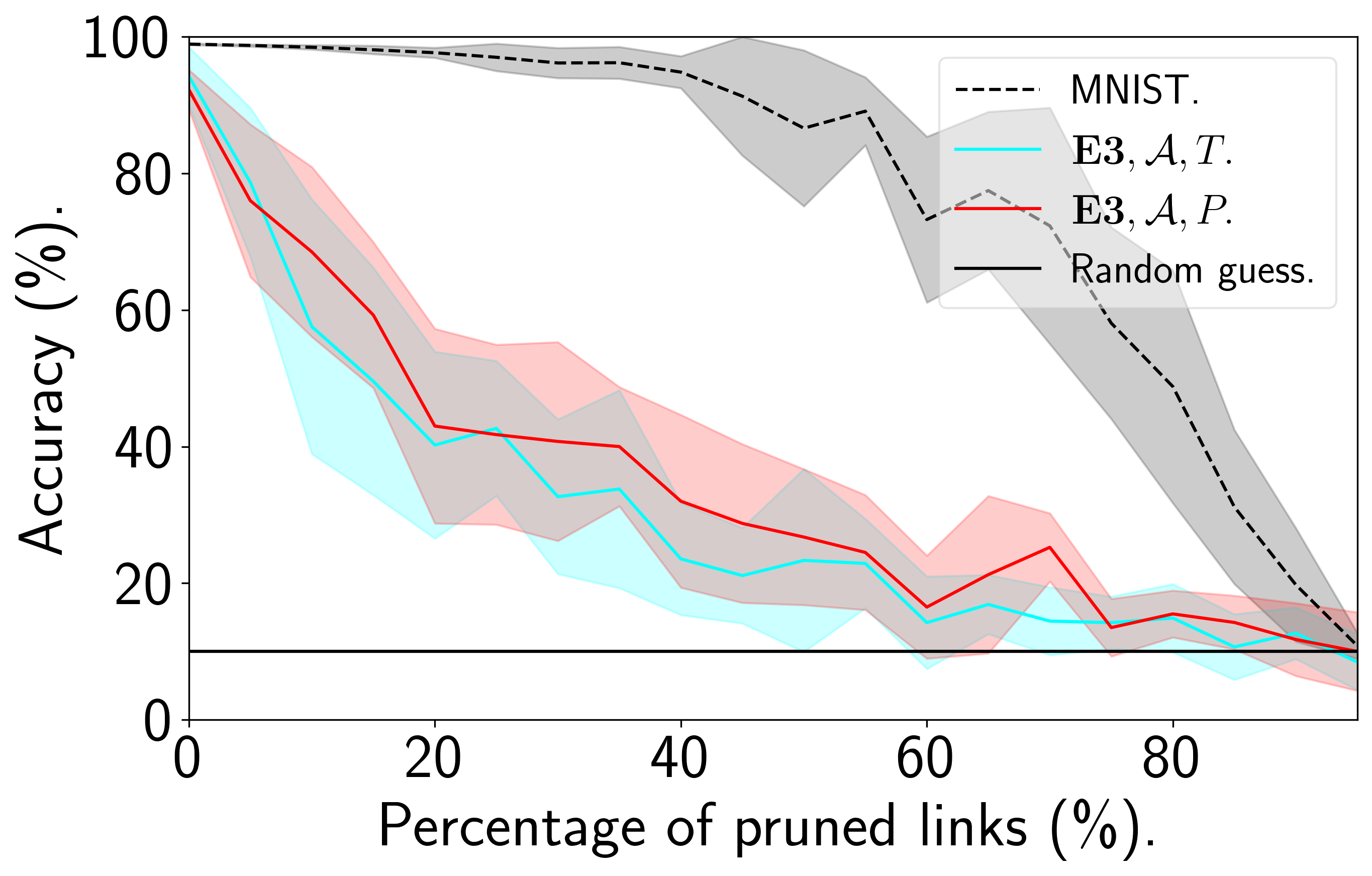}
\end{minipage}%
}%
\subfigure[CIFAR10.]{
\begin{minipage}[htbp]{0.5\linewidth}
\centering
\includegraphics[width=4.2cm]{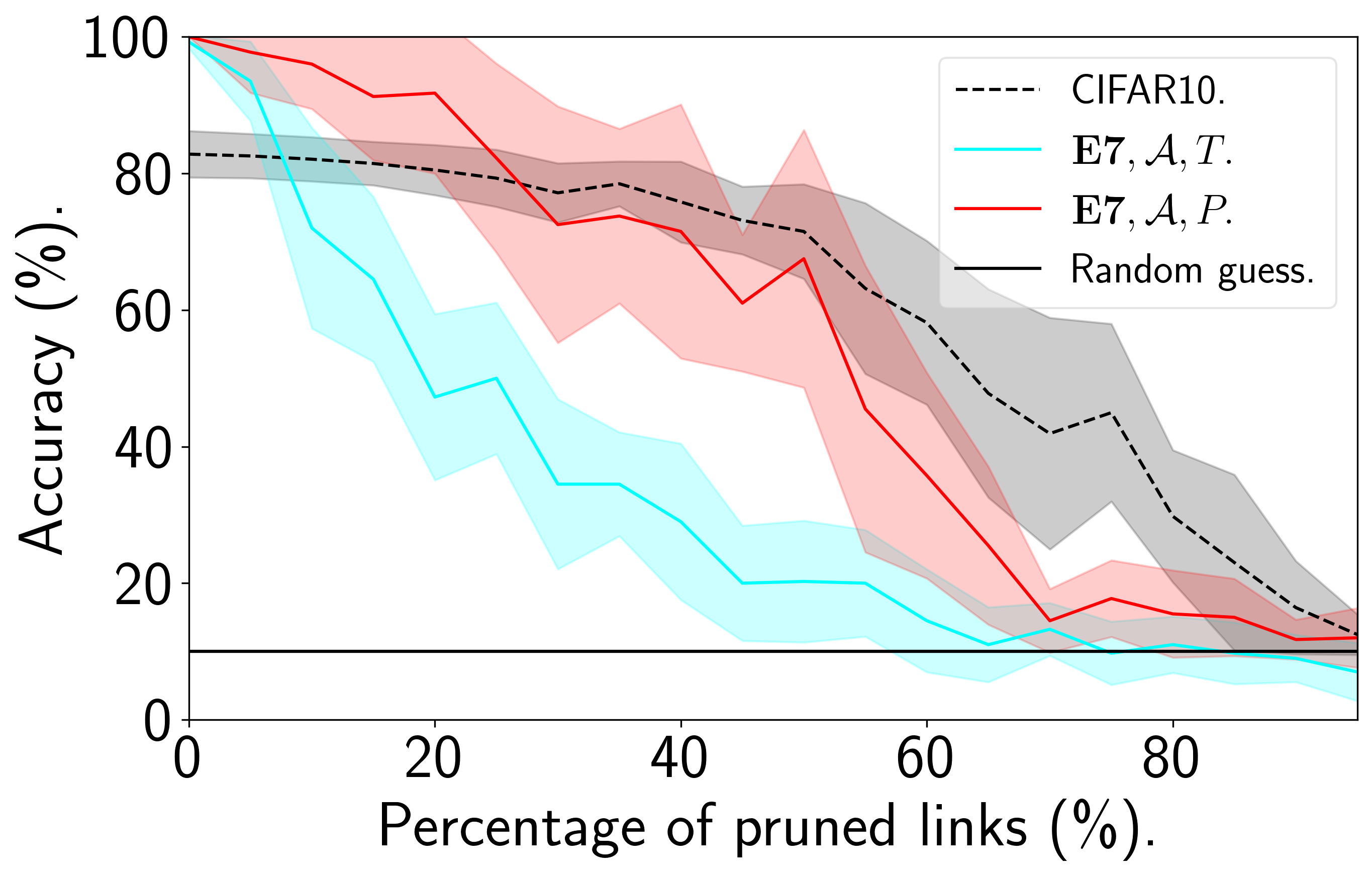}
\end{minipage}%
}
\subfigure[CIFAR100.]{
\begin{minipage}[htbp]{0.5\linewidth}
\centering
\includegraphics[width=4.2cm]{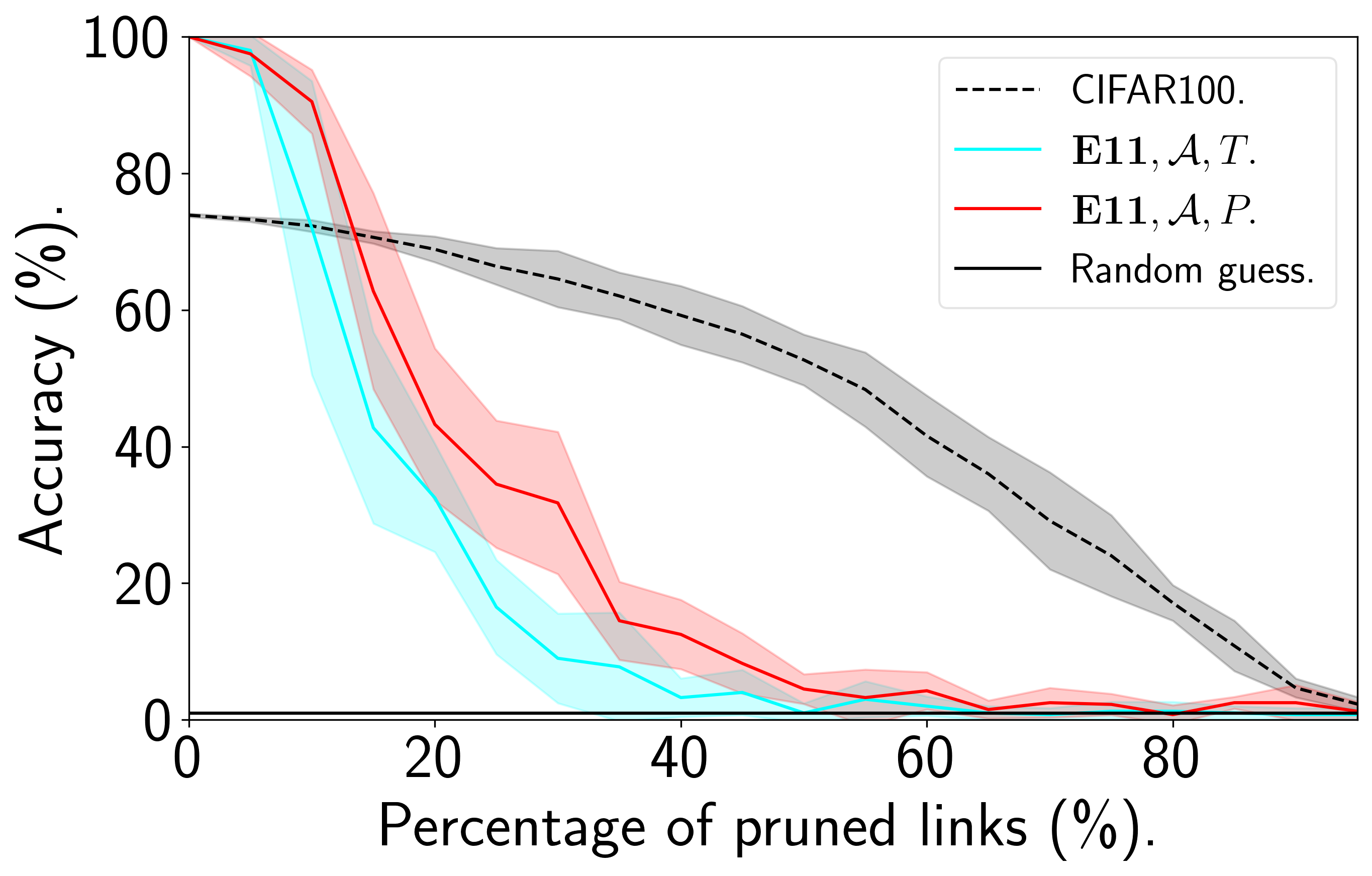}
\end{minipage}%
}%
\subfigure[Caltech101.]{
\begin{minipage}[htbp]{0.5\linewidth}
\centering
\includegraphics[width=4.2cm]{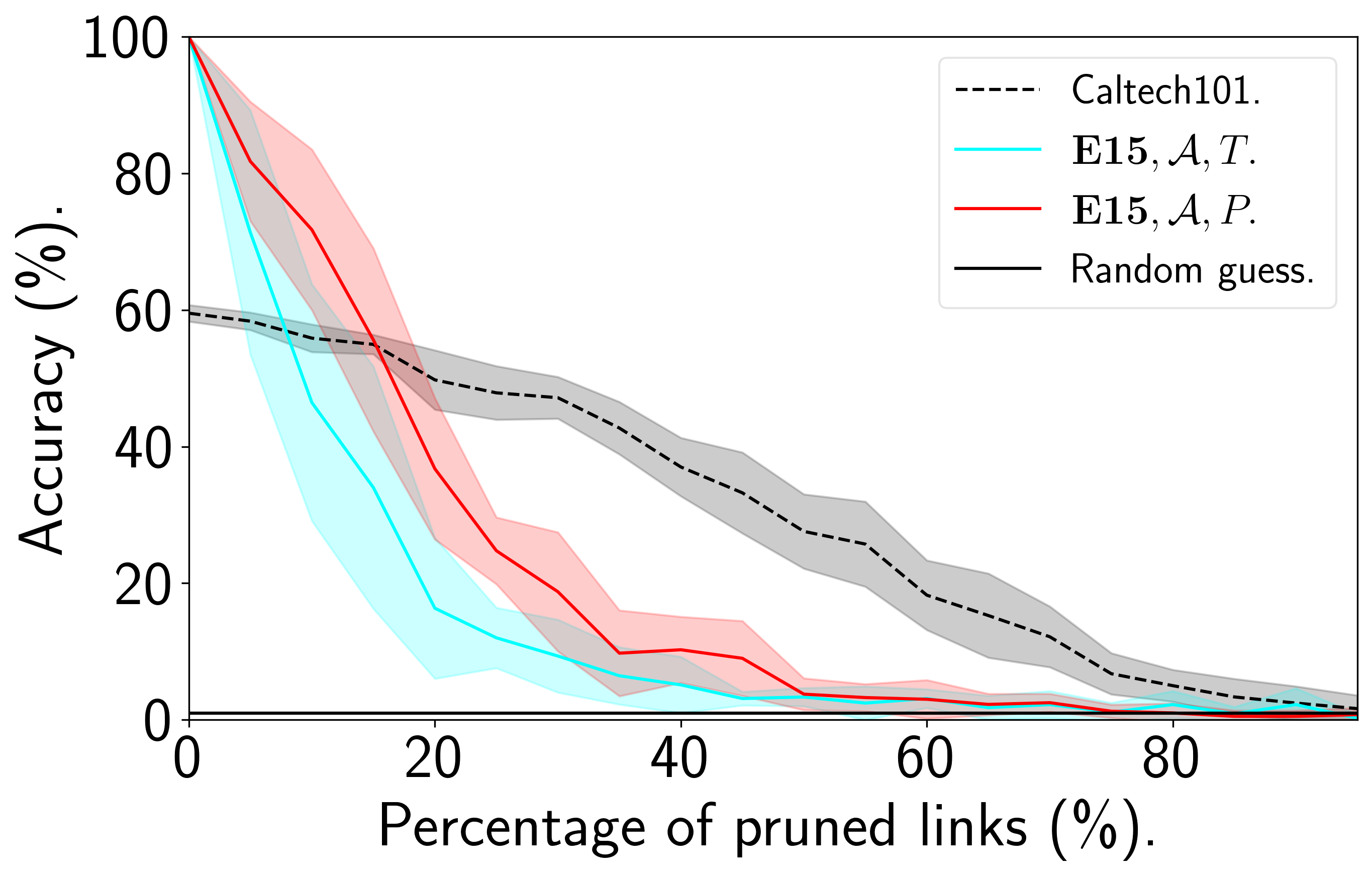}
\end{minipage}%
}
\caption{The change of classification accuracy on normal inputs, triggers, and post-triggers under neuron pruning. 
The configuration is DFD encoder with injection scheme $\mathcal{D}$.}
\label{fig:np}
\end{figure}

Nextly, we considered an adversary with knowledge of the watermarking scheme and evidence.
Concretely, such an adversary has access to the verifier program Algo.~\ref{algorithm:OV}. 
Having acquired the image encoder $T(\cdot)$, the adversary can conduct adversarial tuning to spoil the injected triggers.
In particular, the adversary uses $T(\cdot)$ to generate a series of triggers, assigns them with a random label, and tunes the watermarked DNN on this forged trigger set $\mathcal{B}^{\text{adv}}=\left\{(T(\textbf{u}_{w}^{\text{adv}}),c^{\text{adv}}) \right\}_{w=1}^{W}$ with $\textbf{u}_{w}^{\text{adv}}$ randomly sampled from ${U}$.
This adversarial tuning teaches the DNN to attribute all triggers into one specific class $c^{\text{adv}}$ and thence spoils the watermark.
We applied this adaptive attack and the accuracy of (post-)triggers in the spoiled DNN was shown in Fig.~\ref{fig:adv}.
Although all four watermarking schemes remained valid against the fine-tuning, their persistency was at risk against the adversarial tuning.
In many cases, the classification accuracy on triggers fell significantly around the random guess baseline, so the IP protection was compromised. 
This damage was particularly harmful for WF, whose triggers' pattern is too naive. 
By complicating the pattern of the backdoor, e.g., adopting the DFD generator and using post-triggers, the watermark's persistency against a knowledgeable adversary increased significantly.
Notice that the adversarial tuning simultaneously damages the DNN's performance on normal inputs, yet this attack remains practical since the adversary only needs to deploy the adversarially tuned DNN as a backdoor filter (but not the backend of the service) that blocks inputs with predicted label $c^{\text{adv}}$.
By correctly recognizing triggers as shown in Fig.~\ref{figure:advtune}, such filtering can invalidate the ownership proof by sacrificing the masked intact DNN's performance for at most $\frac{1}{C}$.

We also evaluated the robustness against pruning, where an adversary risks the pirated DNN's performance to invalidate ownership proof.
The decline of the DNN's classification accuracy on normal inputs and (post-)triggers is visualized in Fig.~\ref{fig:np}.
The decline of DNN's performance for four datasets when its accuracy triggers/post-triggers dropped below the random guess baseline is summarized in Table.~\ref{table:prune}. 
Such sacrifice is intolerable for commercial DNN products.

\begin{table}[!t]
\centering
\caption{The decline of DNN's performance when the ownership evidence is invalidated.}
\scalebox{0.8}{
\begin{tabular}{c|c|c|c|c}
\toprule
\textbf{Backdoor} & MNIST & CIFAR10 & CIFAR100 & Caltech101 \\
\toprule
$T$ & 12.3\% & 24.7\% & 11.8\% & 22.5\% \\
$P$ & 25.7\% & 40.9\% & 25.5\% & 41.3\% \\
\bottomrule
\end{tabular}}
\label{table:prune}
\end{table}

Finally, we visualize the $l_{2}$ distance between triggers and post-triggers under different $\lambda_{1}$s in Fig.~\ref{fig:conf}.
For each encoder, we produced 30 triggers, 30 corresponding post-triggers, and another 30 independent triggers.
The threshold $\epsilon$ in Eq.~\eqref{equation:fuzzyOV} for a fixed $\lambda_{1}$ was set to the maximal $l_{2}$ distance between a trigger and its post-trigger w.r.t. all encoders.
It is observed that no confusion was detected, even if the threshold was increased five times.
Therefore, the OV can be conducted safely and unambiguously using post-triggers and the original encoder.

\begin{figure}[!t]
\subfigure[$\lambda_{1}=0.15$.]{
\begin{minipage}[htbp]{0.25\linewidth}
\includegraphics[width=1.9cm]{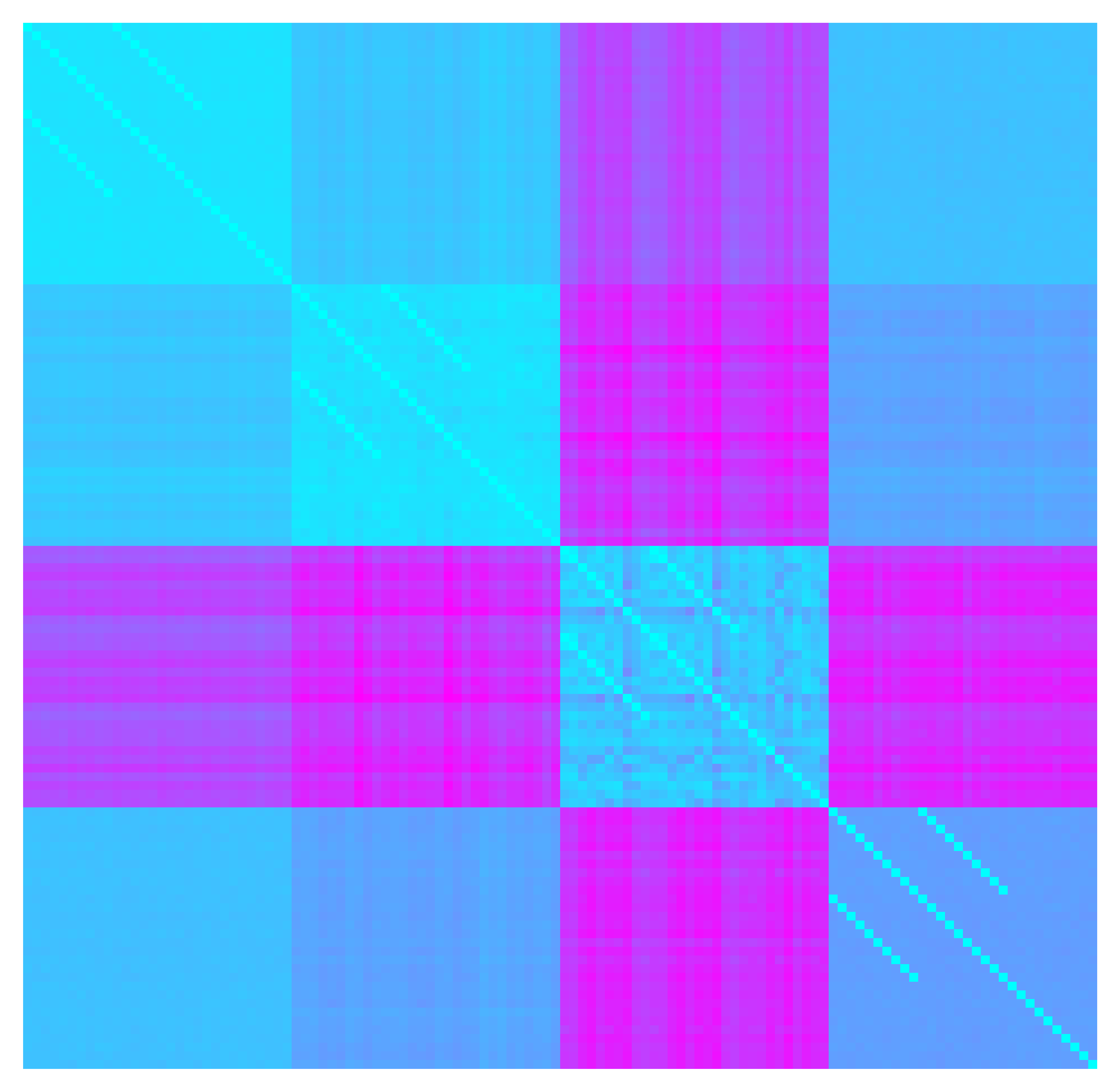}
\end{minipage}%
}%
\subfigure[$5\epsilon$.]{
\begin{minipage}[htbp]{0.25\linewidth}
\includegraphics[width=1.9cm]{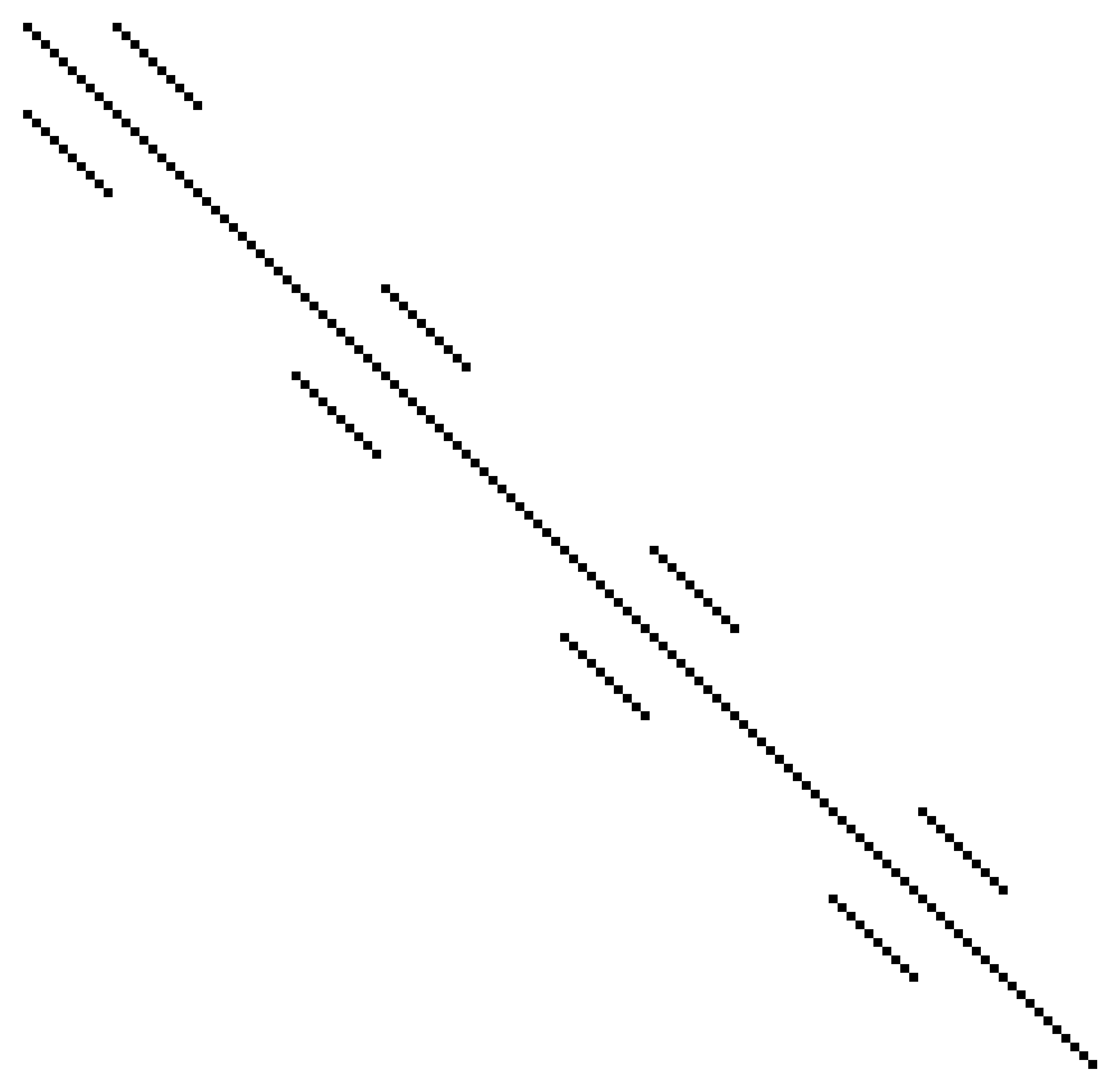}
\end{minipage}%
}%
\subfigure[$10\epsilon$.]{
\begin{minipage}[htbp]{0.25\linewidth}
\includegraphics[width=1.9cm]{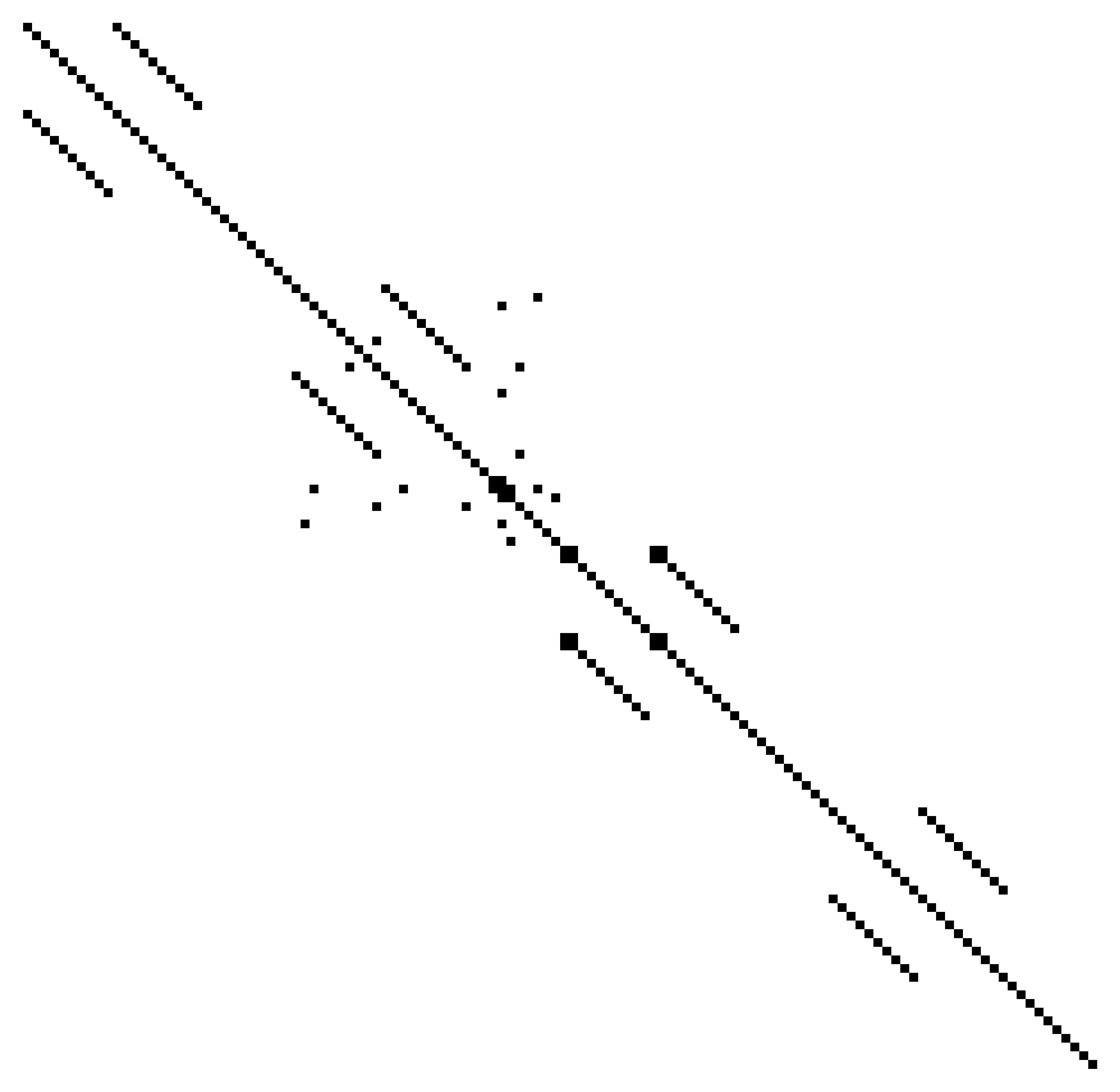}
\end{minipage}%
}%
\subfigure[$15\epsilon$.]{
\begin{minipage}[htbp]{0.25\linewidth}
\includegraphics[width=1.9cm]{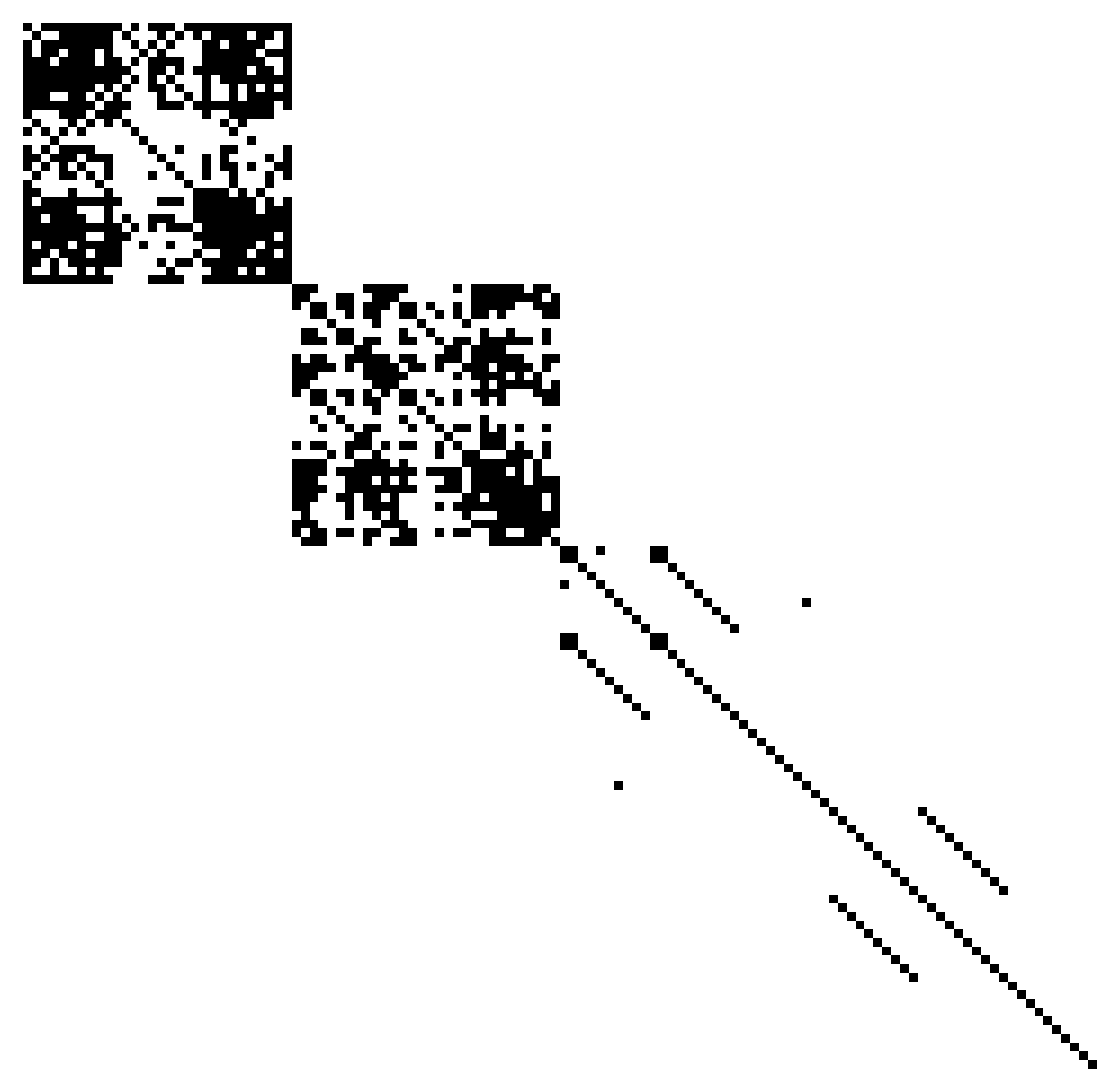}
\end{minipage}%
}
\subfigure[$\lambda_{1}=0.5$.]{
\begin{minipage}[htbp]{0.25\linewidth}
\includegraphics[width=1.9cm]{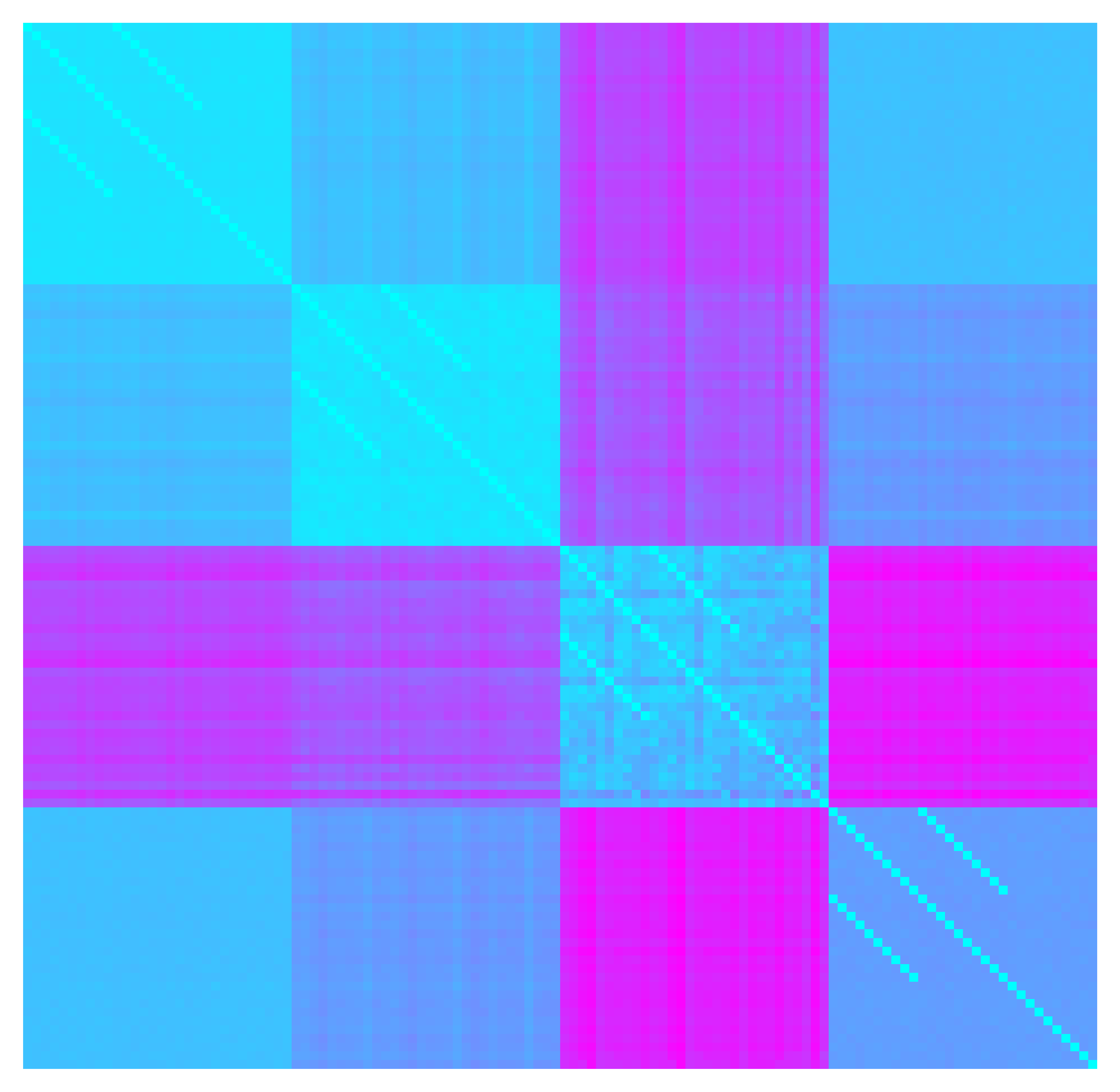}
\end{minipage}%
}%
\subfigure[$5\epsilon$.]{
\begin{minipage}[htbp]{0.25\linewidth}
\includegraphics[width=1.9cm]{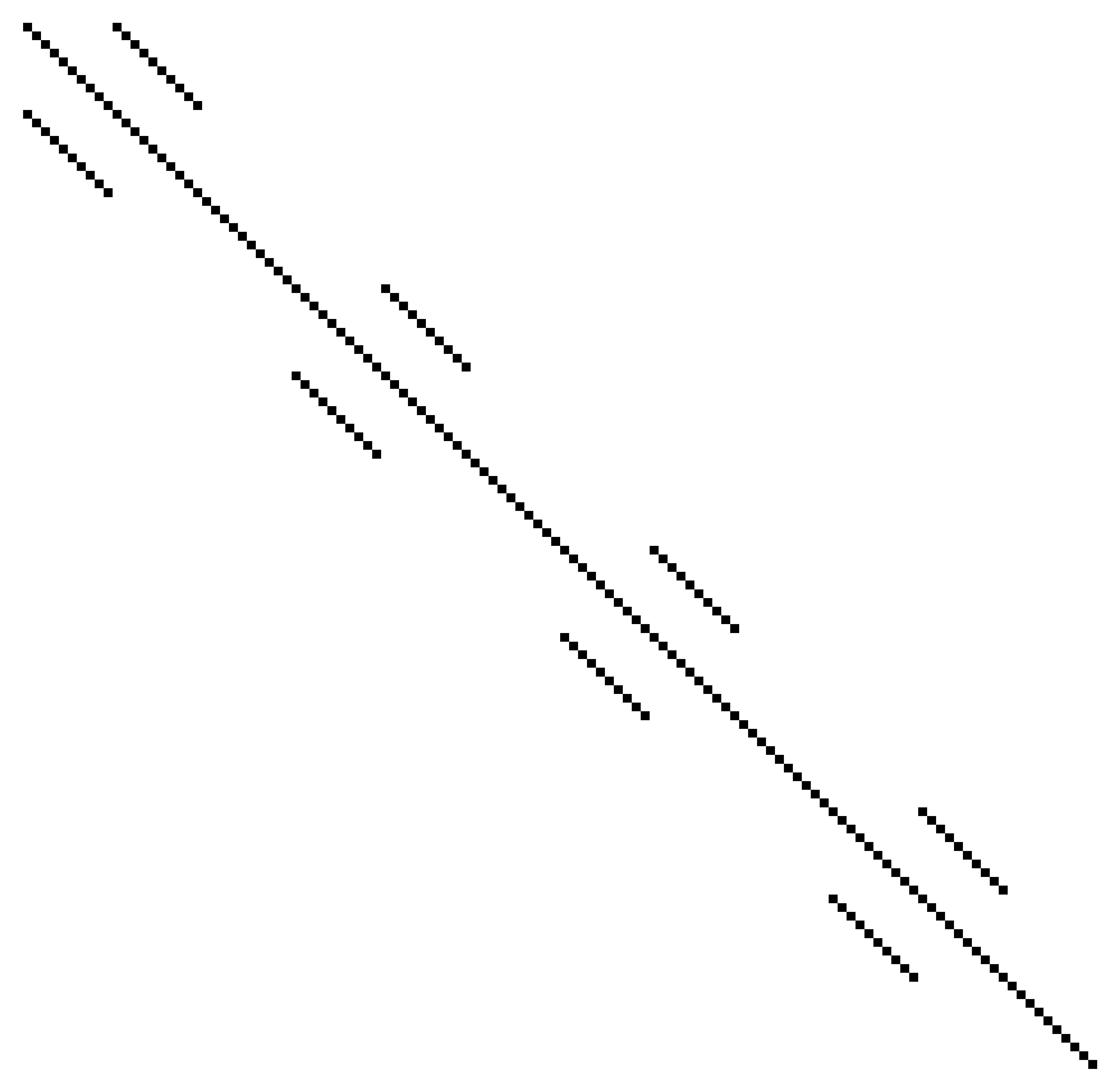}
\end{minipage}%
}%
\subfigure[$10\epsilon$.]{
\begin{minipage}[htbp]{0.25\linewidth}
\includegraphics[width=1.9cm]{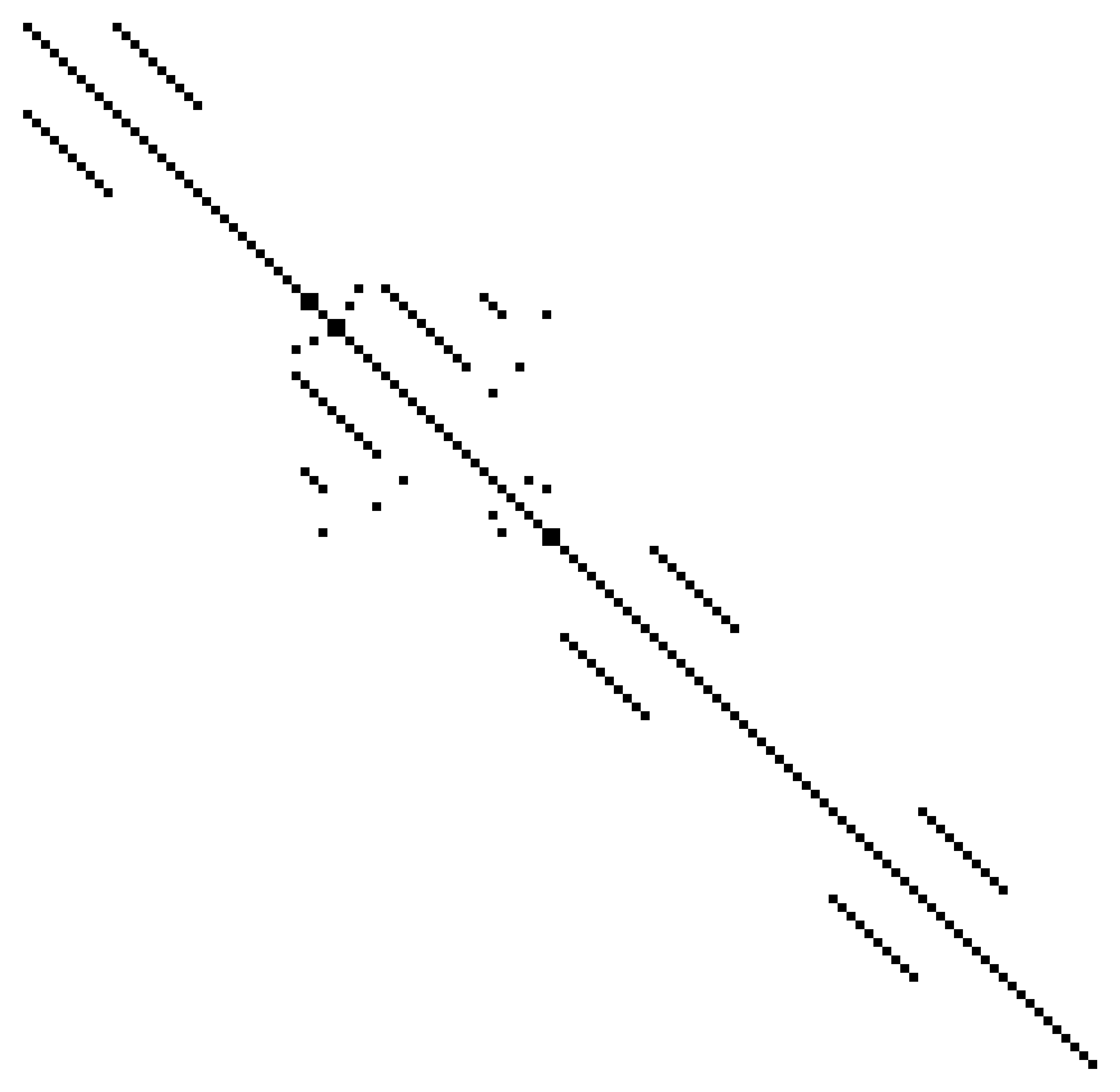}
\end{minipage}%
}%
\subfigure[$15\epsilon$.]{
\begin{minipage}[htbp]{0.25\linewidth}
\includegraphics[width=1.9cm]{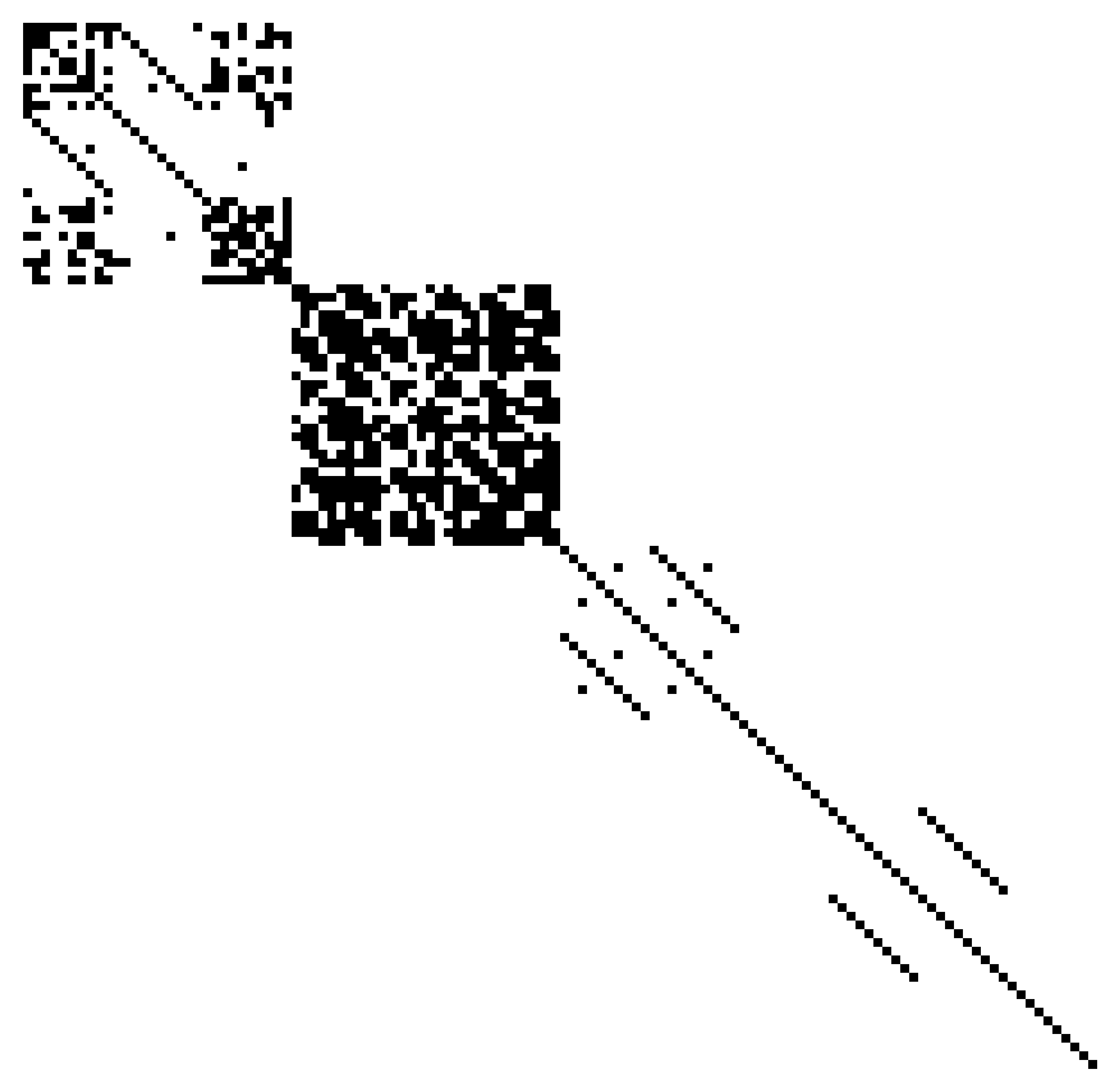}
\end{minipage}%
}
\subfigure[$\lambda_{1}=1$.]{
\begin{minipage}[htbp]{0.25\linewidth}
\includegraphics[width=1.9cm]{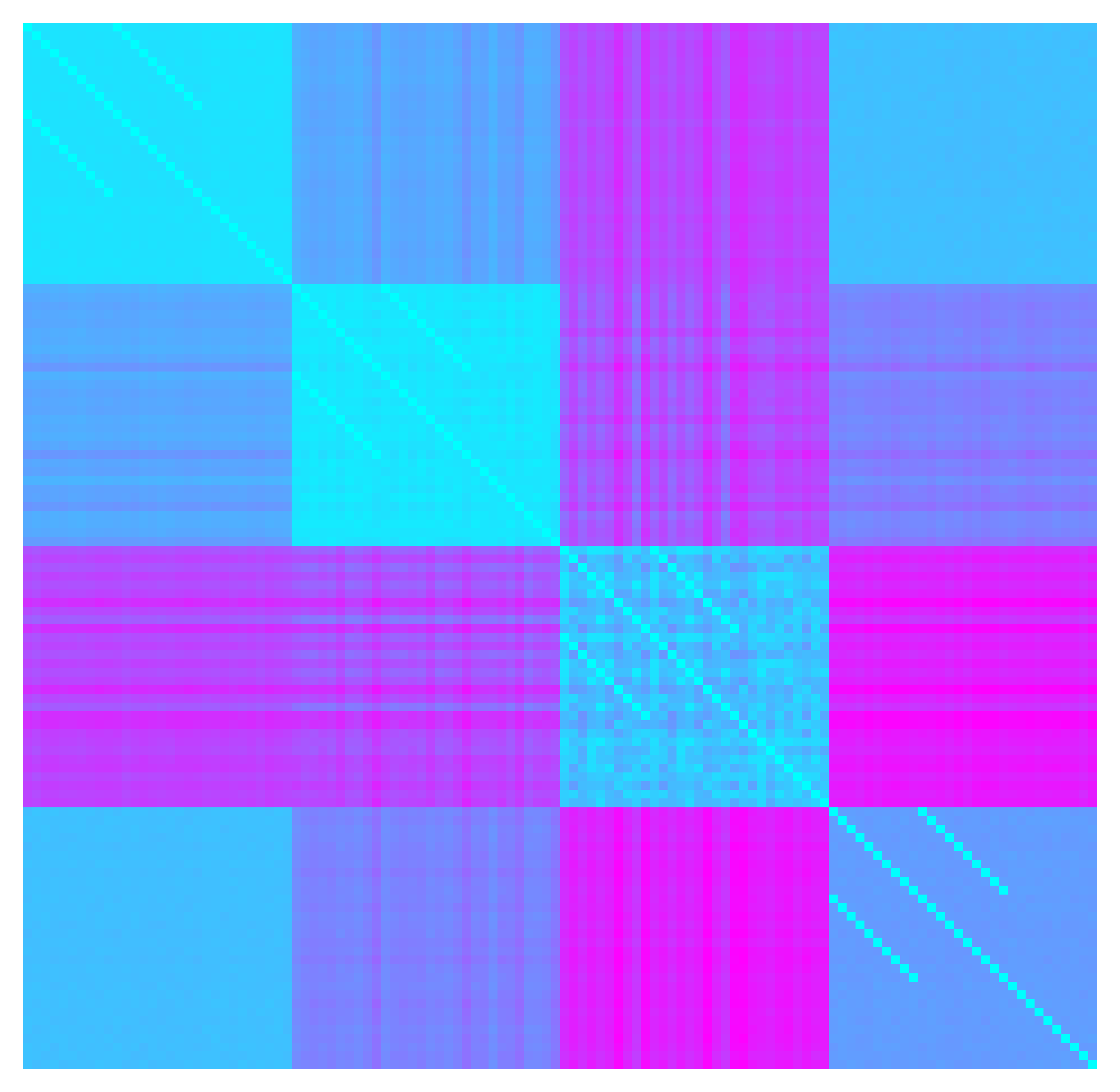}
\end{minipage}%
}%
\subfigure[$5*\epsilon$.]{
\begin{minipage}[htbp]{0.25\linewidth}
\includegraphics[width=1.9cm]{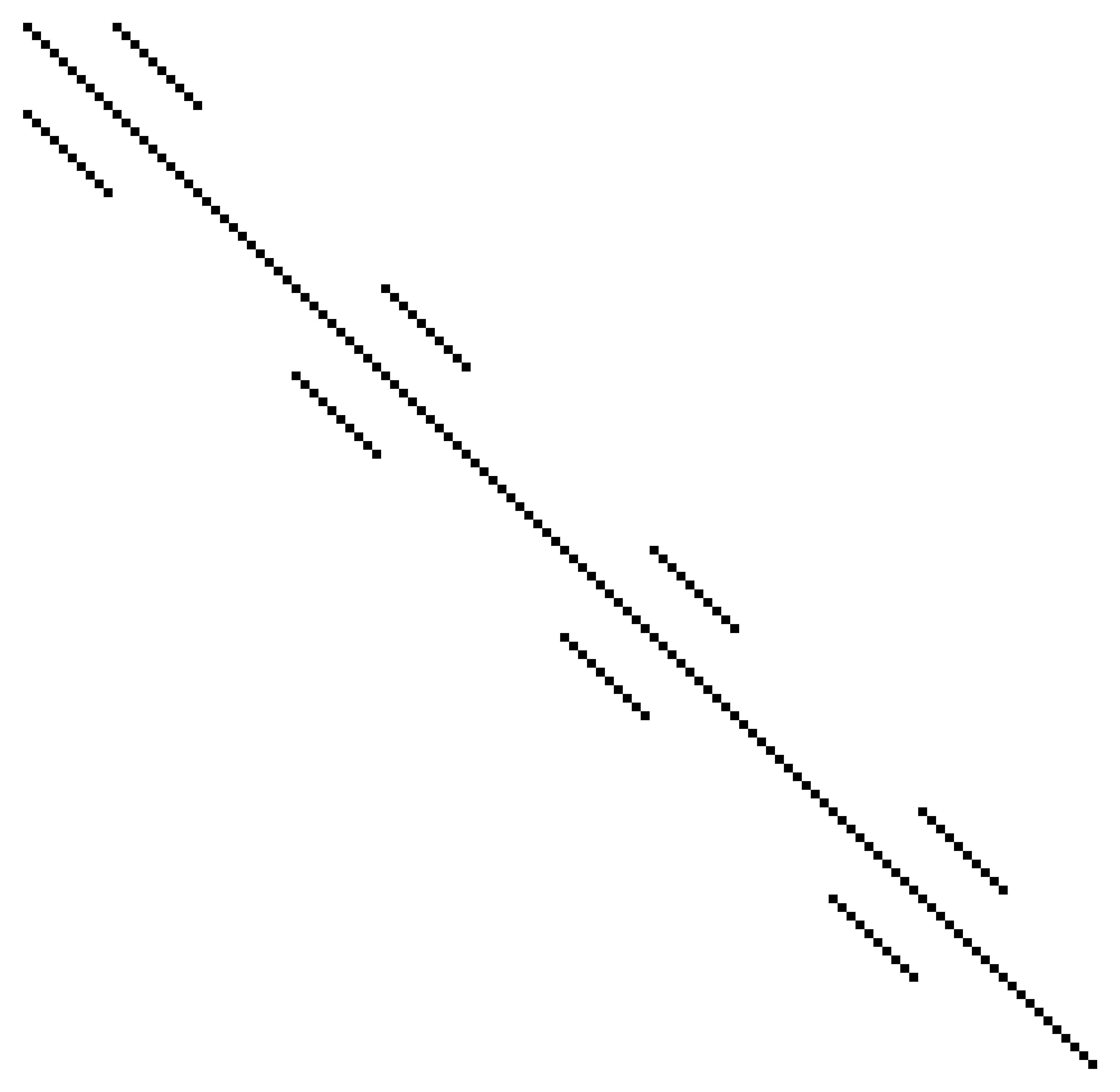}
\end{minipage}%
}%
\subfigure[$10*\epsilon$.]{
\begin{minipage}[htbp]{0.25\linewidth}
\includegraphics[width=1.9cm]{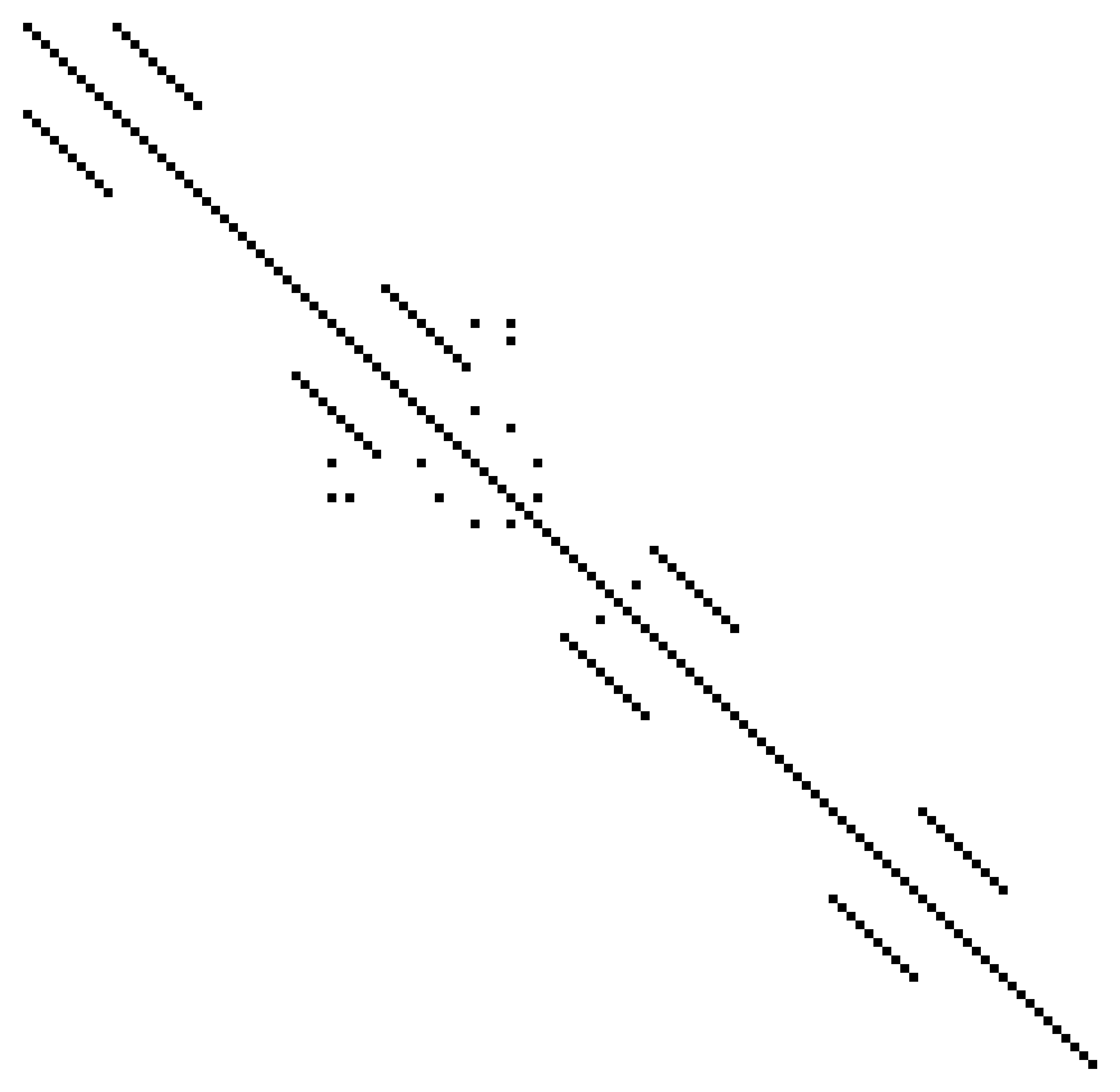}
\end{minipage}%
}%
\subfigure[$15*\epsilon$.]{
\begin{minipage}[htbp]{0.25\linewidth}
\includegraphics[width=1.9cm]{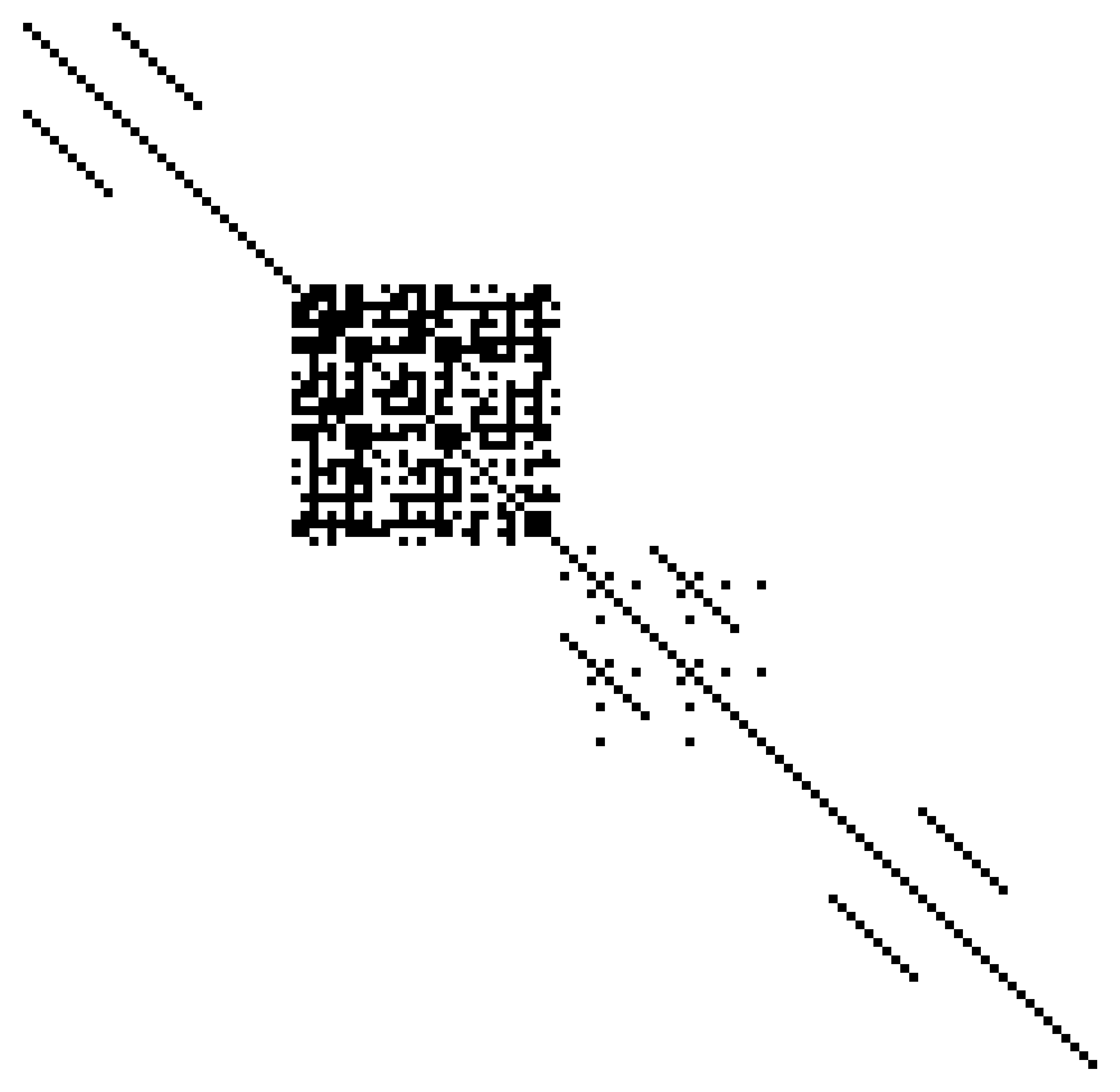}
\end{minipage}%
}%
\caption{The $l_{2}$ distance between triggers and post-triggers for different image encoders under $\lambda_{1}=0.15/0.5/1$ (a)/(e)/(i). The (post-)trigger that satisfies the fuzzy identical relationship defined by Eq.~\eqref{equation:fuzzyOV} are marked in (b)-(d)/(f)-(h)/(j)-(l).}
\label{fig:conf}
\end{figure}

\subsection{The Watermarking Capacity}
To delve into the watermarking capacity, we varied $N$ in range $[50,1000]$ with step 50, post-triggers were injected into the DNN with 50 ones as a group. 
We recorded the classification accuracy of the watermarked DNN on the normal test dataset and the lowest accuracy on post-trigger batches. 
The optimal configuration from previous discussions, the DFD encoder with injection scheme $\mathcal{D}$ and post-triggers $P$ was adopted and results are demonstrated in Fig.~\ref{fig:capacity}. 
The first conclusion is that when $N$ increased, the decline of the DNN's performance does not only rely on the dataset or the network architecture. 
For example, in MNIST and Caltech101, the increase of $N$ brought more damage. 
The second conclusion is that increasing $N$ has little influence to the accurate retrival of post-triggers, i.e., the second term in Eq.~\eqref{equation:bound}. 
From an empirical perspective, the bottleneck for watermarking capacity remains the damage of watermarking to the DNN's functionality. 

\begin{figure}[!t]
\subfigure[MNIST.]{
\begin{minipage}[htbp]{0.5\linewidth}
\centering
\includegraphics[width=4.2cm]{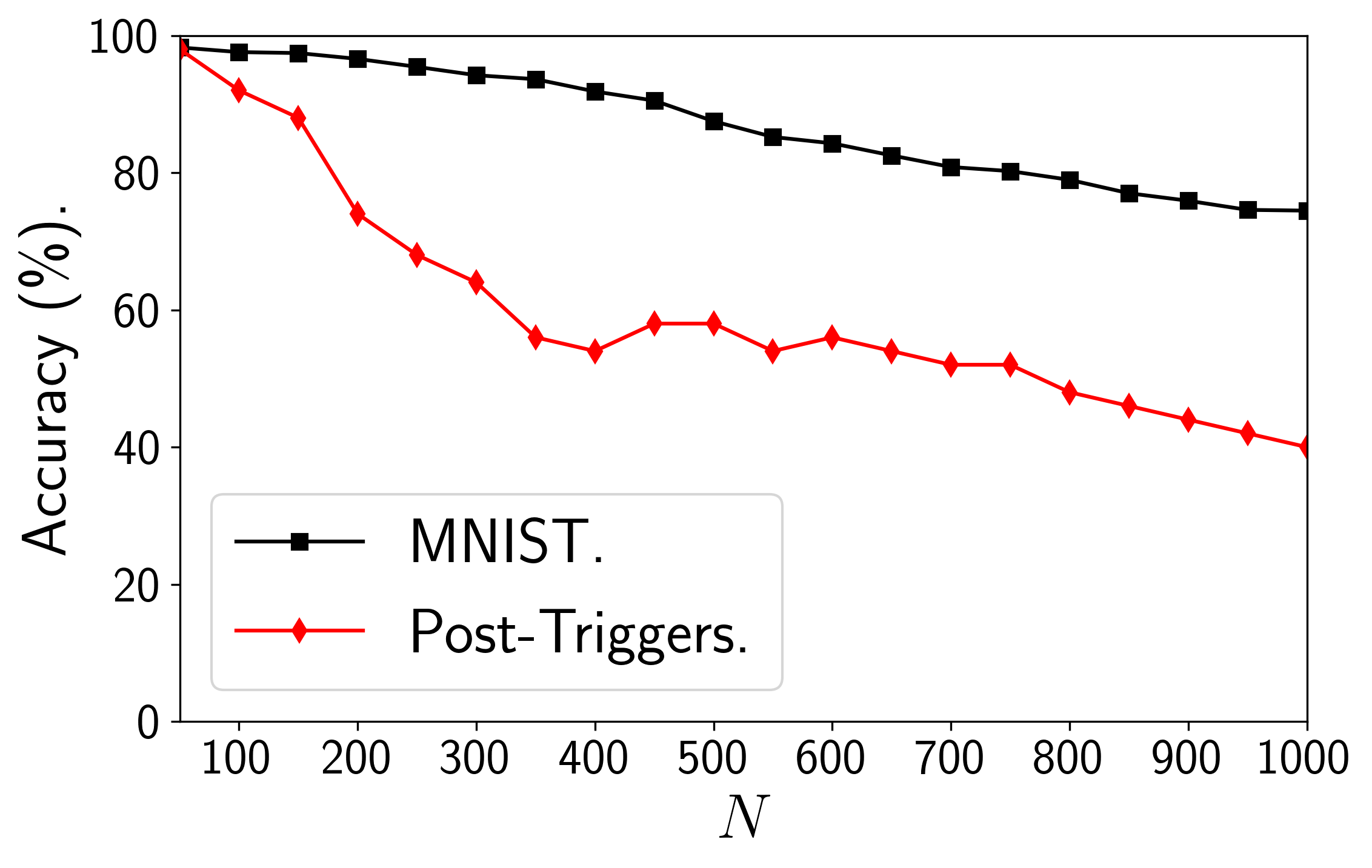}
\end{minipage}%
}%
\subfigure[CIFAR10.]{
\begin{minipage}[htbp]{0.5\linewidth}
\centering
\includegraphics[width=4.2cm]{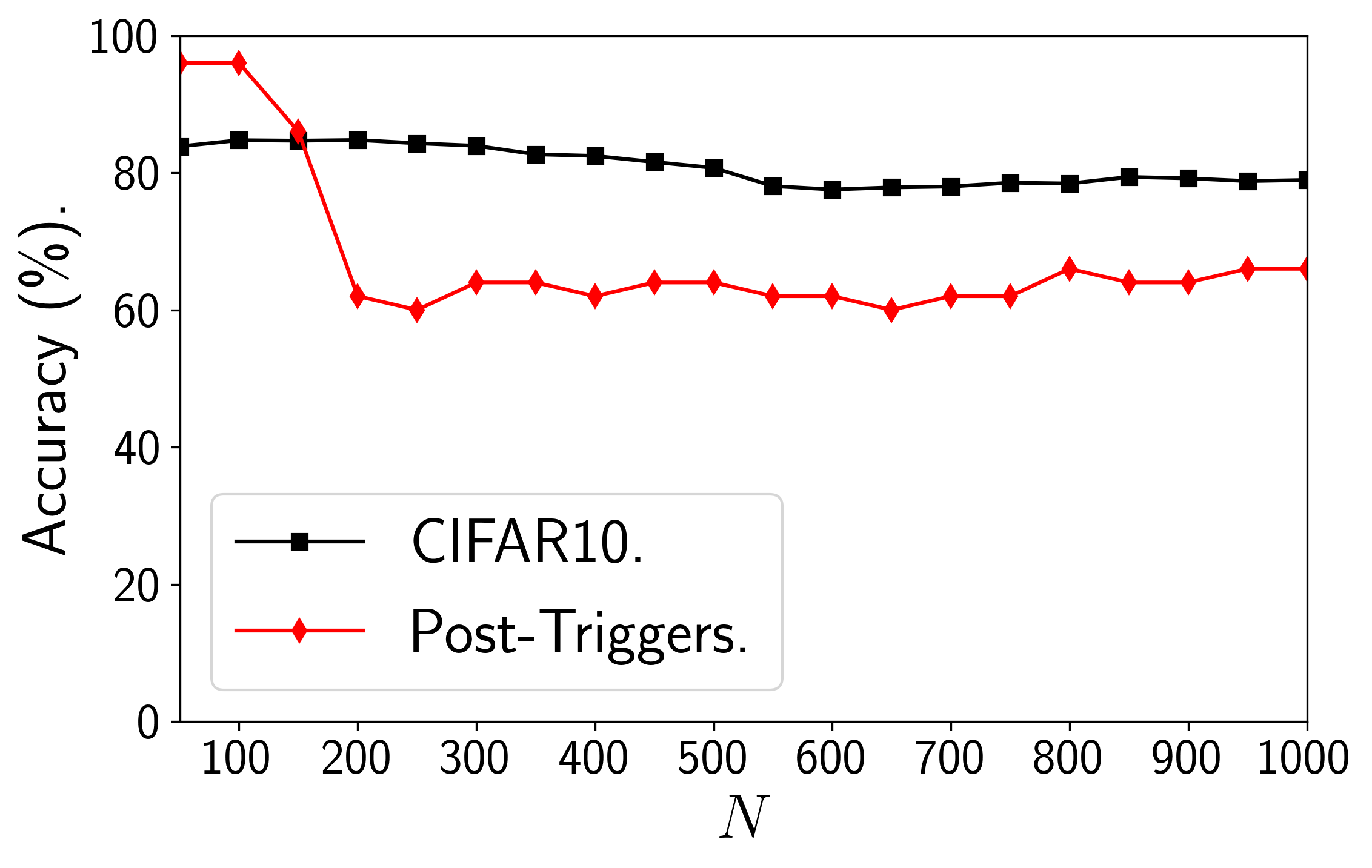}
\end{minipage}%
}
\subfigure[CIFAR100.]{
\begin{minipage}[htbp]{0.5\linewidth}
\centering
\includegraphics[width=4.2cm]{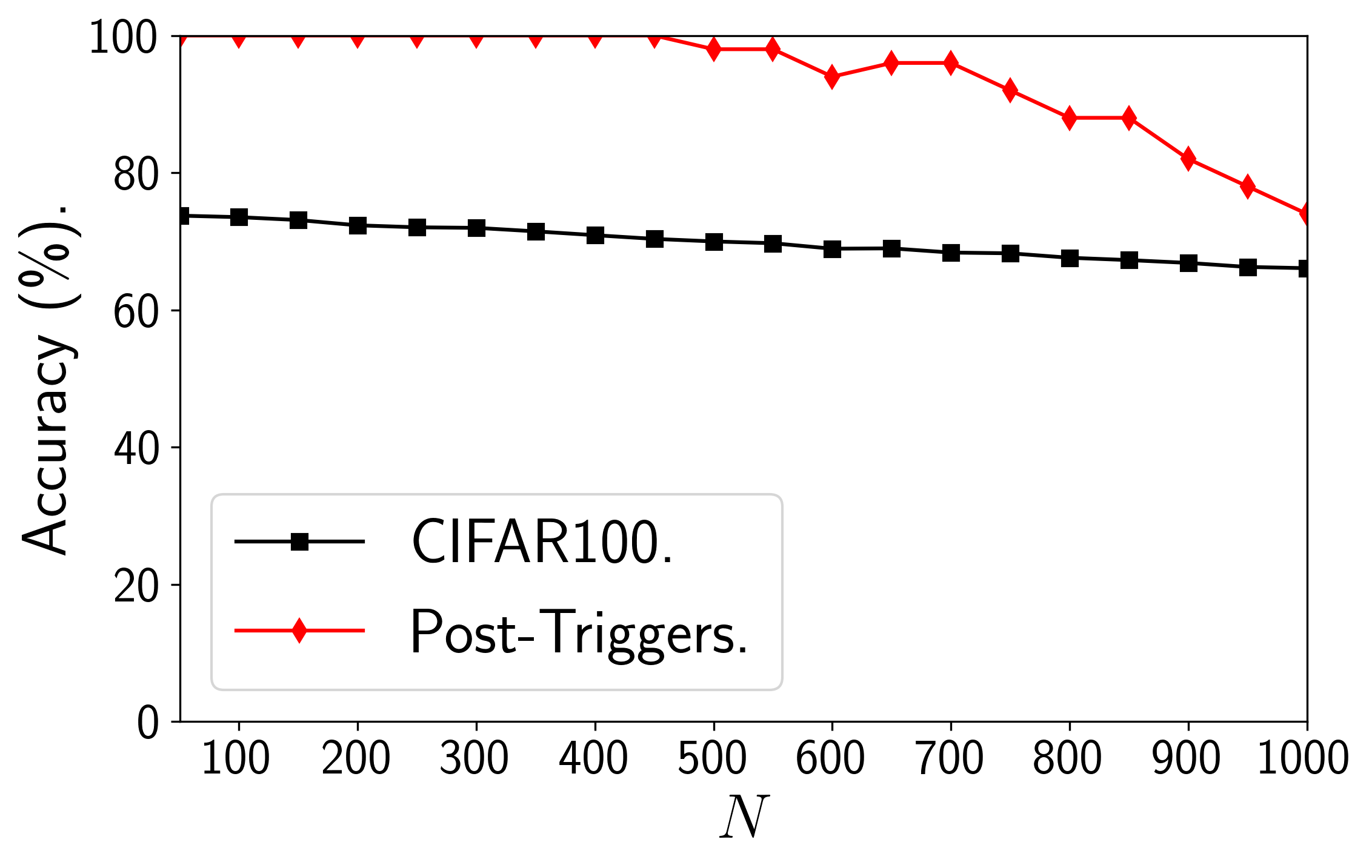}
\end{minipage}%
}%
\subfigure[Caltech101.]{
\begin{minipage}[htbp]{0.5\linewidth}
\centering
\includegraphics[width=4.2cm]{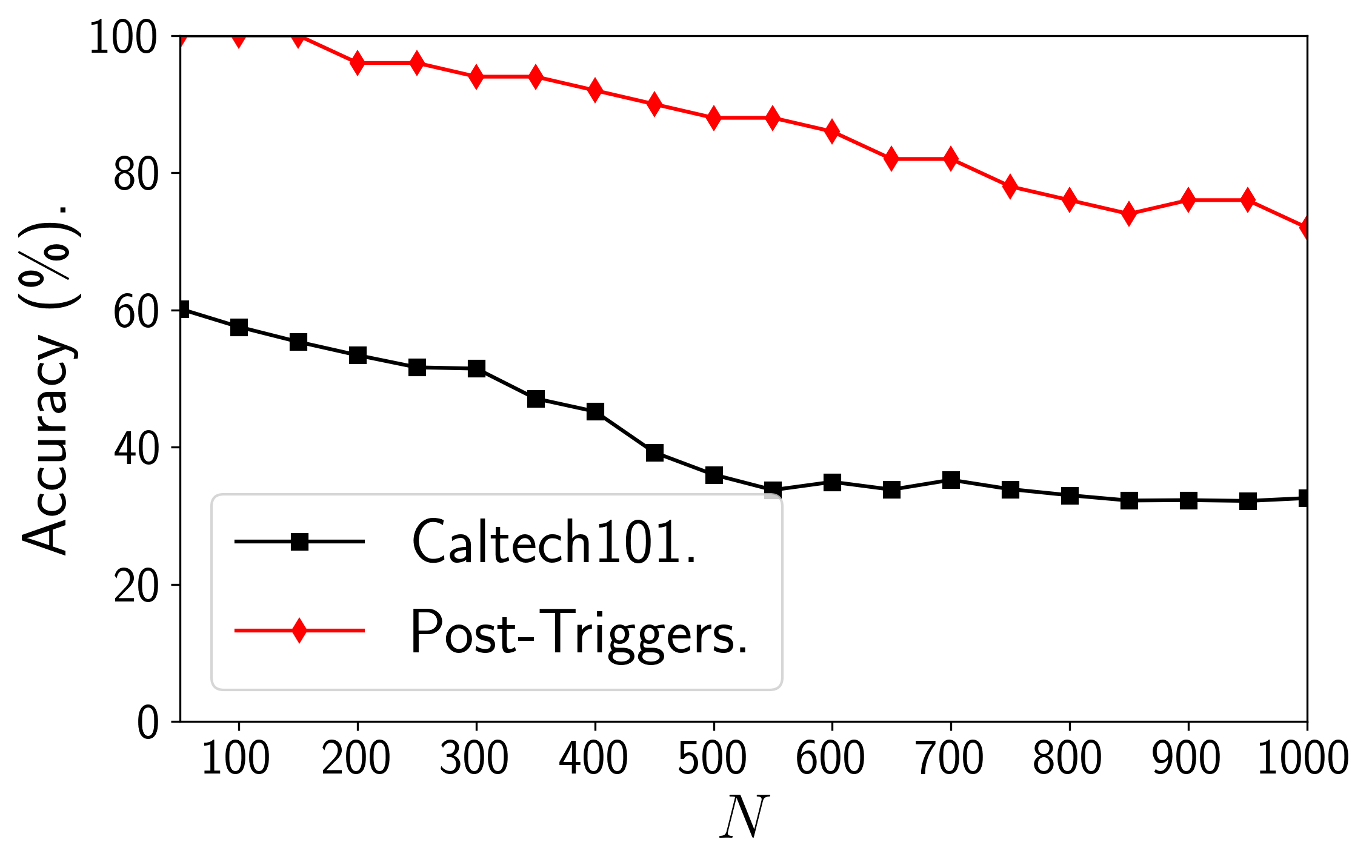}
\end{minipage}%
}
\caption{The accuracy on the test dataset and post-triggers for different $N$s.}
\label{fig:capacity}
\end{figure}

\subsection{Discussions}
Apart from injection along with anchors, regularizers can also be adopted to stabilize the DNN's performance during watermarking~\cite{ours}.
However, distinct regularizers and corresponding optimizers have to be designed for different DNN architectures, reducing their universality.

The general motivation behind performance-preserving DNN watermarking is to find the least perturbation to the network's decision boundary while encoding specific information.
The optimal choice is fixing the decision boundary on normal inputs and outputs.
When the knowledge of normal data is unavailable, the optimization target becomes finding a collection of samples with which the entire decision boundary can be preserved as intact as possible.
Notice that this optimization target is precisely that of DFD, therefore anchors given by the DFD generator outperform other candidates in preserving the DNN's performance during watermark injection. 

Remark that anouncing the generator would not reduce an adversary's cost in distilling the watermarked DNN. 
Fitting a new DNN with the watermarked DNN on samples generated by the generator would only yield a futile model. 
This is because during DFD, the generator and the student model are trained simultaneously, so the anchors are gradually shifted from simple images to hard images. 
Tuning a new model solely from hard images, i.e., anchors generated from the final generator when the DFD terminates, cannot result in a good student. 
Consequently, an adversary who distills the watermarked DNN still has to trade the model's overall performance for the freedom from IP regulation. 

\section{Conclusions}
\label{sec:5}
To regulate image classification DNN as intellectual properties in industrial applications, it is desirable that the watermarking scheme remains independent from the training dataset. 
Moreover, the security of ownership management should be built neither upon the secrecy of the watermarking scheme nor that of the channel in sending ownership evidence. 
In this paper, we propose a knowledge-free black-box watermarking scheme for image classification DNNs, together with the framework for ownership verification and post-trigger tuning. 
The proposed scheme preserves the functionality of the watermarked DNN without accessing the training dataset. 
The security of the ownership proof against a knowledgable adversary has also been verified. 
Experiments and analyses justify that our watermarking scheme achieves the optimal level of performance against the most powerful adaptive attacks and paves the way for DNN watermarking in industrial applications.

\bibliographystyle{IEEEbib}
\bibliography{WM.bib}

\end{document}